\newif\ifblind
\blindfalse

\newif\ifjasa
\jasafalse

\ifjasa
\fi

\ifjasa
\documentclass[12pt]{article}
\else
\documentclass[11pt]{article}
\renewcommand{\baselinestretch}{1}
\fi
\usepackage{amsmath}
\usepackage{amsfonts}
\usepackage{lscape}
\usepackage{multibib}
\usepackage[round]{natbib}
\usepackage{graphicx}
\graphicspath{{figures/}}
\usepackage{amssymb}
\usepackage{color}
\usepackage{bm}
\usepackage{amsthm}
\usepackage[hidelinks]{hyperref}
\usepackage{enumerate}
\usepackage{changepage}
\usepackage[ruled]{algorithm2e}
\usepackage{nicematrix}
\usepackage{comment}
\usepackage[nodisplayskipstretch]{setspace}

\newcites{body}{References}
\newcites{supp}{References}

\theoremstyle{definition}

\newtheorem{proposition}{Proposition}
\newtheorem{corollary}{Corollary}
\newtheorem{lemma}{Lemma}
\newtheorem{remark}{Remark}

\newtheorem{assumption}{Assumption}
\newtheoremstyle{empty}%
{}{}%
{}{}%
{\bfseries}{.}%
{ }%
{\thmnote{#3}} %

\theoremstyle{empty}

\newcommand{\indist}{\overset{d}{\rightarrow}}
\newcommand{\inprob}{\overset{\prob}{\rightarrow}}
\newcommand{\tabby}{\hspace{10pt}}

\newcommand{\ind}{\overset{\mathrm{ind}}{\sim}}
\newcommand{\iid}{\overset{\mathrm{iid}}{\sim}}

\newcommand{\tp}{\intercal}
\newcommand{\expect}{\mathbb{E}}
\newcommand{\prob}{\mathbb{P}}
\newcommand{\vecz}{\operatorname{vec}}
\newcommand{\diag}{\operatorname{diag}}
\newcommand{\var}{\operatorname{Var}}
\newcommand{\real}{\mathbb{R}}

\newcommand{\liza}[1]{\textcolor{blue}{#1}}

\newcommand\numberthis{\addtocounter{equation}{1}\tag{\theequation}}

\def\twoImages#1#2
{
\centerline{\hfill
\includegraphics[width=.55\textwidth]{#1}
\hfill
\includegraphics[width=.55\textwidth]{#2}
\hfill
}
}

\begin{document}


\ifblind
{
  \bigskip
  \bigskip
  \bigskip
  \begin{center}
    {\LARGE\bf Mesoscale two-sample testing for networks}
\end{center}
  \medskip
}
\else
{
  \title{\bf Mesoscale two-sample testing for networks}
  \author{Peter W. MacDonald \ifjasa\thanks{
    The authors gratefully acknowledge support from NSERC grant RGPIN-2025-02892 to PWM, NSF grant 2052918 to EL, and NSF grant 2210439 to EL and JZ}\hspace{.2cm}\fi\\
    Department of Statistics and Actuarial Science, University of Waterloo\\
    and \\
    Elizaveta Levina \\
    Department of Statistics, University of Michigan\\
    and \\
    Ji Zhu \\
    Department of Statistics, University of Michigan}
  \maketitle
}
\fi

\bigskip
\begin{abstract}
  Networks arise naturally in many scientific fields as a representation of pairwise connections.
  Statistical network analysis has most often considered a single large network, but it is common in a number of applications to observe multiple networks on a shared node set.
  When these networks are grouped by case-control status or another categorical covariate, the classical statistical question of two-sample comparison arises.
  In this work, we address the problem of testing for statistically significant differences in a given arbitrary subset of connections.
  This general framework allows an analyst to focus on a single node, a specific region of interest, or compare whole networks.   Our ability to conduct ``mesoscale'' testing on a meaningful group of edges is particularly relevant for applications such as neuroimaging and distinguishes our approach from prior work, which tends to focus either on a single node or the whole network.
  In this mesoscale setting, we develop statistically sound projection-based tests for two-sample comparison in both weighted and binary edge networks.   The key to our approach is to leverage network information from outside the set of interest to learn informative low-rank projections which leads to more powerful tests.  
\end{abstract}

\section{Introduction} \label{sec:intro}

Modern data structures are increasingly complex, and classical statistical questions need to be answered anew for these complex structures.   One such structure is network data, which arise naturally in many scientific fields as a representation of pairwise connections (edges) 
among a collection of units (nodes).   In this work, we study a classical question, two-sample testing for differences in the means, as it applies to networks.   
This problem is especially relevant in applications such as neuroimaging, where it is common to observe multiple of networks (each patient's brain image corresponds to one network), and compare groups by a case-control status or another categorical covariate.

Seminal work on two-sample testing for samples of networks \citepbody[e.g.,][]{ginestet17hypothesis,xia22hypothesis} has focused on either {\em global} testing, where the null hypothesis is of no difference anywhere in the network, or {\em local} testing, which considers univariate null hypotheses of no difference for each node pair.
Global network testing can often correspond to a scientifically uninteresting null hypothesis.  
For popular network models like the
\ifjasa
\else
stochastic block model \citepbody[SBM,][]{holland83stochastic} or the random dot product graph \citepbody[RDPG,][]{athreya18statistical},
\fi
which fit relatively low-dimensional models to large networks, the global null hypothesis would nearly always be rejected on real data.
On the other hand, local testing creates a massive multiple testing problem and therefore loss of power, especially for smaller samples or when there is no information sharing between node pairs.
Our novel contribution is a flexible framework that spans the spectrum between these two extremes, including them as special cases and allowing for testing at an intermediate or {\em meso}scale, 
so an analyst may choose a scale that is sufficiently local to be scientifically interesting, and yet large enough to have substantial power.
We develop mesoscale network tests for  differences in the means of the edge variables for a fixed {\em hypothesis set}, a subset of node pairs  $\mathcal{S} \subseteq \{1,\ldots,n\}^2$.
If $\mathcal{S} = \{1,\ldots,n\}^2$, then the mesoscale null coincides with the global null hypothesis (or the difference in means for all node pairs), while if $\mathcal{S}$ is a singleton set, the mesoscale hypothesis coincides with a local null hypothesis on the difference in the mean for a single node pair.

Mesoscale hypotheses of this form appear naturally in applications.
For example, in neuroimaging, it is common to focus on anatomical regions of the brain or  on  ``functional networks''  \citepbody{fair09functional}, consisting of connections within or between known groups of nodes, with the groups determined by a brain atlas.
We will refer to these as {\em rectangles}, as up to a reordering of nodes they correspond to a rectangular block of entries of the adjacency matrices. 
\ifjasa
\else
In an exploratory analysis, \citebody{fair09functional} compare samples of child and adult brain images for some of the rectangles induced by functional and lobe-based groupings.
\citebody{sripada20prediction} also work at this scale to identify which functional groupings of brain regions are related to a continuous-valued neurocognitive score.
More generally, scientifically interesting rectangular hypothesis sets can arise from any known groupings of nodes.
While in some cases given node groupings are used as motivation for inference under a block-constant model like the SBM, homogeneous edge distributions within each rectangle are a restrictive assumption we would like to avoid making.
Two-sample comparison at the scale of fixed rows, columns, sub-rows, or sub-columns of the adjacency matrices may also be of interest if we want to isolate the influence of a given node on changes in the two samples \citepbody{kuntal19netshift}.
\fi



For the remainder of the paper, we will consider samples of networks with independent edges (possibly conditional on latent variables).
Much of the network literature  assumes some form of this model, which includes the SBM (conditional on latent communities), the RDPG (conditional on latent node positions), and other latent space models.
\ifjasa
\else
The SBM and RDPG are examples of {\em low-dimensional} latent space network models, which restrict the number of free parameters used to specify the expected adjacency matrix.
Under a low-dimensional network model, although the edges in $\mathcal{S}$ are independent of the edges in its complement (henceforth denoted by $-\mathcal{S}$), these edges still provide information about the mesoscale hypothesis through the underlying model parameters.

\fi
While low-dimensional latent space models allow us to incorporate information from the entire network in a given mesoscale hypothesis test, model misspecification can invalidate test results.
This can occur even under mild model misspecification, for instance incorrect choices of tuning parameters like latent space dimension \citepbody[][Figure 2]{ghoshdastidar18practical},
or latent position similarity function (cf. Section~\ref{subsec:gaussian_sims}); these are typically unknown and often chosen using heuristic methods.

Our new proposal for mesoscale two-sample testing takes advantage of low-dimensional latent structure in a much more robust fashion.
In contrast to previous network model-based approaches, our methods work with projections of the edges in $\mathcal{S}$ and thus remain robust to model and tuning parameter misspecification.
Assuming independent edges, we have a natural split of the data into the edges in $\mathcal{S}$, which will be used for testing the mesoscale hypothesis; and the held out edges in its complement $-\mathcal{S}$, which would otherwise be discarded.
In Section~\ref{subsec:proj_learning}, we will show that low-dimensional network models imply shared structure between these two sets that can make $-\mathcal{S}$ informative for mesoscale testing.
Our projection-based method thus makes the most of the available information to test the mesoscale hypothesis: its size properties are robust to any specification of expected adjacency matrices for which the null is satisfied, and when the edges in $-\mathcal{S}$ are informative for the test, incorporating that information increases power.  


Next, we review some of the literature on global and local hypotheses testing for networks.
Both local and mesoscale network testing rely on the {\em alignment} (also called labeling) of nodes across networks.
Without node alignment, one standard approach to global network testing is to evaluate a network summary statistic, such as modularity or edge density, on the two samples and compare their means either with a $t$-test, or a permutation test \citepbody{wozniak13global}.
\ifjasa
\else
Several papers in statistical network analysis have taken more complicated model-based approaches.
Assuming one sample per group, \citebody{agterberg20nonparametric} and \citebody{chung22valid} develop tests for differences in two networks generated from the RDPG model (unconditional on positions), by testing for differences in the distributions of latent positions, up to unknown orthogonal transformation.
\citebody{athreya21estimation}
instead consider testing for differences in specific functionals of these latent position distributions.
Similarly, \citebody{sabanayagam22graphon} assume each network is sampled from a graphon, and test the equality of those two graphons up to measure preserving transformation.
In all of these papers, by positing network models which are well defined for arbitrary node sets, the global test can be performed in terms of some underlying model parameter, without node alignment.
\fi

In the aligned node setting, many global test statistics are based on some distance between the adjacency matrices of the two samples.
\citebody{ghoshdastidar17twosample1} look at a normalized difference between the two samples of adjacency matrices measured in Frobenius or operator norms.
\ifjasa
\else
For the Frobenius norm, they propose an asymptotically Gaussian test statistic, 
as well as a bootstrap algorithm to estimate an empirical null distribution.
For the operator norm they apply the same bootstrap  algorithm, and in later work \citebody{ghoshdastidar18practical} also provide an asymptotically Tracy-Widom distributed test statistic.
\fi
\citebody{chen20spectralbased} develop an asymptotically Gaussian test statistic for the global null based on the third power of the normalized adjacency matrix difference.
In recent work, \citebody{jin25optimal} use a statistic based on cycles to develop a two-sample hypothesis test that achieves the optimal phase transition under a class of degree-corrected mixed-membership models.
\ifjasa
\else
Other statistics for global network testing are based on graph distances \citepbody{donnat18tracking}. 
\citebody{ginestet17hypothesis} propose an asymptotically $\chi^2$ statistic based on the geodesic distance between Fr\'echet means of graph Laplacians in a suitable manifold.
\citebody{lunagomez21modeling} consider more general modeling and Bayesian inference for unimodal distributions of graphs, in which the probability of a particular graph on $n$ nodes depends on its distance from a fixed graph.
\citebody{lovato20modelfree} develop statistics for two-sample testing based on the pairwise distances among the networks,
comparing average distances within and between the two samples, and use the bootstrap to estimate an empirical null.
\fi

Another branch of the global testing literature focuses on latent space models.
In this case, the global null hypothesis is equivalent to equality of all the latent positions, up to some unidentifiable transformation.
These tests typically assume one network per sample, and that the latent space dimension is known.
\citebody{tang17semiparametric} evaluate the distance between the two sets of estimated latent positions after a Procrustes alignment, and estimate the null distribution with a parametric bootstrap.
\citebody{ghoshdastidar18practical} propose a similar bootstrap-based test which compares estimates of the expected adjacency matrices based on the ASE (Adjacency Spectral Embedding, a standard way of estimating latent positions under RDPG).
\ifjasa
\else
\citebody{li18twosample} test for the equality of all node community memberships in the weighted SBM, and develop an asymptotically Gaussian statistic.
\fi
\citebody{levin17central} and \citebody{draves21biasvariance1} instead analyze an omnibus embedding such that under the global null hypothesis, the estimated embeddings for the two networks are comparable without applying an orthogonal transformation.
\ifjasa
\else
While \citebody{levin17central} estimate the null distribution with a bootstrap method, \citebody{draves21biasvariance1} compare to a $\chi^2$ critical value, and calculate power when the expectations of the two adjacency matrices share a common column space.
\citebody{draves21biasvariance1} also discuss an adaptation of their global testing methodology to test for differences in individual latent positions, as their statistic decomposes across nodes.
However, their theoretical null distribution for a given node's contribution to the statistic only holds under the global null hypothesis.
Other recent work \citepbody{du21hypothesis} has considered 
testing for equality of individual node positions under a latent space model.
Unlike mesoscale testing for an entire row or column of the adjacency matrix, latent position testing creates challenges due to non-identifiability, which we will discuss in more detail in Appendix~\ref{subsubsec:positions}.
In all of these cases, testing is done assuming a low-dimensional network model, which can greatly increase test power, but these tests can be highly sensitive to unknown tuning parameters like the latent space dimension or the number of communities.
\fi

Local network testing can more easily leverage classical methods, as the local hypothesis for a given node pair is a univariate comparison of means.
\citebody{xia22hypothesis} study local network testing, and develop univariate test statistics based on the mean difference between the two samples' edge variables for each node pair.
They provide a power enhancement for their multiple testing procedure which takes advantage of sparsity in the network setting.
Multiple testing methods can be used to combine local tests and address mesoscale hypotheses, but achieving non-trivial power still relies on large signals at individual node pairs, which may be unrealistic for small samples and/or sparse networks.

The rest of this paper is organized as follows. In Section~\ref{sec:testing} we specify a flexible network model with independent edges drawn from an exponential family distribution, and develop projection-based mesoscale hypothesis tests and inference for generalized linear models (GLMs).
The connection to GLMs leads to theoretical guarantees presented in
\ifjasa
\else
Section~\ref{subsec:theory_fixed},
and low rank network modeling leads to results about power of the proposed tests  presented in
Section~\ref{subsec:theory_learned}.
\fi
In Section~\ref{sec:simulations}, we provide empirical evaluations of our tests on synthetic networks, as well as an application to comparing fMRI brain scans of Parkinson's patients to healthy controls.    
Section~\ref{sec:conclusion} concludes with discussion.

\section{Mesoscale testing methodology} \label{sec:testing}

\subsection{Notation}

We write $M_{ij}$ or $[M]_{ij}$ for the $(i,j)$ entry of a $p \times q$ matrix $M$.  We denote submatrices of $M$ by $M_{\mathcal{I}}$ or $[M]_{\mathcal{I}}$
for a subset of entries $\mathcal{I} \subseteq \{1,\ldots,p\} \times \{1,\ldots,q\}$.
If $\mathcal{I}$ forms a rectangle, we treat $M_{\mathcal{I}}$ as a matrix.

The function $\vecz$ vertically stacks columns of a $p \times q$ matrix into a $pq$ vector, with inverse $\vecz^{-1}$.
For a vector $x \in \mathbb{R}^p$, $\operatorname{diag}(x)$ is a $p \times p$ diagonal matrix with entries of $x$ on the diagonal.
We use $\otimes$ to denote the Kronecker product, and $\oplus$ for the matrix direct sum \citepbody{bhatia97matrix}.
The Kronecker product and $\vecz$ satisfy the identity
\begin{equation} \label{kron_identity}
  \vecz\left( XMY^{\tp} \right) = (Y \otimes X) \vecz(M)
\end{equation}
for matrices $X$, $Y$, and $M$ of suitable dimensions.
We use $\lVert \cdot \rVert_2$ for the vector Euclidean norm or matrix operator norm, $\lVert \cdot \rVert_{2 \rightarrow \infty}$ for the two-to-infinity matrix norm (maximum of Euclidean norms of rows) 
$\lVert \cdot \rVert_F$ for the Frobenius norm, and $\lVert \cdot \rVert_{\infty}$ for the maximum absolute entry of a vector.
The zero vector in $\mathbb{R}^p$ is denoted by $\bm{0}_p$, the vector of all ones by $\bm{1}_p$, and $I_p$ is the $p \times p$ identity matrix.
Finally, $\mathcal{O}_p$ denoted the group of orthonormal $p \times p$ matrices, and for $q \leq p$, $\mathcal{O}_{p \times q}$ is the set of $p \times q$ matrices with mutually orthonormal columns.

\subsection{An exponential family edge model}

Throughout the paper, we assume we observe two independent samples of networks on a fixed collection of $n$ aligned nodes.
Edges can be binary or weighted.
The networks (layers) are represented by $n \times n$ adjacency matrices
$
  \{A^{(1)}_k\}_{k=1}^{m_1}, \  \{A^{(2)}_k\}_{k=1}^{m_2}
$
for the two samples of sizes  $m_1$ and  $m_2$ networks.
For simplicity we assume that  $m_1 = m_2 = m$; analogous results hold when sample sizes are of the same asymptotic order.
A mesoscale test considers mean differences in a prespecified set of node pairs $\mathcal{S}$, 
testing the hypothesis 
\begin{equation} \label{E_hypothesis}
  \mathcal{H}_{0,\mathcal{S}}: \mathbb{E}[A^{(1)}_1]_{ij} = \mathbb{E}[A^{(2)}_1]_{ij} \tabby \forall (i,j) \in \mathcal{S}.
\end{equation}
We now propose a flexible independent edge model, with the goal of restating \eqref{E_hypothesis} in terms of model parameters.
Assume all entries in all layers are independent (the generalization to undirected networks which requires enforcing symmetry can be found in Appendix~\ref{subsec:nonrect})),  and that for $k=1,\ldots,M$, and $1 \leq i,j \leq n$, the distribution of $[A^{(g)}_k]_{ij}$ follows an {\em exponential family model} with probability function
\begin{equation} \label{expo_edges}
  \bm{p}(z;\Theta^{(g)}_{ij},\phi_{ij}) = C(z,\phi_{ij}) \cdot \operatorname{exp}\left( \frac{1}{\phi_{ij}^2}\left\{z[\Theta^{(g)}]_{ij} - H\left([\Theta^{(g)}]_{ij}\right)\right\}\right),
\end{equation}
where $H$ is a convex log partition function, and $\phi_{ij} > 0$ is a dispersion parameter.  
In many commonly used models for network data, the entries of $\Theta^{(1)}$ and $\Theta^{(2)}$ are further parameterized through latent variables. 
These could be latent community memberships like in SBM, or latent positions in a continuous space like in RDPG. 
In this work, we treat $\Theta^{(1)}$ and $\Theta^{(2)}$ as fixed parameters. 
The model in \eqref{expo_edges} is a standard exponential dispersion family.  Letting $h = H'$, we can,  for $g \in \{1,2\}$,  express edge expectations via  the corresponding edge expectation parameters $\Theta^{(g)}$ as 
$$
  \mathbb{E}[A^{(g)}_k]_{ij} = h\left([\Theta^{(g)}]_{ij}\right).
$$
We refer to $h$ as the inverse link function, and, in a slight abuse of notation, apply it element-wise to both vector and matrix-valued arguments.
This model is similar to exponential family models for independent matrix entries defined in \citebody{lin21exponentialfamily}, and for independent multiplex network edges in
\ifblind
\citebody{zhang20flexible}, 
\else
\citebody{macdonald22latent}.
\fi

This general framework includes two important special cases.
The Gaussian edge model with constant edge variance corresponds to $H(x) = x^2/2$  for $x \in \mathbb{R}$ and $\phi_{ij}^2 = \sigma^2$ for all $i, j$.
We develop mesoscale hypothesis tests with finite sample guarantees on size for this model, and leave heteroscedastic Gaussian edge models for future work.
The logistic link binary edge model corresponds to $H(x) = \log(1 + e^x)$ and $\phi_{ij} = 1$.
It belongs to the general class of exponential family models with a known mean-variance relationship, and thus no unknown dispersion parameters, for which we develop mesoscale hypothesis tests with asymptotic guarantees on size.
We also generalize this methodology to exponential family edge models with one unknown dispersion parameter ($\phi_{ij} \equiv \phi$).

Some of our additional assumptions are made to simplify notation.  We keep the two sample sizes equal, but this is not a requirement and they are not paired.   We work with directed edges and assume the same distribution for the diagonal entries of the adjacency matrices, but this can be easily adjusted to the undirected case by using only the upper triangular entries of each network, with or without the diagonal (see Appendix~\ref{subsec:nonrect}).

\subsection{Mesoscale null hypotheses and projection-based tests}

We can now state the mesoscale null hypothesis in terms of the parameters of the model \eqref{expo_edges}.   Given a set of entries $\mathcal{S}$,
the mesoscale null hypothesis is
\begin{equation} \label{hypothesis}
  \mathcal{H}_{0,\mathcal{S}}: \Theta^{(1)}_{\mathcal{S}} = \Theta^{(2)}_{\mathcal{S}},
\end{equation}
tested against the general two-sided alternative $\Theta^{(1)}_{\mathcal{S}} \neq \Theta^{(2)}_{\mathcal{S}}$.

To further simplify presentation, we take $\mathcal{S}$  to be an $r \times c$  rectangle, without loss of generality made up of the node pairs corresponding to the first $r$ rows and the last $c$ columns of the adjacency matrix.
\ifjasa
\else
However, \eqref{hypothesis} is stated for an arbitrary (directed) subset of the node pairs. We discuss adjustments to our method needed for non-rectangular hypothesis sets and networks with undirected edges in Appendix~\ref{subsec:nonrect}.
\fi

Our projection-based approach is designed to robustly incorporate information from outside $\mathcal{S}$, with the goal of boosting test power.
We consider projections onto the column spaces of two orthonormal matrix-valued tuning parameters, $U \in \mathcal{O}_{r \times d_r}$ and $V \in \mathcal{O}_{c \times d_c}$ for some $d_r \leq r$ and $d_c \leq c$, chosen adaptively to improve power against alternatives with differences in a particular subspace of the $(r \times c)$-dimensional parameter space.  Intuitively, if we could find projection matrices $U$ and $V$ such that $\Theta_{\mathcal{S}}^{(g)} = U \Gamma_g V^{\tp}$, we could just test the hypothesis $\Gamma_1 = \Gamma_2$, which retains all of the available signal with many fewer degrees of freedom.   In practice, we can look for $U$ and $V$ with low column dimensions which keep the approximation error,
\begin{equation} \label{approx_glm}
  \min_{\Gamma_g \in \mathbb{R}^{d_r \times d_c}} \left\lVert \Theta_{\mathcal{S}}^{(g)} - U \Gamma_g V^{\tp} \right\rVert_F ,
\end{equation}
as small as possible for $g=1,2$, while maximizing power.   
With this intuition, we will develop a two-stage approach to projection-based mesoscale hypothesis testing. 
\begin{description}
  \item[Stage I:] Using node pairs in $-\mathcal{S}$, learn orthonormal matrices $\widehat{U} \in \mathbb{R}^{r \times d_r}$ and $\widehat{V} \in \mathbb{R}^{c \times d_c}$.
  \item[Stage II:] Apply a test function  $\Psi_{\alpha,\widehat{U},\widehat{V}}$ to the node pairs in $\mathcal{S}$.
\end{description}

In order to control type I errors at the nominal level $\alpha$ with adaptively chosen projections,
we take advantage of a natural independent split of the data induced by $\mathcal{S}$, and learn the projections  (Stage I) using only the held-out node pairs in $-\mathcal{S}$.
The edge independence assumption of \eqref{expo_edges} allows us to specify any learning procedure in Stage I 
without affecting the size of the Stage II test.  We do need a well specified class of tests for Stage II in order to achieve the (asymptotic) test size $\alpha$.  
We begin by proposing appropriate classes for Stage II in Section~\ref{subsubsec:expo_edges}, for use under different variants of the exponential family edge model \eqref{expo_edges}; we then return to Stage I and propose appropriate ways to estimate the projection matrices in Section~\ref{subsec:proj_learning}.  

\subsection{Classes of projected tests (Stage II)}\label{subsubsec:expo_edges}

In this section, we develop four classes of tests for different exponential families for use in Stage II.
In all cases we define a test procedure for some arbitrary, fixed projection matrices $U \in \mathbb{R}^{r \times d_r}$ and $V \in \mathbb{R}^{c \times d_c}$ with orthonormal columns, $d_r \leq r$, and $d_c \leq c$, which will be determined (from independent data) in Stage I.


We consider families of tests for the following four settings:  (A) a general exponential family edge model \eqref{expo_edges} with known mean-variance relationship (labeled ``E''); (B) a general exponential family edge model with unknown dispersion (``E-UD''); (C) a modification of (A) using the exact null distribution under the Gaussian edge model with unknown dispersion parameter $\sigma$ (``G'') and (D)
a further modification for the Gaussian edge model which operates on only the projected data (``G-P'').  Theoretical results establishing test sizes are provided in 
Section~\ref{sec:theory}.
 Table~\ref{tab:test_usage} sums up the settings and theoretical guarantees on type I error for the four families of tests, to be descibed in detail in the rest of this section. 
 Propositions 4--6 are stated in Section~\ref{sec:theory}, and Corollary~\ref{cor:wobs_od} is stated in Appendix~\ref{subsubsec:theory_od}.

\begin{table}[h!]
\begin{center}
\begin{tabular}{|l|l|l|}
\hline
& \textbf{Edge model}                       & \textbf{Theoretical guarantees}                                                      \\ \hline
E & Exponential family ($\phi=1$)       & Asymptotic size $\alpha$ (Proposition~\ref{prop:wobs_largen}) \\ 
E-UD & Exponential family (unknown $\phi$) & Asymptotic size $\alpha$ (Corollary~\ref{cor:wobs_od}) \\ 
G & Gaussian ($m=1$) & Size $\alpha$ (Proposition~\ref{prop:fobs_tilde})  \\ 
G-P & Gaussian ($m > 1$) & Size $\alpha$ (Proposition~\ref{prop:fobs})  \\ 
\hline
\end{tabular}
\caption{Suggested setting and theoretical guarantees for the proposed test families. 
\label{tab:test_usage}}
\end{center}
\end{table}

\subsubsection{Exponential family model with known mean-variance relationship} \label{subsubsec:glm_test}
We begin with the general exponential family model \eqref{expo_edges} with dispersion parameters $\phi_{ij} \equiv 1$.
Any $U$ and $V$ correspond to a GLM for node pairs in $\mathcal{S}$.
The link function $h$ is specified in \eqref{expo_edges}, while the design depends on $U$ and $V$.
In particular, applying \eqref{kron_identity}, we can write $\vecz(\Theta_{\mathcal{S}}^{(g)}) = h \{ (V \otimes U)\vecz(\Gamma_g) \}$ for $g = 1,2$. 
Reparameterizing, we write 
\begin{equation*} 
  \gamma_1 = \{ \vecz(\Gamma_1) + \vecz(\Gamma_2)\}/2, \tabby
  \gamma_2 = \{ \vecz(\Gamma_1) - \vecz(\Gamma_2)\}/2,
\end{equation*}
where $\gamma_1$ and $\gamma_2$ are $(d_rd_c)$-dimensional parameter vectors.
A classical test for group differences in this parameteric setting is a Wald test of $\gamma_2 = \bm{0}_{d_rd_c}$, based on asymptotic normality of the maximum likelihood estimator (MLE), and an estimator of its asymptotic covariance matrix.
However, we only use this GLM parameterization to define a quasi-likelihood function of $\gamma_1$ and $\gamma_2$, as we do not actually believe or assume the model holds.
The GLM design will in general be misspecified, but following \citebody{lv14model}, we modify the test so it is robust to this misspecification.
  Specifically, let 
  \begin{align}
    Y &= \begin{pmatrix}
  \vecz([A^{(1)}_1]_{\mathcal{S}}) & \cdots & \vecz([A^{(1)}_{m}]_{\mathcal{S}}) & \vecz([A^{(2)}_1]_{\mathcal{S}}) & \cdots &\vecz([A^{(2)}_{m}]_{\mathcal{S}})
\end{pmatrix} \in \real^{rc \times 2m}, \label{wobs_response} \\
  \bm{X} &= \begin{pmatrix} (V \otimes U) & (V \otimes U) \\
  (V \otimes U) & -(V \otimes U) \end{pmatrix} \otimes \bm{1}_m \in \mathbb{R}^{2rcm \times 2d_rd_c} \ ,  \label{wobs_design}
  \end{align}
  and let $(\hat{\gamma}_1, \hat{\gamma}_2 )^{\tp}$ be the MLE found by fitting a GLM with inverse link function $h$, response vector $\vecz(Y)$, and design matrix $\bm{X}$.
  Define matrices
  \begin{equation}
    \widehat{F} = 2 (V \otimes U)^{\tp} \operatorname{diag}\left\{ h'\left( \widetilde{\Theta}_{\mathcal{S}} \right) \right\} (V \otimes U), \ 
    \widehat{G} = 2 (V \otimes U)^{\tp} \operatorname{diag}\left( h'\left\{ (V \otimes U) \hat{\gamma}_1 \right\}\right) (V \otimes U) \label{wobs_cov} ,
  \end{equation}
  where $\widetilde{\Theta}_{\mathcal{S}}$ is an estimator of $(\Theta_{\mathcal{S}}^{(1)} + \Theta_{\mathcal{S}}^{(2)})/2$.
  In Section~\ref{subsubsec:glm_largen}, we put a precise requirement on the consistency of this estimator (Assumption~\ref{assump:theta_est}) to guarantee that the test controls type I errors at the nominal level.
  Define the test statistic $w^{(\mathrm{E})}$ and test function  $\Psi_{\alpha,U,V}^{(\mathrm{E})}$ by  
  \begin{equation*}
    w^{(\mathrm{E})} = m \hat{\gamma}_2^{\tp} \left( \widehat{G} \widehat{F}^{-1} \widehat{G} \right) \hat{\gamma}_2, \quad
    \Psi_{\alpha,U,V}^{(\mathrm{E})}(Y) = \mathbb{I}\left(w^{(\mathrm{E})} > \chi^2_{d_rd_c,1-\alpha} \right),
  \end{equation*}
  where $\chi^2_{d_rd_c,1-\alpha}$ denotes the $(1-\alpha)$-th quantile of a $\chi^2_{d_rd_c}$ distribution.


\begin{remark}
Note that under classical MLE theory, we have
\begin{equation*}
  \sqrt{m} \widehat{G}^{-1/2} \hat{\gamma}_2 \indist \mathcal{N}(\bm{0},I_{d_rd_c})
\end{equation*}
under the null hypothesis.
Hence the matrix $\widehat{F}$ provides the additional adjustment for model misspecification.
In Section~\ref{subsubsec:glm_largen}, we will show that under some regularity conditions, $\Psi^{(\mathrm{E})}_{\alpha,U,V}$ has asymptotic size $\alpha$.
\end{remark}

\begin{remark}
In \eqref{wobs_cov},
$\widehat{F}$ contains an arbitrary consistent estimate of the pooled edge expectation parameters in the hypothesis set.
For experiments and real data analysis in Section~\ref{sec:simulations} using the Gaussian edge model, we suggest the pooled means for each node pair.
For the logistic link binary edge model, we instead suggest a regularized Bayes estimator
\begin{equation} \label{theta_bayes}
    \widetilde{\Theta}_{\mathcal{S}} = \frac{1}{2m + 2} \sum_{k=1}^m \left( [A_k^{(1)}]_{\mathcal{S}} + [A_k^{(1)}]_{\mathcal{S}} + 1 \right) 
\end{equation}
to avoid boundary issues evaluating the inverse logistic link at $0$ or $1$.
\end{remark} 

\subsubsection{Exponential family edge model with unknown dispersion} \label{subsubsec:glm_test_od}

We now consider a general exponential family edges with a single unknown dispersion parameter, $\phi_{ij} \equiv \phi$.
Building off of classical quasilikelihood results for GLMs, we propose a rescaled statistic for testing $\mathcal{H}_{0,\mathcal{S}}$ based on an estimator of the dispersion parameter.
In our setting and notation, the classical consistent estimator of $\phi^2$ \citepbody{mccullagh83generalized} is given by
\begin{equation}\label{phi_hat1}
  \hat{\phi}^2 = \frac{1}{2rc(m-1)} \sum_{g=1}^2 \sum_{k=1}^m \sum_{(ij) \in \mathcal{S}} \frac{([A^{(g)}_k]_{ij} - h([\widetilde{\Theta}^{(g)}]_{ij}))^2}{h'([\widetilde{\Theta}^{(g)}]_{ij})}, 
\end{equation}
where for $(i,j) \in \mathcal{S}$  and  $g = 1, 2$,  $[\widetilde{\Theta}^{(g)}]_{ij}$ is a consistent estimator of $[\Theta^{(g)}]_{ij}$.
  Define 
  $w^{(\mathrm{E})}$ as in Section~\ref{subsubsec:glm_test}.
Then we can define the test statistic and test function for general exponential family edges with unknown dispersion by
\begin{equation*}
  w^{(\mathrm{E-UD})}  = w^{(\mathrm{E})} / \hat{\phi}^2, \quad
  \Psi_{\alpha,U,V}^{(\mathrm{E-UD})}(Y)  = \mathbb{I}\left(w^{(\mathrm{E-UD})} > \chi^2_{d_rd_c,1-\alpha} \right) . 
\end{equation*}



\subsubsection{Gaussian edge model} \label{subsubsec:gaussian_test1}

If we assume that the edges are Gaussian, we have $H(x) = x^2/2$, $\phi_{ij} \equiv \sigma$, and the GLM reduces to a standard linear model with Gaussian errors.
In this special case, we define a modified class of test statistics with rejection thresholds based on the quantiles of its exact $F$-distribution. Let
\begin{equation}
    \bar{Y}^{(\mathrm{diff})} = \frac{1}{m} \sum_{k=1}^m \vecz([A^{(1)}_k]_{\mathcal{S}}) - \frac{1}{m} \sum_{k=1}^m \vecz([A^{(2)}_k]_{\mathcal{S}}) \in \real^{rc},
\end{equation}
the vector of edgewise mean differences.
  Let $\nu_1 = d_rd_c$ and $\nu_2= rc - d_rd_c$, and let $Q^{\perp}(U,V)$ denote
  the projection matrix onto the orthogonal complement of $\operatorname{col}\{ (V \otimes U) \} \subseteq \mathbb{R}^{rc}$.
  Following \eqref{approx_glm}, a well chosen projection $(V \otimes U)$ should capture most of the structure in $\vecz(\Theta^{(1)}_{\mathcal{S}} - \Theta^{(2)}_{\mathcal{S}})$ and thus satisfy $Q^{\perp}(U,V)\vecz(\Theta^{(1)}_{\mathcal{S}} - \Theta^{(2)}_{\mathcal{S}}) \approx \bm{0}_{rc}$,
    whether or not the null hypothesis holds.
    Thus, the sum of squares of $\bar{Y}^{(\mathrm{diff})}$, after proper scaling and projection onto $Q^{\perp}(U,V)$, provides an estimate of the edge variance $\sigma^2$ which is independent of $(V \otimes U) \bar{Y}^{(\mathrm{diff})}$ and not overly conservative under the alternative hypothesis.
 We thus define the test statistic and test function as 
  \begin{equation*}
    w^{(\mathrm{G})} = \frac{\nu_2 \lVert (V \otimes U)^{\tp} \bar{Y}^{(\mathrm{diff})} \rVert_2^2}{\nu_1 \lVert Q_{\perp}(U,V) \bar{Y}^{(\mathrm{diff})} \rVert_2^2} \ , \quad
    \Psi_{\alpha,U,V}^{(\mathrm{G})}(Y) = \mathbb{I}\left(w^{(\mathrm{G})} > F_{\nu_1,\nu_2,1-\alpha} \right) \, ,
  \end{equation*}
  where $F_{\nu_1,\nu_2,1-\alpha}$ denotes the $1-\alpha$ quantile of a $F_{\nu_1,\nu_2}$ distribution.
In Section~\ref{subsec:theory_gaussian}, we show that each $\Psi^{(\mathrm{G})}_{\alpha,U,V}$ has size $\alpha$ under the Gaussian edge model.

\subsubsection{Gaussian edge model with projected data} \label{subsubsec:gaussian_test2}

The tests $\Psi^{(\mathrm{G})}_{\alpha,U,V}$ are for the Gaussian edge model with constant variance, which assumes the edge covariance matrix is a multiple of the identity.
In particular, it requires that the sum of squares of $Q_{\perp}(U,V) \bar{Y}^{(\mathrm{diff})}$ is a good estimator of the edge variance.   However, application of our methods to real data show that empirically, if $U$ and $V$ are well-chosen to preserve the structure in $\Theta^{(1)}_{\mathcal{S}} - \Theta^{(2)}_{\mathcal{S}}$, then the variability of the projection of the edges onto $\operatorname{col}(V \otimes U)$ tends to be higher than variability of the projection onto its orthogonal complement.
This means that the estimator of $\sigma^2$ based on $Q_{\perp}(U,V) \bar{Y}^{(\mathrm{diff})}$ tends to be too small, leading to inflated rejection rates.
Thus, for the Gaussian edge model with $m > 1$, we propose a modified class of tests which takes advantage of replication and uses only the projected data to get an unbiased estimate of $\sigma^2$ under a milder assumption on the covariance structure of the edge variables (see Remark~\ref{rem:restricted_cov}).

  Let $\nu_1 = d_rd_c$, $\nu_2' = d_rd_c(2m - 2)$, and $Y^{(\mathrm{resid})}$ the $rc \times 2m$ matrix with columns
\begin{equation}
    \vecz([A^{(g)}_k]_{\mathcal{S}}) - \frac{1}{m} \sum_{k=1}^m \vecz([A^{(g)}_{k'}]_{\mathcal{S}}) \in \real^{rc},
\end{equation}
for $k=1,\ldots,m$ and $g=1,2$, collecting the node pairs in $\mathcal{S}$ after group-wise centering.
Then we define the test statistic and test function by
\begin{equation*}
    w^{(\mathrm{G-P})} = \frac{m \nu'_2 \lVert (V \otimes U)^{\tp} \bar{Y}^{(\mathrm{diff})} \rVert_2^2}{2 \nu_1 \lVert (V \otimes U)^{\tp} Y^{(\mathrm{resid})} \rVert_F^2}, \quad \Psi_{\alpha,U,V}^{(\mathrm{G-P})}(Y) = \mathbb{I}\left(w^{(\mathrm{G-P})} > F_{\nu_1,\nu_2',1-\alpha} \right) \, . 
  \end{equation*}



\subsection{Learning projections from held-out edges} \label{subsec:proj_learning}

Having described tests to use in Stage II, we return to learning appropriate projections $U$ and $V$ from the independent held out edges in $-\mathcal{S}$ (Stage I).
We will show in Section~\ref{sec:theory} that under the Gaussian edge model and for any fixed alternative, maximizing power is equivalent to maximizing the following non-centrality parameter: 
\begin{equation} \label{ncp}
  \psi(U,V) = \frac{m}{2\sigma^2} \left\lVert U^{\tp} ( \Theta_{\mathcal{S}}^{(1)} - \Theta_{\mathcal{S}}^{(2)}) V \right\rVert_F^2 \leq \frac{m}{2\sigma^2} \left\lVert \Theta_{\mathcal{S}}^{(1)} - \Theta_{\mathcal{S}}^{(2)} \right\rVert_F^2. 
\end{equation}
This implies that the optimal $d$-dimensional projections 
are orthonormal bases for the leading left and right singular subspaces of $\Theta_{\mathcal{S}}^{(1)} - \Theta_{\mathcal{S}}^{(2)}$. 
While this is exactly true only under the Gaussian edge model, we find that in practice these provide a good approximation in other settings as well, and refer to these singular subspaces as the {\em oracle} left and right projections, regardless of the underlying model.

One sensible approach to estimate the oracle projections is to leverage existing approaches for matrix completion to estimate
$\Theta_{\mathcal{S}}^{(1)} - \Theta_{\mathcal{S}}^{(2)}$ \citepbody{mazumder10spectral,lin21exponentialfamily},
and then compute the leading singular subspaces of the estimate.
\ifjasa
\else
Iterative approaches to low rank matrix completion have been developed that are well suited to this problem, either by completing
  $M^{-1} \sum_{k=1}^M [ A_k ]_{-\mathcal{S}}$
for $g=1,2$, and then taking a difference; or directly completing the difference
$m^{-1} \sum_{k=1}^m \left\{ [ A_k^{(1)} ]_{-\mathcal{S}} - [ A_k^{(2)} ]_{-\mathcal{S}} \right\}.$
\citebody{mazumder10spectral}, for instance, develop a low rank matrix completion approach for the Gaussian setting by using alternating least squares.
This approach has been extended to matrices with general exponential family entries in \citebody{lin21exponentialfamily} (eSVD), which can be adapted to the case with missing entries.
\citebody{mazumder10spectral} and \citebody{lafond15low} also consider low rank matrix completion based on solving a convex optimization problem with nuclear norm penalty, in the least squares and general exponential family settings respectively.
These existing matrix completion algorithms are generally developed for unstructured or random missingness patterns.
\fi
These are especially useful for estimating projections with non-rectangular hypothesis node pairs, since they provide estimates of the singular subspaces of the entire expected adjacency matrix, not just for the hypothesis set.

When the set $\mathcal{S}$ is rectangular, we can take advantage of this structure to characterize the oracle projections in terms of the node pairs in $-\mathcal{S}$.
Partition the indices of the adjacency matrices into four contiguous blocks
\begin{equation*}
  \begin{pNiceMatrix}[c]
  & \Block[draw]{1-1}{\mathcal{C}} & \Block[draw,fill=lightgray]{1-1}{\mathcal{S}} & \\
  & \Block[draw]{1-1}{\mathcal{D}} & \Block[draw]{1-1}{\mathcal{R}} &\\
\end{pNiceMatrix},
\end{equation*}
where $\mathcal{S}$ is the $r \times c$ hypothesis set.  
We also write $[r]$ for the first $r$ nodes and  $-[r]$ for its complement, as well as $[c]$ for the last $c$ nodes and $-[c]$ for its complement.  

Suppose $\Theta^{(1)} - \Theta^{(2)}$ has rank $d_*$ and its singular value decomposition (SVD) is given by   $U_{\mathrm{diff}}S_{\mathrm{diff}}V_{\mathrm{diff}}^{\tp}$.  Partitioning this matrix into the same four blocks $\mathcal{C, S, D, R}$, we have 
\begin{equation} \label{block_svd}
  U_{\mathrm{diff}}S_{\mathrm{diff}}V_{\mathrm{diff}}^{\tp} = \begin{pmatrix}
    U_{[r]}S_{\mathrm{diff}}V_{-[c]}^{\tp} & U_{[r]}S_{\mathrm{diff}}V_{[c]}^{\tp} \\
    U_{-[r]}S_{\mathrm{diff}}V_{-[c]}^{\tp} & U_{-[r]}S_{\mathrm{diff}}V_{[c]}^{\tp}
  \end{pmatrix}.
\end{equation}
Note that the blocks of \eqref{block_svd} are not the SVD's of the individual blocks, as first and third factors of each are in general not orthonormal.
Moreover, the rank  $d_{\mathcal{S}}$ of the block corresponding to $\mathcal{S}$ may be smaller than $d_*$.
By \eqref{ncp}, the non-centrality parameter is maximized by projections of dimension $d = \min\{d_{\mathcal{S}},r,c\}$, with additional projection dimensions only increasing the test degrees of freedom.

Define
\begin{equation} \label{T_matrix}
  \mathbb{T} = (\Theta^{(1)}_{\mathcal{C}} - \Theta^{(2)}_{\mathcal{C}})(\Theta^{(1)}_{\mathcal{D}} - \Theta^{(2)}_{\mathcal{D}})^{\dagger}(\Theta^{(1)}_{\mathcal{R}} - \Theta^{(2)}_{\mathcal{R}}) \in \mathbb{R}^{r \times c},
\end{equation}
where $\dagger$ denotes the Moore-Penrose pseudoinverse.
The following two propositions show how and when we can relate the spectral structure of $\mathbb{T}$ to the oracle projections. 
\begin{proposition} \label{prop:onestep_weak}
  Suppose $U_{-[r]}$ and $V_{-[c]}$ have linearly independent columns.
  Then $\mathbb{T}$ has rank $d_{\mathcal{S}}$, and the column and row spaces of $\mathbb{T}$ coincide with the column and row spaces of $\Theta^{(1)}_{\mathcal{S}} - \Theta^{(2)}_{\mathcal{S}}$.
\end{proposition}

\begin{proposition} \label{prop:onestep_strong}
  Suppose
  $  \left( U_{-[r]} \right)^{\tp} U_{-[r]} = \rho_U I_{d_*}$, $\left( V_{-[c]} \right)^{\tp} V_{-[c]} = \rho_V I_{d_*}$
  for strictly positive $\rho_U$ and $\rho_V$.
  Then $\mathbb{T}$ has rank $d_{\mathcal{S}}$, and for any $1 \leq d \leq d_{\mathcal{S}}$ the leading $d$-dimensional left and right singular subspaces of $\mathbb{T}$ coincide with the leading $d$-dimensional left and right singular subspaces of $\Theta^{(1)}_{\mathcal{S}} - \Theta^{(2)}_{\mathcal{S}}$.
\end{proposition}

Proposition~\ref{prop:onestep_weak} puts a weaker condition on the singular vectors, such that $\mathbb{T}$ has the correct column and row spaces, but may order the singular values incorrectly.
In contrast, Proposition~\ref{prop:onestep_strong} ensures the singular vectors of $\mathbb{T}$ are correctly ordered.
When this stronger condition holds, Proposition~\ref{prop:onestep_strong} justifies the optimality of the singular subspaces of $\mathbb{T}$ in the case where $\Theta^{(1)}_{\mathcal{S}} - \Theta^{(2)}_{\mathcal{S}}$ is not exactly low rank, but still has low effective rank.
\ifjasa
\else

Thus, we propose the singular subspaces of an estimate $\widehat{\mathbb{T}}$ as a plug-in for the singular subspaces of $\Theta^{(1)}_{\mathcal{S}} - \Theta^{(2)}_{\mathcal{S}}$, where $\widehat{\mathbb{T}}$ itself plugs in estimates for each of the block edge expectation parameters for each group.
In the Gaussian case, we naturally use
\begin{equation*}
  \widehat{\Theta}^{(g)}_{\mathcal{I}} = \frac{1}{m} \sum_{k=1}^m [A_k^{(g)}]_{\mathcal{I}}
\end{equation*}
for $g=1,2$ and $\mathcal{I} \in \{\mathcal{C}, \mathcal{R}, \mathcal{D}\}$, plus an additional rank truncation step for block $\mathcal{D}$ to control the error after applying the pseudoinverse.
In the case with no self loops, rank truncations may also be used to fill in the missing diagonal entries.
With general exponential family edges and a well-behaved link function, we could apply the link function to the group mean adjacency matrices, or perform an initial low rank denoising before calculating $\widehat{\mathbb{T}}$.
We analyze the theoretical performance of this estimator for the Gaussian edge model in Section~\ref{subsec:theory_learned}.
\fi

\ifjasa
\else
\ifjasa
\subsection{Row and column hypothesis sets}
\else 
\subsubsection{Row and column hypothesis sets}
\fi
One special case of interest is when $\mathcal{S}$ is a row (without loss of generality the first row) of the adjacency matrix ($r=1$ and $c=n$), and the hypothesis is about  edges originating from a particular node.
Rewriting \eqref{block_svd} with just two blocks, $\mathcal{S}$ and $\mathcal{R}$, we have 
\begin{equation*}
  \Theta^{(1)} - \Theta^{(2)} = \begin{pmatrix}
    \Theta^{(1)}_{\mathcal{S}} - \Theta^{(2)}_{\mathcal{S}} \\
    \Theta^{(1)}_{\mathcal{R}} - \Theta^{(2)}_{\mathcal{R}}
    \end{pmatrix} = U_{\mathrm{diff}}S_{\mathrm{diff}}V_{\mathrm{diff}}^{\tp} = \begin{pmatrix}
    U_{[1]}S_{\mathrm{diff}}V_{\mathrm{diff}}^{\tp} \\
    U_{-[1]}S_{\mathrm{diff}}V_{\mathrm{diff}}^{\tp}
  \end{pmatrix}.
\end{equation*}
In this case, since $r=1$, the non-centrality parameter can be maximized by a trivial left projection, and a one-dimensional right projection, onto the direction of $\Theta^{(1)}_{\mathcal{S}} - \Theta^{(2)}_{\mathcal{S}}$.
However, learning this direction requires an estimate of $U_{[1]}^{\tp} \in \real^{d_*}$, which is not possible using only the held-out edges in $\mathcal{R}$.

On the other hand, $\Theta^{(1)}_{\mathcal{S}} - \Theta^{(2)}_{\mathcal{S}}$ will share a row space with $\Theta^{(1)}_{\mathcal{R}} - \Theta^{(2)}_{\mathcal{R}}$, a subspace which can be recovered from the held out edges.
While we cannot get a guarantee as strong as Proposition~\ref{prop:onestep_strong}, if $U_{-[r]}$ has linearly independent columns, then all the signal is captured in an estimable $d_*$-dimensional right projection.
That is, the projection test with left projection $1$ and right projection
$\operatorname{row}\left( \Theta^{(1)}_{\mathcal{R}} - \Theta^{(2)}_{\mathcal{R}}\right)$ will have maximum non-centrality parameter
$$
  \frac{m}{2\sigma^2} \lVert \Theta^{(1)}_{\mathcal{S}} - \Theta^{(2)}_{\mathcal{S}} \rVert_2^2.
$$
The case  $r=n$, $c=1$ is analogous and corresponds to testing for differences in the edges which terminate at a particular node.
\fi

\section{Theoretical guarantees} \label{sec:theory}

This section presents theoretical properties of our projection testing approaches, for fixed projections.
For the general exponential family edge model with known mean-variance relationship, Propositions~\ref{prop:wobs} and \ref{prop:wobs_largen} give the asymptotic distribution of $w^{(\mathrm{E})}$ under a sequence of local alternatives, and the asymptotic size of the tests $\Psi_{\alpha,U,V}^{(\mathrm{E})}$.  
For the Gaussian edge model, Propositions~\ref{prop:fobs_tilde} and \ref{prop:fobs} establish the exact distributions of $w^{(\mathrm{G})}$ (see Section~\ref{subsubsec:gaussian_test1}) and $w^{(\mathrm{G-P})}$ (see Section~\ref{subsubsec:gaussian_test2}) under any fixed alternative, and therefore the exact size of the tests $\Psi_{\alpha,U,V}^{(\mathrm{G})}$ and $\Psi_{\alpha,U,V}^{(\mathrm{G-P})}$.  
These results provide intuition about test power and oracle projections. 
\ifjasa
\else
A result on the asymptotic distribution of $w^{(\mathrm{E-UD})}$ is stated in  Section~\ref{subsubsec:theory_od}.  In Section~\ref{subsec:theory_learned},  we expand on the properties of the projection test with learned projections based on a plug-in estimate of the matrix $\mathbb{T}$, defined in Section~\ref{subsec:proj_learning}.
Under the Gaussian edge model, we prove a non-asymptotic lower bound on power for rectangular hypothesis sets, which depends on the spectral properties of the edge expectation parameter matrices.
\fi


\ifjasa
\else
\subsection{Test properties with fixed projections} \label{subsec:theory_fixed}

\subsubsection{The general exponential family edge model} \label{subsubsec:glm_largen}
\fi
Under the general exponential family edge model with a known mean-variance relationship, we will establish the asymptotic results for the distribution for our test statistic $w^{(\mathrm{E})}$.
For simplicity, we suppose that the hypothesis set is rectangular with $r$ rows and $c$ columns, and that we have two balanced groups ($m_1=m_2=m$).
Consider an asymptotic regime $n \rightarrow \infty$, where $n$ is the number of nodes in each sampled network, and allow $m = m(n)$, $r = r(n)$, $c = c(n)$ to diverge with $n$.
We consider a sequence of local alternatives which scale with $m$ and provide intuition about the optimal projection matrices for $\Psi_{\alpha,U,V}^{(\mathrm{E})}$.
We write the edge expectation parameters $\Theta^{(1)}_{\mathcal{S},n}$ and $\Theta^{(2)}_{\mathcal{S},n}$ as
\begin{equation} \label{loc_alt}
  \Theta^{(1)}_{\mathcal{S},n} = \operatorname{vec}^{-1}(\tau_n) + m^{-1/2} \operatorname{vec}^{-1}(\Delta_n), \quad \Theta^{(2)}_{\mathcal{S},n} = \operatorname{vec}^{-1}(\tau_n) - m^{-1/2} \operatorname{vec}^{-1}(\Delta_n)
\end{equation}
for $(rc)$-dimensional vectors $\tau_n$ and $\Delta_n$.
For each $n$, suppose that we apply the mesoscale test with arbitrary, deterministic sequences of orthonormal matrices $\{U_n\}$ and $\{V_n\}$, of dimensions $r \times d_r$ and $c \times d_c$.
We assume the following asymptotic {\em coherence} condition on the orthonormal matrices $(V_n \otimes U_n)$, which have unit norm columns, but growing row dimension $rc$ as $n$ grows \citepbody{cape19twotoinfinity}.
\begin{assumption} \label{assump:coherence}
  ~ 
  \begin{equation} \label{coherence_cond}
    \lVert (V_n \otimes U_n) \rVert_{2 \rightarrow \infty} = O\left((rc)^{-1/2}\right).
  \end{equation}
\end{assumption}
For the remainder of this section, let
\begin{equation*}
  \mu_n = \operatorname{vec}\begin{pmatrix}
  \Theta_{\mathcal{S},n}^{(1)} \\ \Theta_{\mathcal{S},n}^{(2)}
  \end{pmatrix} \in \mathbb{R}^{2rc}, \quad
  \bar{X}_n = \begin{pmatrix}
  V_n \otimes U_n & V_n \otimes U_n \\
  V_n \otimes U_n & -(V_n \otimes U_n)
  \end{pmatrix} \in \mathbb{R}^{2rc \times 2d_rd_c},
\end{equation*}
and let $\gamma_n = (\gamma_{1,n}, \gamma_{2,n})^{\tp}$ be the vector that solves the system of equations
\begin{equation} \label{population_glm}
\bar{X}_n^{\tp} \left\{ h(\mu_n) - h\left(\bar{X}_n \gamma_n \right) \right\} =   \bm{0}_{2d_rd_c}. 
\end{equation}
Throughout, we require the following regularity conditions for the edge distribution and inverse link function, similar to \citetbody{lv14model}.
\begin{assumption} \label{assump:glm}
  ~
  \begin{enumerate}[{(A)}]
    \item The function $h$ is strictly increasing and continuously differentiable.
    \item For all $g=1,2$, $k=1,\ldots,m$, $(i,j) \in \mathcal{S}$,
    $\mathbb{E} \left\lvert [A^{(g)}_k]_{ij} - h([\Theta^{(g)}_{\mathcal{S},n}]_{ij}) \right\rvert^3 $
    is bounded uniformly over $n$.
  \end{enumerate}
\end{assumption}
Under Assumption~\ref{assump:glm}, \eqref{population_glm} has a unique solution \citepbody{lv14model}.
We also assume that there exist consistent estimators for the unknown edge expectation parameters, in the following sense.
\begin{assumption} \label{assump:theta_est}
  Suppose there exists a sequence of estimators $\widetilde{\Theta}_{n}$ of $\tau_n$ such that
  \begin{equation} \label{thetahat_cond}
    \frac{1}{rc} \sum_{(ij) \in \mathcal{S}} \left\lvert h'([\widetilde{\Theta}_n]_{ij}) - h'([\vecz^{-1}(\tau_n)]_{ij}) \right\rvert = o_{\prob}(1)
  \end{equation}
\end{assumption}
Assumption~\ref{assump:theta_est} does not require uniform consistency, and thus is reasonable as long as $m \rightarrow \infty$.

In the following Proposition~\ref{prop:wobs}, we allow for local alternatives, but fix the size and shape of the hypothesis set, the edge expectation parameters, and the projections used for testing.
For simplicity of notation, we omit the subscripts from $U$, $V$, $\tau$, and $\Delta$.
These restrictions mean that Assumption~\ref{assump:coherence} holds trivially, but we require one additional assumption on $\gamma_n$.
\begin{assumption} \label{assump:gamma_size}
    Let $\gamma_n$ be the solution of \eqref{population_glm}.  We assume that
    \begin{equation*} 
  \sup_{n \geq 1} \left\{ \lVert (\gamma_{1,n} , \sqrt{m} \gamma_{2,n} ) \rVert_2 \right\} = K_{\gamma} < \infty.
\end{equation*}
\end{assumption}
From \eqref{population_glm}, $\gamma_n = (\gamma_{1,n},\gamma_{2,n})^{\tp}$ is the coefficient vector for the best model in the misspecified GLM family (with design matrix $\bm{X}$) to approximate the distribution of $Y$ \citepbody{lv14model}. It is reasonable to assume that these coefficients are uniformly bounded, and the norm of $\gamma_{2,n}$ scales like the norm of the mean difference, which has rate $m^{-1/2}$ under the sequence of local alternatives \eqref{loc_alt}.
We are now ready to state the first main proposition.
\begin{proposition} \label{prop:wobs}
  Assume the exponential family edge model \eqref{expo_edges} with dispersion $\phi_{ij} \equiv 1$ satisfies Assumptions~\ref{assump:glm}--\ref{assump:gamma_size}.
  Assume the maximum likelihood estimate from fitting a GLM with responses $\vecz(Y)$, and design matrix $\bm{X}$ (as in Section~\ref{subsubsec:glm_test}) exists.
  Define the function
  \begin{equation*}
    Q(x,y) = 2\begin{pmatrix}
    (V \otimes U)^{\tp} \left\{ h(\tau) - h((V \otimes U)x) \right\} \\
    (V \otimes U)^{\tp}\left\{ \operatorname{diag}\{h'(\tau)\} \Delta - \operatorname{diag}\{h'((V \otimes U) x)\}(V \otimes U)y \right\}
    \end{pmatrix}.
  \end{equation*}
  which has a (unique) root defined by $Q(\tilde{\tau},\tilde{\Delta}) = \bm{0}_{2d_rd_c}$.
  Finally, define
  \begin{equation*}
    \tilde{F} = 2 (V \otimes U)^{\tp} \operatorname{diag}\{h'(\tau)\} (V \otimes U), \quad
    \tilde{G} = 2 (V \otimes U)^{\tp} \operatorname{diag}(h'\{(V \otimes U)\tilde{\tau}\}) (V \otimes U).
  \end{equation*}
  Then $w^{(\mathrm{E})}$ has an asymptotic $\chi^2_{d_rd_c}(\tilde{\psi})$ distribution as $n \rightarrow \infty$, with non-centrality parameter
  \begin{equation} \label{glm_ncp}
    \tilde{\psi} = \tilde{\Delta}^{\tp} \tilde{G} \tilde{F}^{-1} \tilde{G} \tilde{\Delta}.
  \end{equation}
\end{proposition}
The proof of Proposition~\ref{prop:wobs} can be found in Appendix~\ref{sec:app:proof2}; it establishes the asymptotic size of the test $\Psi_{\alpha,U,V}^{(\mathrm{E})}$, and motivates an optimal choice of projections $U$ and $V$ to maximize the asymptotic non-centrality parameter.
Given $\tilde{\tau}$, $\tilde{\Delta}$ solves a generalized least squares problem and satisfies
\begin{equation*}
  \tilde{\Delta} = \frac{\sqrt{m}}{2} \left\{ (V \otimes U)^{\tp} \operatorname{diag}\{ h'((V \otimes U)\tilde{\tau}) \} (V \otimes U) \right\}^{-1} (V \otimes U)^{\tp} \operatorname{diag}\{ h'(\tau) \} \operatorname{vec}( \Theta_{\mathcal{S},n}^{(1)} - \Theta_{\mathcal{S},n}^{(2)}).
\end{equation*}
Thus, after some standardization, the asymptotic non-centrality parameter \eqref{glm_ncp} is proportional to a squared norm of the projection of $\vecz( \Theta_{\mathcal{S},n}^{(1)} - \Theta_{\mathcal{S},n}^{(2)})$ onto $\operatorname{col}(V \otimes U)$. 

In the following Proposition~\ref{prop:wobs_largen}, we work in the general asymptotic regime where the size and shape of the hypothesis set are allowed to grow, but we no longer allow for a sequence of local alternatives, and instead work with the null hypothesis  $\Delta_n \equiv 0$ for all $n$. 
We now include the subscript $n$ for $U$, $V$, and $\tau$.
\begin{proposition} \label{prop:wobs_largen}
  Assume the exponential family edge model \eqref{expo_edges} with dispersion $\phi_{ij} \equiv 1$ satisfies Assumptions~\ref{assump:glm} and \ref{assump:theta_est}, and $\{U_n\}$ and $\{V_n\}$ satisfy Assumption~\ref{assump:coherence}.
  Define $\gamma_n$ as the solution of \eqref{population_glm}, and assume there exist constants $K^-$ and $K^+$ such that
  \begin{equation} \label{hprime_cond}
      0 < K^- \leq h'(\bm{w}_n) \leq K^+ < \infty
  \end{equation}
  elementwise and uniformly over $n$, when $\bm{w}_n$ is replaced by either $\tau_n$ or $\bar{X}_n\gamma_{n}$.
  Assume the maximum likelihood estimate from fitting a GLM with responses $\vecz(Y)$, and design matrix $(\bar{X}_n~\otimes~\bm{1}_m)$ exists.
  Then as long as  $rcm \rightarrow \infty$,
  $w^{(\mathrm{E})}$ has an asymptotic $\chi^2_{d_rd_c}$ distribution as $n \rightarrow \infty$.
\end{proposition}

The proof of Proposition~\ref{prop:wobs_largen} can be found in Appendix~\ref{sec:app:proof2}; it establishes the asymptotic size of the test $\Psi_{\alpha,U,V}^{(\mathrm{E})}$.
Condition \eqref{hprime_cond} will always hold under the Gaussian edge model since $h'(x) = 1$ for all $x$. The upper bound will always hold under the logistic link binary edge model model, since $h'(x)$ is uniformly bounded.
For other models (e.g., Poisson or gamma) and for the lower bound under the logistic link binary edge model, a sufficient condition for \eqref{hprime_cond} is a uniform (over $n$) bound on $\lVert \mu_n \rVert_{\infty}$ and $\lVert \bar{X}_n\gamma_{n} \rVert_{\infty}$.
It is reasonable to expect a uniform bound on $\lVert \bar{X}_n\gamma_{n} \rVert_{\infty}$ would hold given a uniform bound on $ \lVert \mu_n \rVert_{\infty}$, since by \eqref{population_glm}, $\bar{X}_n\gamma_{n}$ is analogous to a projection of $\mu_n$.

\begin{remark} \label{rem:sparse_logit}
  Under the logistic link binary model, the lower bound in \eqref{hprime_cond} can be relaxed to allow sparsity, where the edge probabilities go to 0 as $n$ increases.
  When $h^{-1}$ is the logistic link,
  $
  h'(\theta) = h(\theta)\{1 - h(\theta)\} \sim h(\theta)
$
  as $\theta \rightarrow -\infty$. Suppose that $\tau_n$ satisfies $h(\tau_n) \sim a_n$ entrywise, for some $a_n \rightarrow 0$.
  If the edge probabilities based on the population GLM coefficients also satisfy $h(\bar{X}_n\gamma_{n}) \sim a_n$ and $a_n \gg (rcm)^{-1/3}$, then the conclusions of Proposition~\ref{prop:wobs_largen} will still hold.
\end{remark}

\ifjasa
\else
\ifjasa
\subsection{The general exponential family edge model with overdispersion}
\else
\subsubsection{The general exponential family edge model with overdispersion}
\fi
\label{subsubsec:theory_od}

In this section, we state asymptotic results which characterize the behavior of $w^{(\mathrm{E-UD})}$ under the mesoscale null hypothesis, in an asymptotic regime similar to Section~\ref{subsubsec:glm_largen}.
As in Section~\ref{subsubsec:glm_largen}, assume the hypothesis set is rectangular with $r$ rows and $c$ columns, mesoscale null hypothesis $\Theta^{(1)}_{\mathcal{S},n} = \Theta^{(2)}_{\mathcal{S},n}$ holds for all $n$, and denote the common matrix of edge expectation parameters by $\Theta_n$.
First, we state a stronger version of  Assumption~\ref{assump:glm} which is sufficient for consistency of $\hat{\phi}^2$, as defined in Section~\ref{subsubsec:glm_test_od}, and justifies that the test defined in Section~\ref{subsubsec:glm_test_od} has asymptotic size $\alpha$.

\begin{assumption} \label{assump:glm_od}
  \hfill
  \begin{enumerate}[{(A)}]
    \item $h$ is strictly increasing and twice continuously differentiable.  
    \item For all $g=1,2$, $k=1,\ldots,m$, $(i,j) \in \mathcal{S}$, and $n$,
    $\mathbb{E} \left\lvert [A^{(g)}_k]_{ij} - h([\Theta_n]_{ij}) \right\rvert^4 $
    is uniformly bounded.
  \end{enumerate}
\end{assumption}

\begin{lemma} \label{lem:od_consistent}
    Assume the exponential family edge model \eqref{expo_edges} with overdispersion parameter $\phi_{ij} \equiv \phi$ satisfies Assumptions~\ref{assump:theta_est} and~\ref{assump:glm_od}, and there exist consistent estimators of $\tau_n$ which satisfy
$$
  [\widetilde{\Theta}^{(g)}]_{ij} - [\vecz^{-1}(\tau_n)]_{ij} = o_{\prob}(1)
$$
as $n \rightarrow \infty$, for each $g \in \{1,2\}$ and $(ij) \in \mathcal{S}$. Then $\hat{\phi}^2 \inprob \phi^2$, where $\hat{\phi}^2$ is defined in \eqref{phi_hat1}.
\end{lemma}

Combining Lemma~\ref{lem:od_consistent} with Proposition~\ref{prop:wobs_largen} gives the following asymptotic result for the distribution of $w^{(\mathrm{E-UD})}$ with unknown dispersion.

\begin{corollary} \label{cor:wobs_od}
    Assume the exponential family edge model \eqref{expo_edges} with overdispersion parameter $\phi_{ij} \equiv \phi$ satisfies Assumptions~\ref{assump:theta_est} and~\ref{assump:glm_od}, and suppose $\{U_n\}$ and $\{V_n\}$ satisfy Assumption~\ref{assump:coherence}.
  Let $\gamma_n$ be the solution of \eqref{population_glm}, and assume there exist constants $K^-$ and $K^+$ such that
  \begin{equation} 
      0 < K^- \leq h'(\bm{w}_n) \leq K^+ < \infty
  \end{equation}
  elementwise and uniformly over $n$, when $\bm{w}_n$ is replaced by either $\tau_n$ or $\bar{X}_n\gamma_{n}$.
  Assume the maximum likelihood estimate from fitting a GLM with responses $\vecz(Y)$, and design matrix $(\bar{X}_n~\otimes~\bm{1}_m)$ exists, and denote it by $(\hat{\gamma}_{1,n}, \hat{\gamma}_{2,n} )^{\tp}$.
  Then as long as the product $rcm \rightarrow \infty$,
  $w^{(\mathrm{E-UD})}$ has an asymptotic $\chi^2_{d_rd_c}$ distribution as $n \rightarrow \infty$.
\end{corollary}
\fi

\ifjasa
\else
\subsubsection{The Gaussian edge model} \label{subsec:theory_gaussian}
\fi


Under the Gaussian edge model, we have the following two results, which hold under arbitrary fixed alternatives.
\begin{proposition} \label{prop:fobs_tilde}
  Under the Gaussian edge model with edge variance $\sigma^2$, $w^{(\mathrm{G})}$ 
  has a doubly non-central $F_{\nu_1,\nu_2}(\psi,\zeta)$ distribution \citepbody{bulgren71representations}, where
  \begin{equation*}
    \hspace{-0.5cm}
    \psi = 
    \frac{m}{2\sigma^2} \left\lVert U^{\tp} ( \Theta_{\mathcal{S}}^{(1)} - \Theta_{\mathcal{S}}^{(2)}) V \right\rVert_F^2, \quad
    \zeta = 
    \frac{m}{2\sigma^2} \left\lVert Q^{\perp}(U,V) \vecz(
    \Theta^{(1)}_{\mathcal{S}} - \Theta^{(2)}_{\mathcal{S}}) \right\rVert_2^2,
  \end{equation*}
  and $Q^{\perp}(U,V)$ is defined in Section~\ref{subsubsec:gaussian_test1}.
\end{proposition}
Note that $\psi$ corresponds to the non-centrality in the numerator, and $\zeta$ to the non-centrality in the denominator of the test statistic.
Under the null hypothesis, we have $\psi = 0$ and $\zeta = 0$ for all $U$ and $V$, so that $w^{(\mathrm{G})}$ follows an $F$-distribution.
This establishes the size of the test given in Section~\ref{subsubsec:gaussian_test1}.

\begin{proposition} \label{prop:fobs}
  Suppose $m>1$.
  Under the Gaussian edge model with edge variance $\sigma^2$, $w^{(\mathrm{G-P})}$ (see Section~\ref{subsubsec:gaussian_test2})
  has a non-central $F_{\nu_1,\nu_2'}(\psi)$ distribution, where $\psi$ is defined in Proposition~\ref{prop:fobs_tilde}.
\end{proposition}

Under the mesoscale null hypothesis, $\psi=0$ regardless of the choices of $U$ and $V$, which establishes the size of the test given in Section~\ref{subsubsec:gaussian_test2}.
Relative to Proposition~\ref{prop:fobs_tilde}, this size guarantee is robust to certain types of covariance misspecification.
\begin{remark} \label{rem:restricted_cov}
    Under the Gaussian edge model, $\operatorname{Cov}\{\vecz(A_k^{(g)})\} = \sigma^2 I_{rc}$ for $g=1, 2$ and $k=1,\ldots,m$. However, Proposition~\ref{prop:fobs} will continue to hold under the weaker condition $\operatorname{Cov}\{(V \otimes U)^{\tp} \vecz(A_k^{(g)})\} = \sigma^2 I_{d_rd_c}$.
\end{remark}
Propositions~\ref{prop:fobs_tilde} and \ref{prop:fobs} show that test power depends on the model parameters through the non-centrality parameter $\psi$, and hence motivates choices of $U$ and $V$ which maximize $\psi$.
The non-centrality parameter $\psi$ is a special case of the non-centrality parameter $\tilde{\psi}$ defined in Proposition~\ref{prop:wobs}.

As noted in Section~\ref{subsec:proj_learning}, for any choice of $d = d_r = d_c$, the optimal asymptotic power is the result of choosing projections which correspond to the top $d$ left and right singular vectors of
$
  \Theta^{(1)}_{\mathcal{S}} - \Theta^{(2)}_{\mathcal{S}}.
$
While this oracle power is well-defined for fixed $d$, even with known edge expectation parameters, the optimal choice of $d$ does not necessarily maximize $\psi$, but depends on the trade-off between the non-centrality parameter, and the (numerator) degrees of freedom of the test statistic \citepbody{dasgupta74power}.

\ifjasa
\else
To help interpret the oracle projections, we consider a particular class of low rank alternatives for ranks $d_0$, $d_1$ and $d_2$ satisfying $d_0+d_1+d_2 \leq \min\{r,c\}$.
Suppose the edge expectation parameter matrices in $\mathcal{S}$ have SVDs given by
$$
  U^{(0)} \operatorname{diag}(\Gamma^{(g)}) (V^{(0)})^{\tp} + U^{(g)} \operatorname{diag}(\Lambda^{(g)}) \left\{ V^{(g)} \right\}^{\tp}.
$$
for $g=1,2$, up to reordering of columns.
That is, the entries $\Gamma^{(g)} \in \mathbb{R}^{d_0}$ and $\Lambda^{(g)} \in \mathbb{R}^{d_g}$ need not be in decreasing order, internally or relative to each other.
Under this parameterization, $U^{(0)}$ and $V^{(0)}$ describe the column and row space structure shared between both samples, while $U^{(g)}$ and $V^{(g)}$ for $g=1,2$ describe structure which is individual to one sample, similar to models for multilayer networks in
\ifblind
\citebody{zhang20flexible}, \citebody{loyal21eigenmodel}, and others.
\else
\citebody{macdonald22latent}, \citebody{zhang20flexible}, \citebody{loyal21eigenmodel}, and others.
\fi
Also suppose that the columns of $U^{(1)}$ and $U^{(2)}$ are mutually orthogonal,
and the same for the right singular vectors.
In this way, the parameterization decomposes both $\mathbb{R}^r$ and $\mathbb{R}^c$ into four mutually orthogonal subspaces: a common part, an individual part for each sample, and a residual part which appears in neither sample's edge expectation parameter matrix.

Under this alternative, we can write the SVD of the difference of edge expectations $\Theta^{(1)}_{\mathcal{S}} - \Theta^{(2)}_{\mathcal{S}}$ (up to sign flips and reordering of columns) as
$$
\begin{pmatrix}
  U^{(0)} & U^{(1)} & U^{(2)}
\end{pmatrix} \left\{ \operatorname{diag}\begin{pmatrix}
  \Gamma^{(1)} - \Gamma^{(2)} & \Lambda^{(1)} & \Lambda^{(2)}
\end{pmatrix} \right\} \begin{pmatrix}
  V^{(0)} & V^{(1)} & -V^{(2)}
\end{pmatrix}^{\tp}. 
$$
Thus for any choice of $d \leq d_0 + d_1 + d_2$, the optimal asymptotic power is the result of choosing the projection spaces correspond to the top singular subspaces, ordered according to the the entries of
$$
  \begin{pmatrix}
    \lvert \Gamma^{(1)} - \Gamma^{(2)} \rvert & \Lambda^{(1)} & \Lambda^{(2)}
  \end{pmatrix}.
$$
For directions in the space common to both samples, the priority is to include dimensions with different singular values, while directions in the individual spaces are ordered according to the magnitude of their singular values.
In particular, directions in $\operatorname{col}(U^{(0)})$ and $\operatorname{col}(V^{(0)})$ corresponding to large singular values in the means of each sample may cancel in the difference, and thus no longer be useful for discriminating those two samples.
Moreover, even if the expected adjacency matrices are not low rank, if $d_0$ is large and many entries of $\Gamma^{(1)} - \Gamma^{(2)}$ are close to zero, we may still be able to identify low-dimensional projections which capture much of the discriminative signal in the difference.
\fi

\ifjasa
\else
\subsection{Results for learned projections} 
\label{subsec:theory_learned}


In this section, we restrict our attention to the Gaussian edge model and the case $m_1 = m_2 = m$.  In a non-asymptotic setting with a fixed alternative, we will establish a bound on the difference in power between an oracle projection test $\mathcal{H}_{0,\mathcal{S}}$ and the one using learned projections based on an estimator $\widehat{\mathbb{T}}$, as described in  Section~\ref{subsec:proj_learning}.

We present this result in two parts.  First, we state a generic bound on the difference in power between the oracle and learned projection tests given a high probability bound on the subspace errors, and then establish a high probability bound on the subspace errors of our learned projections.
\begin{proposition} \label{prop:power_bound}
  For a fixed alternative, let $U_*$ and $V_*$ denote the oracle $d_r$ and $d_c$ dimensional projections for testing $\mathcal{H}_{0,\mathcal{S}}$.
  Let
  $$
    \psi_{\mathrm{orc}} = \frac{m}{2\sigma^2} \lVert U_*^{\tp} \left( \Theta^{(1)}_{\mathcal{S}} - \Theta^{(2)}_{\mathcal{S}} \right) V_* \rVert_F^2
  $$
  denote the oracle non-centrality parameter for $w^{(\mathrm{G-P})}$ with dimensions $d_r$ and $d_c$. Let
  $$
    \psi_{\mathrm{max}} = \frac{m}{2\sigma^2} \lVert \Theta^{(1)}_{\mathcal{S}} - \Theta^{(2)}_{\mathcal{S}} \rVert_F^2,
  $$
  the maximum non-centrality parameter for any choice of dimension.
  For some $\alpha \in (0,1)$, denote the type II error rate of the level $\alpha$ oracle projection test by $\beta_{\mathrm{orc}}$ and the type II error rate of the level $\alpha$ learned projection test by $\beta_{\mathrm{learn}}$.

  Suppose the learned projection test is performed with orthonormal bases $\hat{U} \in \mathbb{R}^{r \times d_r}$ and $\hat{V} \in \mathbb{R}^{c \times d_c}$ which are independent of the test statistics, and satisfy
  \begin{equation} \label{subspace_cond}
    \min_{O \in \mathcal{O}_{d_r}} \lVert \hat{U} - U_*O \rVert_2 + \min_{O \in \mathcal{O}_{d_c}} \lVert \hat{V} - V_*O \rVert_2 \leq \epsilon < \min\left\{ \sqrt{\frac{\psi_{\mathrm{orc}}}{\psi_{\mathrm{max}}}}, \frac{1}{\sqrt{\psi_{\mathrm{orc}}\psi_{\mathrm{max}}}} \right\} 
  \end{equation}
  with probability at least $1-\xi$.
  Then
  $$
    \beta_{\mathrm{orc}} \leq \beta_{\mathrm{learn}} \leq \left( \frac{1}{1 - \epsilon \sqrt{\psi_{\mathrm{orc}}\psi_{\mathrm{max}}}} \right) \beta_{\mathrm{orc}} + \xi.
  $$
\end{proposition}
As one would expect, if the subspace errors are small with high probability, the power of the tests using oracle and learned projections will be similar.

To complete the main result of this section, we will now prove a high probability bound on the subspace errors of our learned projections.
We first recall some of the notation introduced in Section~\ref{subsec:proj_learning}.
We split $\Theta^{(1)} - \Theta^{(2)}$ into four blocks
\begin{equation*}
  \Theta^{(1)} - \Theta^{(2)} = \begin{pmatrix}
    \Theta^{(1)}_{\mathcal{C}} - \Theta^{(2)}_{\mathcal{C}} & \Theta^{(1)}_{\mathcal{S}} - \Theta^{(2)}_{\mathcal{S}} \\
    \Theta^{(1)}_{\mathcal{D}} - \Theta^{(2)}_{\mathcal{D}} & \Theta^{(1)}_{\mathcal{R}} - \Theta^{(2)}_{\mathcal{R}}
  \end{pmatrix},
\end{equation*}
with SVD
\begin{equation*}
  U_{\mathrm{diff}}S_{\mathrm{diff}}V_{\mathrm{diff}}^{\tp} = \begin{pmatrix}
    U_{[r]}S_{\mathrm{diff}}V_{-[c]}^{\tp} & U_{[r]}S_{\mathrm{diff}}V_{[c]}^{\tp} \\
    U_{-[r]}S_{\mathrm{diff}}V_{-[c]}^{\tp} & U_{-[r]}S_{\mathrm{diff}}V_{[c]}^{\tp}
  \end{pmatrix},
\end{equation*}
and assume that this matrix has rank $d_* < \min\{r,c\}$.
Denote the maximum and minimum entries of $S$ by $s_{\mathrm{max}}$ and $s_{\mathrm{min}}$ respectively.

We also make an assumption on the left and right singular vectors which is intermediate to the conditions assumed in Propositions~\ref{prop:onestep_weak} and \ref{prop:onestep_strong}.
\begin{assumption} \label{assump:spectral_conditioning}   Assume that 
  \begin{align*}
    \rho_U I_{d_*} \preceq \left( U_{[r]} \right)^{\tp} U_{[r]} \preceq \kappa_U I_{d_*}, \ \ 
    \rho_V I_{d_*} \preceq \left( V_{[c]} \right)^{\tp} V_{[c]} \preceq \kappa_V I_{d_*}.
  \end{align*}
  for $0 < \rho_U \leq \kappa_U < 1$ and $0 < \rho_V \leq \kappa_V < 1$.
\end{assumption}

Finally, we formally define our plug-in estimator of $\mathbb{T}$,
\begin{equation}
  \widehat{\mathbb{T}} = (\widehat{\Theta}^{(1)}_{\mathcal{C}} - \widehat{\Theta}^{(2)}_{\mathcal{C}})\left(\left[ \widehat{\Theta}^{(1)}_{\mathcal{D}} - \widehat{\Theta}^{(2)}_{\mathcal{D}} \right]_{(d_*)} \right)^{\dagger}(\widehat{\Theta}^{(1)}_{\mathcal{R}} - \widehat{\Theta}^{(2)}_{\mathcal{R}}),
\end{equation}
where $[ \cdot ]_{(d^*)}$ denotes the best rank $d_*$ approximation of its matrix argument, and
\begin{equation} \label{block_means}
  \widehat{\Theta}^{(g)}_{\mathcal{I}} = \frac{1}{m} \sum_{k=1}^m [A_k^{(g)}]_{\mathcal{I}}
\end{equation}
for $g=1,2$,  $\mathcal{I} \in \{\mathcal{C},\mathcal{D},\mathcal{R}\}$.
The learned projections will have orthonormal bases $\hat{U}$ and $\hat{V}$ given by the leading $d_*$ left and right singular vectors of $\widehat{\mathbb{T}}$.
Note that by Assumption~\ref{assump:spectral_conditioning} and Proposition~\ref{prop:onestep_weak}, the leading $d_*$ left and right singular vectors of $\mathbb{T}$ are orthonormal bases for the oracle projections.

First, we state a simultaneous high probability bound on the error of each factor of $\widehat{\mathbb{T}}$, before stating a high probability bound on the resulting subspace error.
\begin{lemma} \label{lem:operator_norms}
  Assume the Gaussian edge model holds.
  For integers $p,q \geq 1$, define
  $$
    \delta(p,q) = \frac{15}{2\log(3/2)} \cdot \frac{\sqrt{\log(\min\{p,q\})}}{\sqrt{p} + \sqrt{q}}.
  $$
  Note that $\delta(p,q) < 5$ for all $p,q \geq 1$, and $\delta(p,q) \rightarrow 0$ as $\max\{p,q\} \rightarrow \infty$.

  Define
  \begin{equation*}
    \hspace{-.5cm}
    \xi = 2\bigg( \exp\left\{ - \frac{\sigma^2}{8m} (\sqrt{r} + \sqrt{n-c})^2 \right\}
     + \exp\left\{  - \frac{\sigma^2}{8m} (\sqrt{n-r} + \sqrt{c})^2 \right\}  + \exp\left\{  - \frac{\sigma^2}{8m} (\sqrt{n-r} + \sqrt{n-c})^2 \right\} \bigg).
  \end{equation*}
  Then, with probability at least $1-\xi$, we have that
  \begin{align*}
    \lVert (\widehat{\Theta}^{(1)}_{\mathcal{C}} - \widehat{\Theta}^{(2)}_{\mathcal{C}}) - (\Theta^{(1)}_{\mathcal{C}} - \Theta^{(2)}_{\mathcal{C}}) \rVert_2 &\leq 2\{2 + \delta(r,n-c)\} \sigma \frac{\sqrt{r} + \sqrt{n-c}}{\sqrt{m}} =: \epsilon_{\mathcal{C}}, \\
    \lVert (\widehat{\Theta}^{(1)}_{\mathcal{R}} - \widehat{\Theta}^{(2)}_{\mathcal{R}}) - (\Theta^{(1)}_{\mathcal{R}} - \Theta^{(2)}_{\mathcal{R}}) \rVert_2 &\leq 2\{2 + \delta(n-c,r)\} \sigma \frac{\sqrt{n-c} + \sqrt{r}}{\sqrt{m}} =: \epsilon_{\mathcal{R}}, \\
    \left\lVert \left[ \widehat{\Theta}^{(1)}_{\mathcal{D}} - \widehat{\Theta}^{(2)}_{\mathcal{D}} \right]_{(d_*)} -  ( \Theta^{(1)}_{\mathcal{D}} - \Theta^{(2)}_{\mathcal{D}} ) \right\rVert_2 &\leq 4\{2 + \delta(n-r,n-c)\} \sigma \frac{\sqrt{n-r} + \sqrt{n-c}}{\sqrt{m}} =: \epsilon_{\mathcal{D}}.
  \end{align*}
  simultaneously over blocks $\mathcal{C}$, $\mathcal{R}$, and $\mathcal{D}$.
\end{lemma}


\begin{proposition} \label{prop:subspace_error}
  Under the Gaussian edge model, suppose the difference in expected adjacency matrices has rank $d_*$ and Assumption~\ref{assump:spectral_conditioning} holds, so that  $U_{[r]}$ and $V_{[c]}$ are the oracle projections of dimensions $d_r = d_c = d_*$ (see Proposition~\ref{prop:power_bound}). 
  Denote $\xi$, $\epsilon_{\mathcal{C}}$, $\epsilon_{\mathcal{R}}$, and $\epsilon_{\mathcal{D}}$ as in Lemma~\ref{lem:operator_norms}.  
  Suppose that
  \begin{equation} \label{D_err}
    2\epsilon_{\mathcal{D}} \leq s_{\mathrm{min}} \sqrt{(1-\rho_U)(1-\rho_V)}.
  \end{equation}
  Define
  \begin{align*} \label{epsilon_T}
    \epsilon_{\mathbb{T}} &= \sum_{j=1}^7 \epsilon_{\mathbb{T}}^{(j)}, \\
    \epsilon_{\mathbb{T}}^{(1)} &= \epsilon_{\mathcal{C}}\epsilon_{\mathcal{R}}\frac{1}{s_{\mathrm{min}}}, \quad
    &\epsilon_{\mathbb{T}}^{(2)} &= 4 \epsilon_{\mathcal{C}} \epsilon_{\mathcal{R}} \epsilon_{\mathcal{D}}\frac{1}{s_{\mathrm{min}}^2}, \\
    \epsilon_{\mathbb{T}}^{(3)} &=  \rho_V^{1/2} \epsilon_{\mathcal{C}} \frac{s_{\mathrm{max}}}{s_{\mathrm{min}}}, \quad
    &\epsilon_{\mathbb{T}}^{(4)} &= \rho_U^{1/2} \epsilon_{\mathcal{R}} \frac{s_{\mathrm{max}}}{s_{\mathrm{min}}}, \\
    \epsilon_{\mathbb{T}}^{(5)} &= 4 \rho_U^{1/2} \epsilon_{\mathcal{C}} \epsilon_{\mathcal{D}} \frac{s_{\mathrm{max}}}{s_{\mathrm{min}}^2},\quad
    &\epsilon_{\mathbb{T}}^{(6)} &= 4 \rho_V^{1/2} \epsilon_{\mathcal{R}} \epsilon_{\mathcal{D}} \frac{s_{\mathrm{max}}}{s_{\mathrm{min}}^2}, \\
    \epsilon_{\mathbb{T}}^{(7)} &= 4  (\rho_U\rho_V)^{1/2} \epsilon_{\mathcal{D}} \frac{s_{\mathrm{max}}^2 }{s_{\mathrm{min}}^2}. & &
  \end{align*}

  Then with probability at least $\xi$,
  $$
    \min_{O \in \mathcal{O}_{d_*}} \lVert \hat{U} - U_{[r]}O \rVert_2 + \min_{O \in \mathcal{O}_{d_*}} \lVert \hat{V} - V_{[c]}O \rVert_2 \leq \frac{6 \epsilon_{\mathbb{T}}}{(1-\rho_U)^2(1-\rho_V)^2} \left( \frac{s_{\mathrm{max}}}{s_{\mathrm{min}}^2}\right),
  $$
  where $\hat{U}$ and $\hat{V}$ are the leading $d_*$ left and right singular vectors of $\widehat{\mathbb{T}}$.
\end{proposition}
As $\widehat{\mathbb{T}}$ is independent of the test statistics, the main result of this section is obtained by plugging Proposition~\ref{prop:subspace_error} into Proposition~\ref{prop:power_bound}.
Although this result is fully non-asymptotic, to understand the bounds, we can consider an asymptotic regime where the overall network size $n$ grows, while $m$, $r$, $c$, and $\sigma$ remain fixed.
In this case, if the singular vectors satisfy a coherence condition similar to condition \eqref{coherence_cond} of Assumption~\ref{assump:coherence}, \citepbody{cape19twotoinfinity}, we would expect $\rho_U$ and $\rho_V$ to have asymptotic order $1/n$,
with $\psi_{\mathrm{orc}} = \psi_{\mathrm{max}}$, and $\psi_{\mathrm{orc}}$ asymptotically bounded.
Note that for fixed $m$ and $\sigma$, $\xi$ will go to zero exponentially in $n$, so all that remains for convergence in type II error is that \eqref{D_err} holds for sufficiently large $n$, and
\begin{equation} \label{epsilon_T_cond}
  \frac{s_{\mathrm{max}} \epsilon_{\mathbb{T}}}{s_{\mathrm{min}}^2} \rightarrow 0
\end{equation}
as $n \rightarrow \infty$.
In this regime, the asymptotic leading term of $\epsilon_{\mathbb{T}}$ will be either $\epsilon_{\mathbb{T}}^{(1)}$ or $\epsilon_{\mathbb{T}}^{(2)}$.
Both \eqref{D_err} and \eqref{epsilon_T_cond} are satisfied if the difference in adjacency matrices is well conditioned (i.e., $s_{\mathrm{max}}/s_{\mathrm{min}}$ is asymptotically bounded), and its singular values dominate $\sqrt{n}$ asymptotically, similar to a spiked singular value condition often seen in random matrix theory 
\ifjasa
\citepsupp{gavish17optimal}.
\else 
\citepbody{gavish17optimal}.
\fi
\fi

\section{Experiments on simulated and real data} \label{sec:simulations}

\ifjasa
\else
In this section, we apply our projection tests to mesoscale network testing with synthetic networks from the Gaussian edge model with unknown edge variance $\sigma^2$, and the logistic link binary edge model with both known and unknown dispersion parameters.
\fi
\ifjasa
\else
We verify the size of our testing procedures under any choice of projection dimension $d$, as well as their empirical power properties compared to classical tests which ignore the network structure, and bootstrap tests based on a latent space model, which are expensive to calibrate, and are not robust to model misspecification.
\fi
We then apply our tests to compare functional magnetic resonance imaging (fMRI) brain images between patients with Parkinson's disease and healthy controls.

\ifjasa
\else
\subsection{Synthetic Gaussian edge networks} \label{subsec:gaussian_sims}

In this section, we generate synthetic networks from a Gaussian edge model on $n=100$ nodes, with variance $\sigma^2=50$.
We test the equality of a $20 \times 30$ off-diagonal rectangle with nominal level $\alpha=0.05$,
and vary $m$, the number of networks per sample from $1$ to $10$.
For each of $500$ independent replications, the expected adjacency matrices are generated according to the following setup.
\begin{enumerate}
  \item The network is generated from a $3$-dimensional inner product latent space model, with
  $\Theta^{(g)} = X^{(g)}[Y^{(g)}]^{\tp}$. The entries of $X^{(g)}$ and $Y^{(g)}$ are Gaussian random variables.
  \item In the null setting,
  \begin{align*}
    &X_i^{(1)} \sim \mathcal{N}(0,I_3), &\quad &X_i^{(2)} = X_i^{(1)} &\quad &(i=1,\ldots,20), \\
    &X_i^{(1)} \sim \mathcal{N}(0,I_3), &\quad &X_i^{(2)} \sim \mathcal{N}(0,I_3) &\quad &(i=21,\ldots,100), \\
    &Y_i^{(1)} \sim \mathcal{N}(0,I_3), &\quad &Y_i^{(2)} \sim \mathcal{N}(0,I_3) &\quad &(i=1,\ldots,70), \\
    &Y_i^{(1)} \sim \mathcal{N}(0,I_3), &\quad &Y_i^{(2)} = Y_i^{(1)} &\quad &(i=71,\ldots,100). \\
  \end{align*}
  \item In the alternative setting,
  \begin{align*}
    &X_i^{(1)} \sim \mathcal{N}(0,I_3), &\quad &X_i^{(2)} = X_i^{(1)} + \mathcal{N}(0,\frac{1}{3\sqrt{2}}I_3) &\quad &(i=1,\ldots,20), \\
    &X_i^{(1)} \sim \mathcal{N}(0,I_3), &\quad &X_i^{(2)} \sim \mathcal{N}(0,I_3) &\quad &(i=21,\ldots,100), \\
    &Y_i^{(1)} \sim \mathcal{N}(0,I_3), &\quad &Y_i^{(2)} \sim \mathcal{N}(0,I_3) &\quad &(i=1,\ldots,70), \\
    &Y_i^{(1)} \sim \mathcal{N}(0,I_3), &\quad &Y_i^{(2)} = Y_i^{(1)} + \mathcal{N}(0,\frac{1}{3\sqrt{2}}I_3) &\quad &(i=71,\ldots,100). \\
  \end{align*}
\end{enumerate}
This setup produces networks with exactly low rank edge expectation matrices in both the null and alternative cases.
Moreover, since the positions are generated with spherical covariance, the conditions of Proposition~\ref{prop:onestep_strong} should hold in expectation, and thus $\widehat{\mathbb{T}}$ should provide a good estimate of the oracle projections for any choice of dimension.

The first competing method (``Basic'') is a baseline approach which ignores the network structure, discards the edges in $-\mathcal{S}$, and applies the classical $F$-test of equality of means to the two vectorized samples,
$$
  \{ \operatorname{vec}([A_k^{(1)}]_{\mathcal{S}}) \}_{k=1}^m, \tabby \{ \operatorname{vec}([A_k^{(2)}]_{\mathcal{S}}) \}_{k=1}^m.
$$
Note that this test is only possible for $m > 1$.
In these settings with relatively small $m$, this amounts to $600$ univariate comparisons, each based on at most $20$ Gaussian samples.
While this test does pool the univariate comparisons, it does no dimension reduction.

The second competing method (``Posn-$p$'') is a bootstrap test based on an inner product latent space model for the networks in each sample.
For $p=2,3,4$, we estimate expected adjacency matrices for each group under both the null and alternative mesoscale hypotheses and assuming a $p$-dimensional latent space model.
Under the alternative, $\widehat{\Theta}^{(1)}$ and $\widehat{\Theta}^{(2)}$ amount to rank $p$ truncation of the mean adjacency matrix for each sample.
Under the null we fit low rank estimates $\widetilde{\Theta}^{(1)}$ and $\widetilde{\Theta}^{(2)}$ constrained to have $\widetilde{\Theta}^{(1)}_{\mathcal{S}} = \widetilde{\Theta}^{(2)}_{\mathcal{S}}$.
Our implementation of this constrained estimation is described in Appendix~\ref{app:admm}.

Then, similar to the ``Boot-EPA'' test described in \citebody{ghoshdastidar18practical}, we calculate a Frobenius norm statistic
$
    \lVert \widehat{\Theta}^{(1)}_{\mathcal{S}} - \widehat{\Theta}^{(2)}_{\mathcal{S}} \rVert_F,
$
and calibrate a cutoff based on $500$ replications under a parametric bootstrap, fit using the estimated expected adjacency matrices and edge variance under the null. 
In this setting, the models with $p \geq 3$ are correctly specified, although the model with $p=4$ is overparameterized.

Finally, for $m>1$ we apply our projection tests with statistic $w^{(\mathrm{G-P})}$ (``OrcProj'', ``Proj-$d$'') using both oracle and learned left and right projections.
For $m=1$ we instead use the projection test with statistic $w^{(\mathrm{G})}$.
Under the oracle, the left and right projections are the $6$-dimensional singular subspaces of the true difference $\Theta^{(1)}_{\mathcal{S}} - \Theta^{(2)}_{\mathcal{S}}$, while our learned projections will be based on the leading $d$-dimensional singular subspaces of $\widehat{\mathbb{T}}$, for $d=4,6,8$.
We implement these projection tests with twice the dimensions since this is the analogous dimension to the one used in the position model bootstrap test.
The results for all of these methods are shown in Figure~\ref{gauss1}.

\begin{figure}
\twoImages{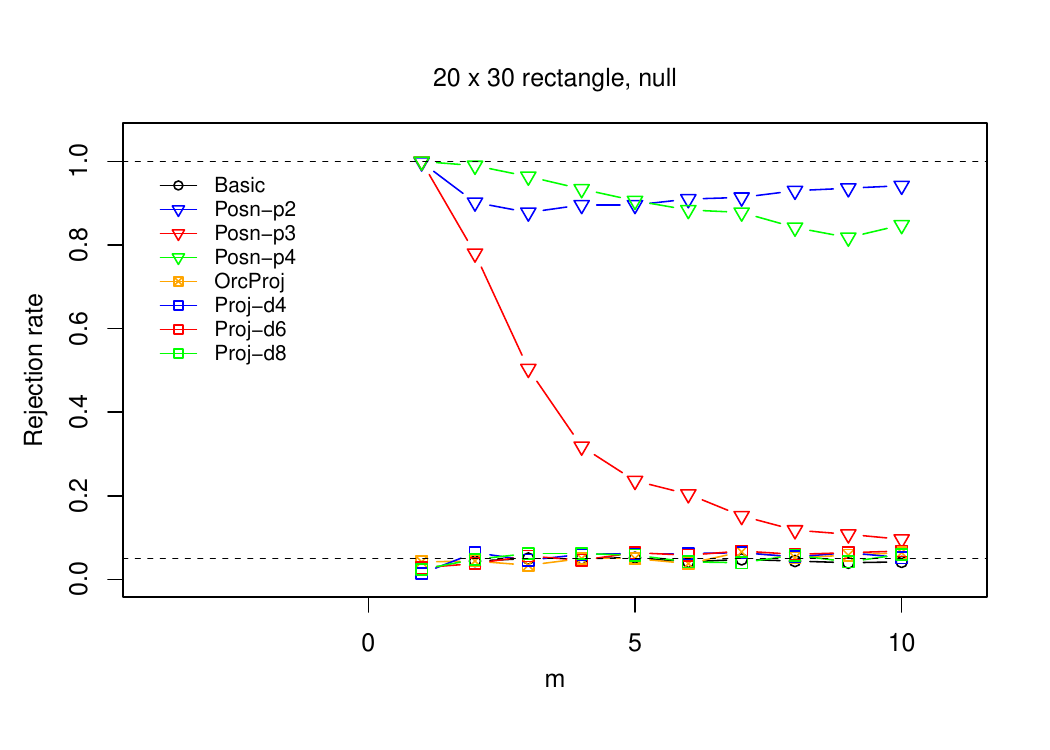}{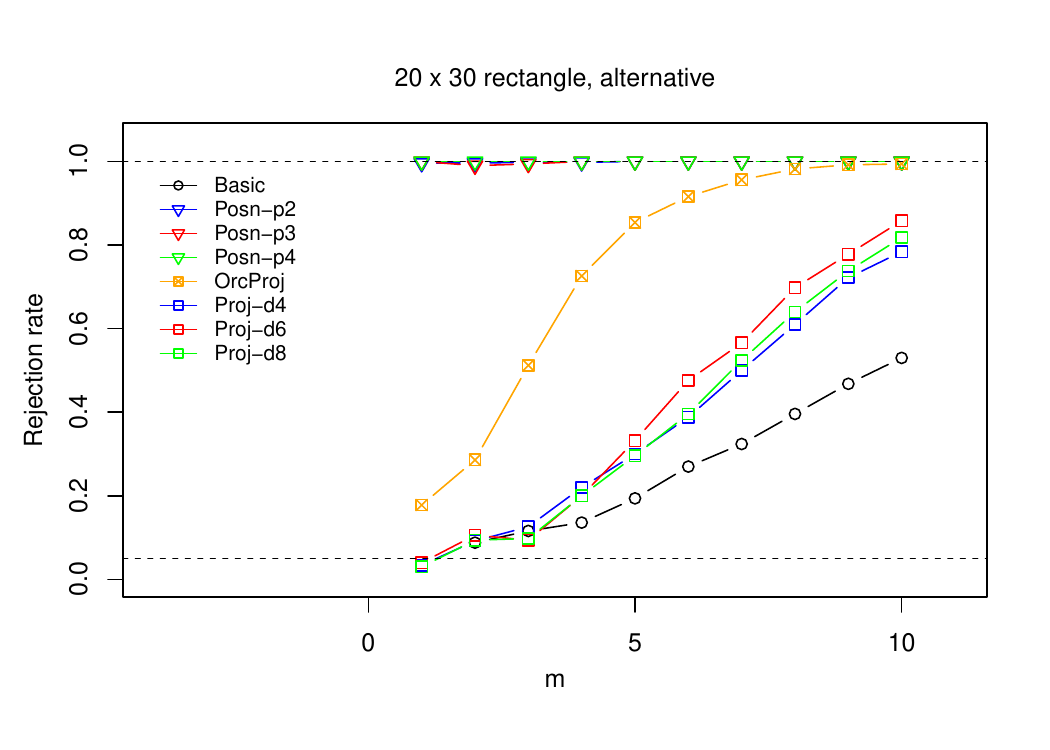}
\caption{Rejection rate for Gaussian edge networks with inner product latent space model. Dashed lines correspond to $\alpha=0.05$ and $1$. 
Point colors correspond to $2p$ or $d$, with basic tests displayed in black and oracle projection tests in orange.}
\label{gauss1}
\end{figure}

In the left panel of Figure~\ref{gauss1}, we compare the methods under the null.
As expected, the basic and all the projection tests control the type I error at the nominal level.
On the other hand, the position model bootstrap tests have at best asymptotic control---the rejection rate for $p=3$ appears to be decreasing towards the nominal level, but this is not the case for the misspecified model with $p=2$ or the overparameterized model with $p=4$.

In the right panel of Figure~\ref{gauss1}, we compare the methods under the alternative.
The position model bootstrap tests have the highest empirical power, but this is not comparable given their issues maintaining control of the type I error rate.
The projection tests show increased power compared to the basic test, especially as $m$ grows and the learned projections approach the oracle.
The rejection rates of the tests based on learned projections are also fairly insensitive to the choice of dimension.

As a second comparison, we demonstrate the efficacy of our method when the expected adjacency matrices are not exactly low rank, but instead have low effective rank.
Similar to the first setting, we test the equality of a $20 \times 30$ off-diagonal rectangle with nominal level $\alpha=0.05$, in Gaussian edge networks on $n=100$ nodes, and again varying $m$ from $1$ to $10$.
We generate the expected adjacency matrices according to the following setup.
\begin{enumerate}
  \item The network is generated from a $3$-dimensional Euclidean distance latent space model with
  $
    \Theta^{(g)}_{ij} = \lVert X^{(g)}_i - Y^{(g)}_j \rVert_2
  $. The entries of $X^{(g)}$ and $Y^{(g)}$ are Gaussian random variables.
  \item In the null setting,
  \begin{align*}
    &X_i^{(1)} \sim \mathcal{N}(0,I_3), &\quad &X_i^{(2)} = X_i^{(1)} &\quad &(i=1,\ldots,20), \\
    &X_i^{(1)} \sim \mathcal{N}(0,I_3), &\quad &X_i^{(2)} \sim \mathcal{N}(0,I_3) &\quad &(i=21,\ldots,100), \\
    &Y_i^{(1)} \sim \mathcal{N}(0,I_3), &\quad &Y_i^{(2)} \sim \mathcal{N}(0,I_3) &\quad &(i=1,\ldots,70), \\
    &Y_i^{(1)} \sim \mathcal{N}(0,I_3), &\quad &Y_i^{(2)} = Y_i^{(1)} &\quad &(i=71,\ldots,100). \\
  \end{align*}
  \item In the alternative setting,
  \begin{align*}
    &X_i^{(1)} \sim \mathcal{N}(0,I_3), &\quad &X_i^{(2)} = X_i^{(1)} + \mathcal{N}(0,\frac{\sqrt{2}}{3}I_3) &\quad &(i=1,\ldots,20), \\
    &X_i^{(1)} \sim \mathcal{N}(0,I_3), &\quad &X_i^{(2)} \sim \mathcal{N}(0,I_3) &\quad &(i=21,\ldots,100), \\
    &Y_i^{(1)} \sim \mathcal{N}(0,I_3), &\quad &Y_i^{(2)} \sim \mathcal{N}(0,I_3) &\quad &(i=1,\ldots,70), \\
    &Y_i^{(1)} \sim \mathcal{N}(0,I_3), &\quad &Y_i^{(2)} = Y_i^{(1)} + \mathcal{N}(0,\frac{\sqrt{2}}{3}I_3) &\quad &(i=71,\ldots,100). \\
  \end{align*}
\end{enumerate}


In this setting, the expected adjacency matrices have low dimensional, nonlinear structure.
However, they still exhibit a low effective rank, with visual inspection of a scree plot showing $5$ spiked singular values.
We apply the same methods, but for the projection tests the learned projections are the left and right singular subspaces of a rank $d$ imputation of the mean difference in adjacency matrices, after removing the edges in $\mathcal{S}$ \citepbody{mazumder10spectral}.
The results are shown in Figure~\ref{gauss2}.

\begin{figure}
\twoImages{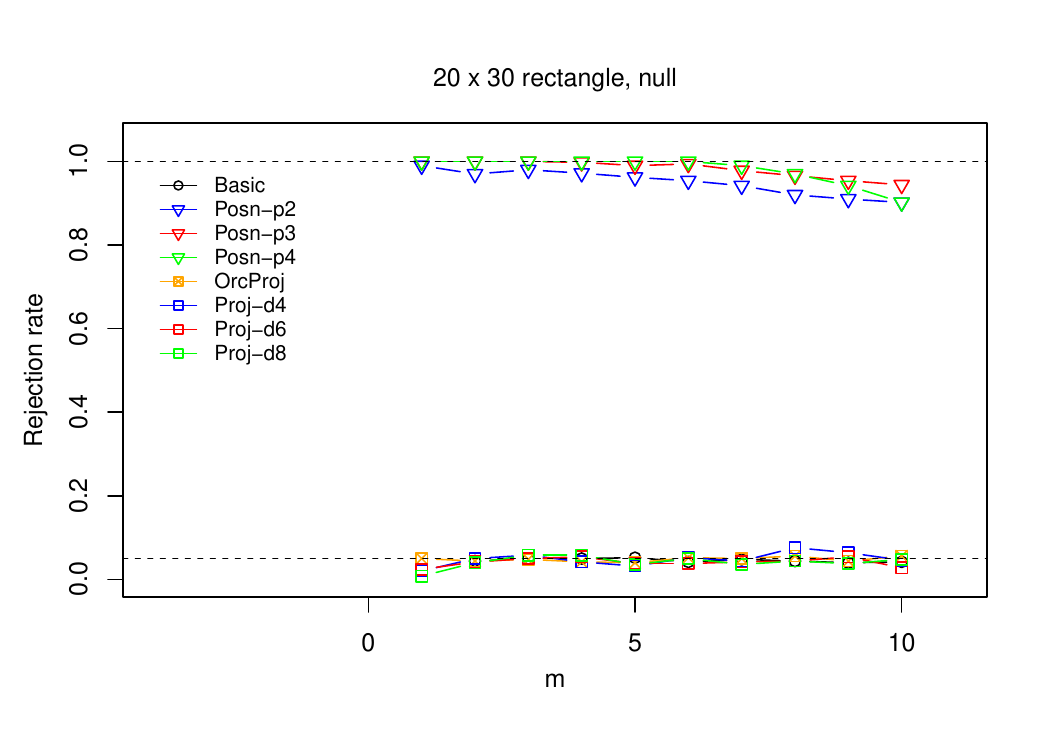}{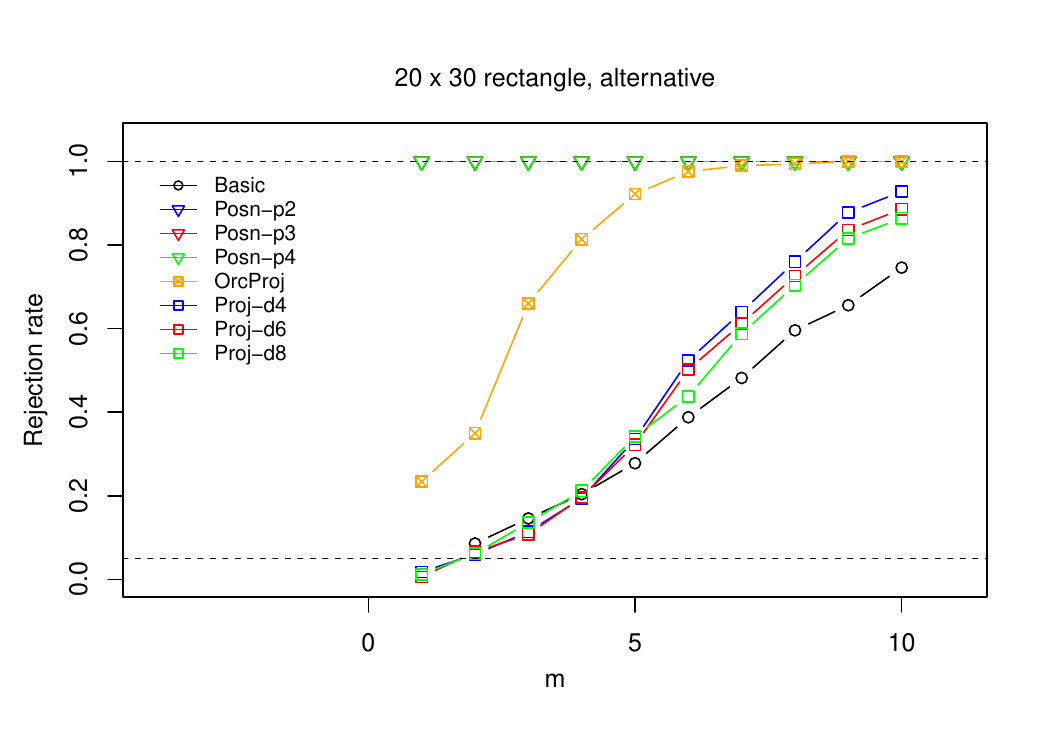}
\caption{Rejection rate for Gaussian edge networks with Euclidean distance latent space model. Dashed lines correspond to $\alpha=0.05$ and $1$. 
Point colors correspond to $2p$ or $d$, with basic tests displayed in black and oracle projection tests in orange.}
\label{gauss2}
\end{figure}

In the left panel of Figure~\ref{gauss2}, we compare the methods under the null.
As before, the basic and all the projection tests exactly control type I error at the nominal level, but now all of the position model bootstrap tests are misspecified, and nearly always reject under the null.

In the right panel of Figure~\ref{gauss2}, we compare the methods under the alternative.
Even though the expected adjacency matrices are not exactly low rank, the projection tests show increased power compared to the basic test.
Although the scree plot for $\Theta^{(g)}$ would suggest a $5$-dimensional low rank model for each sample, the highest rejection rate for the projection test is with $d=4$, decreasing as $d$ grows.
This implies that the dimensions corresponding to the smaller spiked singular values either cannot be reliably recovered due to the noise, or are no longer spiked after taking the difference between the adjacency matrices.

\fi

\subsection{Synthetic binary edge networks} \label{subsec:binary_sims}

To evaluate mesoscale two-sample comparisons of networks with binary edges, we test two different mesoscale hypotheses 
using our test based on $w^{(\mathrm{E})}$ and its competitors.  For the first example, we generate synthetic logistic link binary edge networks, which falls under the general exponential family edge model case, with support $\{0,1\}$, dispersion $\phi=1$, and inverse link function $h(x) = e^x/(1 + e^x)$.
We vary $m$ from $1$ to $10$ and fix $n=100$ nodes.

The mesoscale hypothesis is about the connections between the first $r =20$ and the last $c=30$ nodes, and thus the set $\mathcal{S}$ is a $20 \times 30$ rectangle.
For each of $500$ independent replications, we generate the network from a $2$-dimensional inner product latent space model with
  $h^{-1}\left( \mathbb{E}A_1^{(g)} \right) = \Theta^{(g)} = X^{(g)}[Y^{(g)}]^{\tp}$. The entries of $X^{(g)}$ and $Y^{(g)}$ are generated as i.i.d.\ $\mathcal{N}(0,1)$ for all $i$ for $g=1$, and for all but the following exceptions for $g = 2$:      
\begin{enumerate}
  \item Under the null,  $X_i^{(2)} = X_i^{(1)}$ for $i=1,\ldots,20$, and $Y_i^{(2)} = Y_i^{(1)}$ for $i=71,\ldots,100$.
  \item Under the alternative, $X_i^{(2)} = X_i^{(1)} + \mathcal{N}(0,16^{-1}I_2)$  for $i=1,\ldots,20$, and  $Y_i^{(2)} = Y_i^{(1)} + \mathcal{N}(0,16^{-1}I_2)$ for $i=71,\ldots,100$.
\end{enumerate}
This setup produces edge expectation parameter matrices with rank $2$; non-zero differences in edge expectation parameter matrices have rank $4$.  The expected adjacency matrices $h(\Theta^{(g)})$ for $g=1,2$  also have low effective rank, with visual inspection of a scree plot showing three spiked singular values.  

We compare the following methods.  The ``Basic'' approach performs $600$ univariate (two-sided) binomial proportion tests for all the node pairs in $\mathcal{S}$. 
This gives us 600 independent test statistics asymptotically distributed as $\chi^2_1$, so the test is conducted by comparing their sum to the appropriate quantile of the $\chi^2_{600}$ distribution.
When $m=1$, this test is trivial and will never reject the null hypothesis.

The second competing method (``Posn-$p$'') is a bootstrap test based on an inner product latent space model for the networks in each sample.
For $p=2, 3, 4$, we estimate expected adjacency matrices for each group under both the null and alternative mesoscale hypotheses and assuming a $p$-dimensional inner product latent space model.
 Ignoring the link function, this corresponds to a projection test with $d = 2p$, as we assume each expected adjacency matrix has rank $p$, so their differences have rank at most $2p$.
Under the alternative, $\widehat{\Theta}^{(1)}$ and $\widehat{\Theta}^{(2)}$ are the rank $p$ truncations of the mean adjacency matrices for each sample.
Under the null, we compute rank $p$ estimates $\widehat{\Theta}^{(1,0)}$ and $\widehat{\Theta}^{(2,0)}$ constrained to have $\widehat{\Theta}^{(1,0)}_{\mathcal{S}} = \widehat{\Theta}^{(2,0)}_{\mathcal{S}}$.
Our implementation of this constrained estimation is described in Appendix~\ref{app:admm}.
Then, similar to the ``Boot-EPA'' test of \citebody{ghoshdastidar18practical}, we calculate a Frobenius norm statistic $\lVert \widehat{\Theta}^{(1)}_{\mathcal{S}} - \widehat{\Theta}^{(2)}_{\mathcal{S}} \rVert_F$, and calibrate a cutoff based on $500$ replications under a parametric bootstrap, resampling from a model fit with estimated expected adjacency matrices under the null. 
While the inner product latent space model does not match the model likelihood, it is common in the RDPG literature to apply such spectral approaches directly to the adjacency matrix, under a semiparametric signal-plus-noise assumption.

Finally, we apply our projection test based on $w^{(\mathrm{E})}$ with both oracle and learned projections of dimension $d = 4, 6, 8$ (``OrcProj'', ``Proj-$d$'').
Recall that $d=4$ is the correct rank of the difference in edge expectation matrices, and is used for the oracle projection.
In the computation of the test statistic, we use the regularized estimator defined in \eqref{theta_bayes} to compute $\widehat{F}$.  
The learned projections are based on an estimate of $\mathbb{T}$, after plugging in estimates of the matrix blocks learned from rank $d$ truncation of $h^{-1}\{(m+2)^{-1}\sum_k (A_k^{(g)} + 1)\}$ for $g=1, 2$, where the probability estimates are regularized as in \eqref{theta_bayes} to avoid evaluating $h^{-1}$ at $0$ or $1$.

\begin{figure}
\twoImages{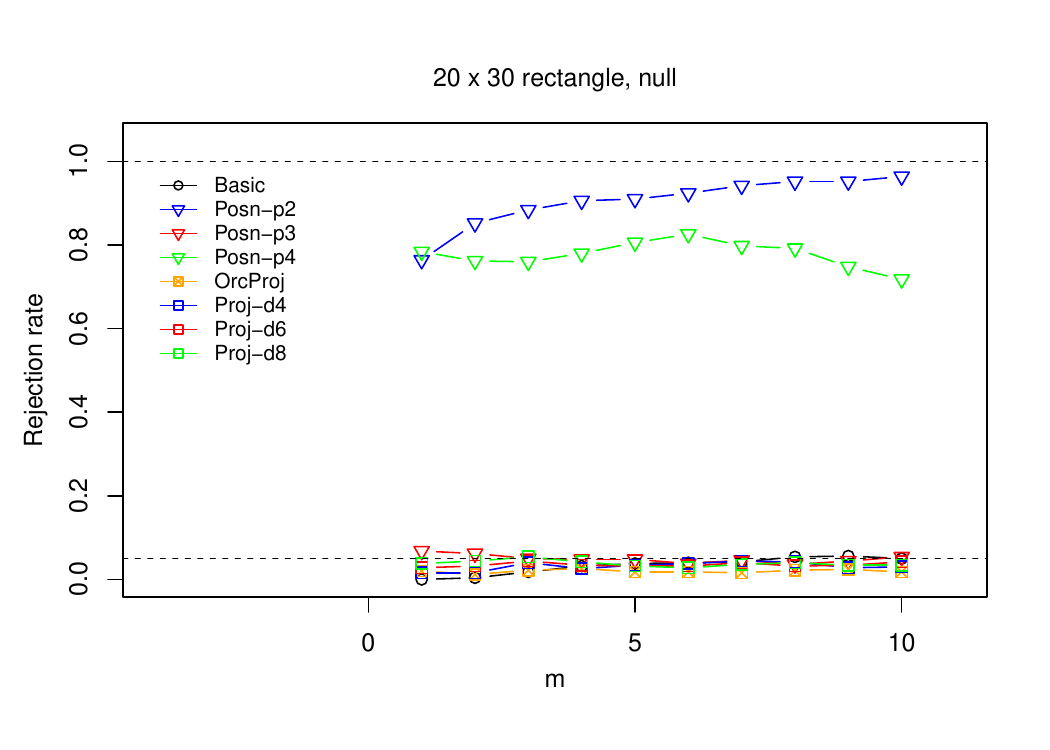}{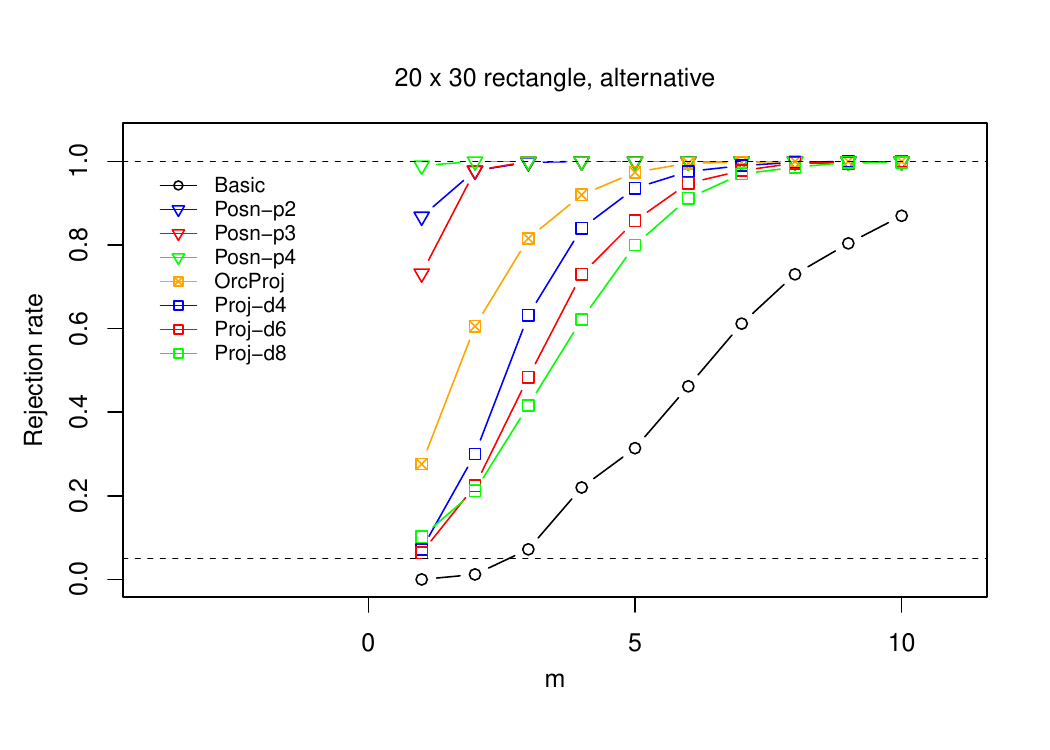}
\caption{Rejection rates for binary networks modeled with logistic link and the inner product latent space model. Dashed lines correspond to $\alpha=0.05$ and $1$. 
  Point colors correspond to $2p$ or $d$, with basic tests displayed in black and oracle projection tests in orange.}
\label{logit1}
\end{figure}

In the left panel of Figure~\ref{logit1}, we compare test performance under the null.
We see that our projection tests empirically control the type I error at the nominal level for any choice of projection dimension.
The position model bootstrap test empirically controls the type I error for $p=3$ but fails to come anywhere close to the nominal level for either $p=2$ or $p=4$.
The right panel of Figure~\ref{logit1} shows power (rejection rate under the alternative).
The projection tests are substantially more powerful than the basic test for any choice of projection dimension, with $d=4$ giving the best empirical power. 
The bootstrap tests achieves the highest power, but given that it only controls type I error for exactly the true latent dimension, one would be reluctant to recommend using it in practice.  

To investigate robustness properties of our projection tests, we also consider the case of overdispersed binary edges.
\ifjasa
\else
In the context of logistic regression with grouped data \citepbody{mccullagh83generalized}, overdispersion can occur when there is positive correlation among the trials in each group, such that the sum of those trials has variance greater than what would be expected from a binomial distribution with the same mean.
In our specific setting, the ``grouping'' comes from the $m$ binary edges for each fixed $(g,i,j)$ triple for $g=1,2$ and $1 \leq i,j \leq n$.

We generate expected adjacency matrices as in the previous binary edge logistic link model comparison, but introduce overdispersion with a beta-binomial model.

For $m > 1$ and a constant (over node pairs) overdisperion parameter $\eta \in [1,m)$, we generate edge variables by sampling
\begin{align*}
  B_{ij,g} &\sim \operatorname{Beta}\left( \frac{h(\Theta^{(g)}_{ij})(m - \eta)}{\eta - 1}, \frac{\{1-h(\Theta^{(g)}_{ij})\}(m - \eta)}{\eta - 1}\right), \\
  \sum_{k=1}^m [ A_k^{(g)} ]_{ij} &\sim \operatorname{Binomial}(m, B_{ij,g}).
\end{align*}
then choosing a configuration of $1$'s and $0$'s for the $m$ layers uniformly at random.
\citebody{ahn95generation} show that this hierarchical sampling strategy produces edges which are independent over $(g,i,j)$ triples,
with dispersion parameter $\eta$.
and satisfy
\begin{equation*}
  \mathbb{E}( [ A_k^{(g)} ]_{ij} ) = h(\Theta^{(g)}_{ij}), \tabby
  \operatorname{Var}\left( \sum_{k=1}^m [ A_k^{(g)} ]_{ij} \right) = \eta m h(\Theta^{(g)}_{ij}) \{ 1 - h(\Theta^{(g)}_{ij})\}
\end{equation*}
for $g=1,2$, $k=1,\ldots,m$ and node pairs $(i,j)$.
In this simulation we set $\eta=2$, and vary $m$ from $3$ to $10$.
\fi
We then apply the projection test described in Section~\ref{subsubsec:glm_test_od}, using the modified statistic $w^{(\mathrm{E-UD})}$ (``OrcProjUD'', ``ProjUD-$d$'').
For comparison we also apply the usual projection test based on $w^{(\mathrm{E})}$, and the same competing basic test and the position model bootstrap test with $p=3$.
For all of the projection tests, we use a learned $4$-dimensional projection.
The results are shown in Figure~\ref{logit2}.  

\begin{figure}
\twoImages{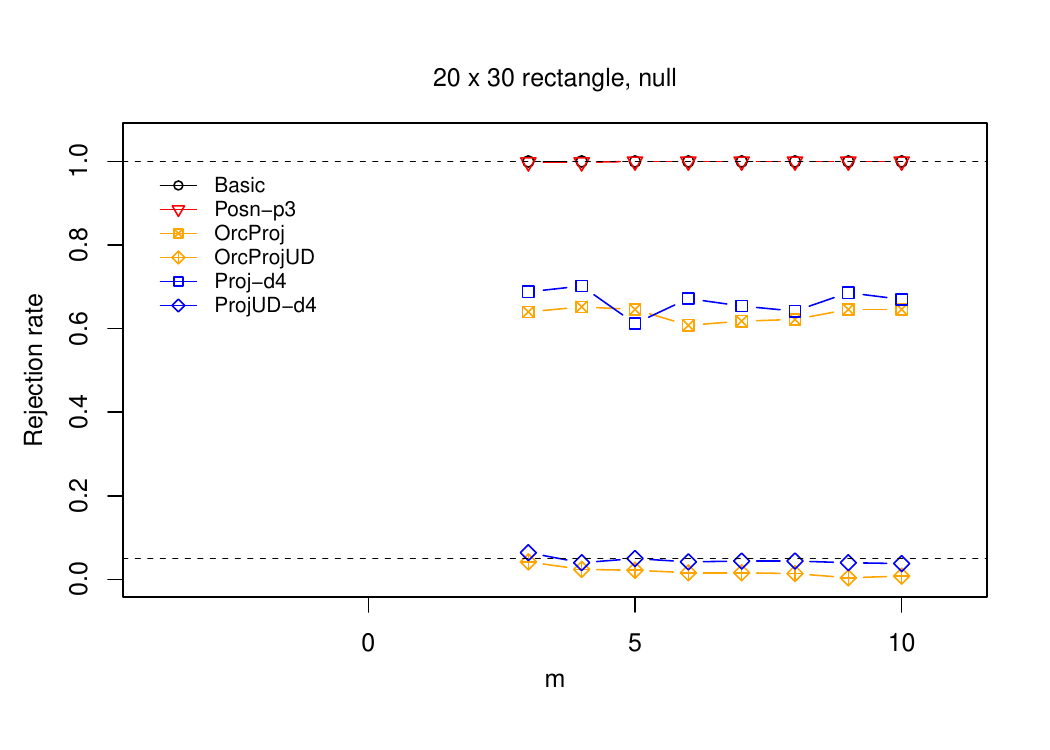}{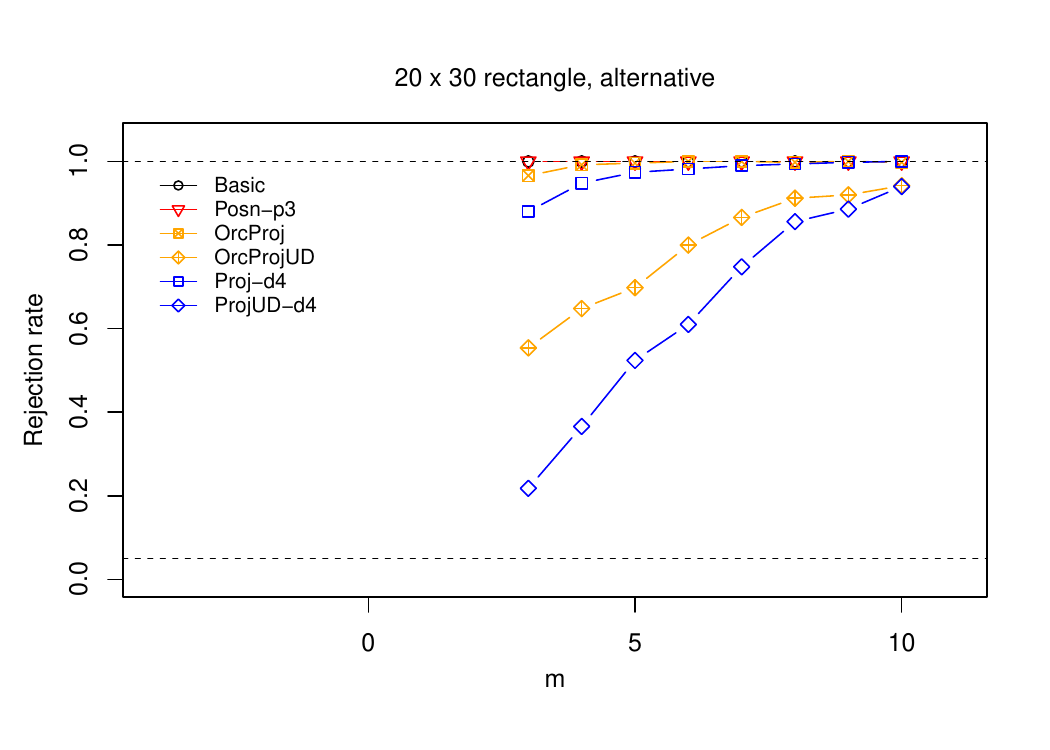}
\caption{Rejection rates for binary edge logistic link networks with inner product latent space model and overdispersion parameter $\eta=2$. Dashed lines correspond to $\alpha=0.05$ and $1$. 
Point colors correspond to $2p$ or $d$, with basic tests displayed in black and oracle projection tests in orange.} 
\label{logit2}
\end{figure}

In the left panel of Figure~\ref{logit2}, we see that other than the projection tests, all the methods nearly always reject the mesoscale null hypothesis, suggesting that accounting for overdispersion is  necessary to control type I error.  
The right panel of Figure~\ref{logit2} shows results under the alternative.  The modified projection tests have good power, with the power of the learned projection test approaching that of the oracle projection test as $m$ grows. 

\subsection{An application to neuroimaging data} \label{subsec:mri_data}

In this section, we apply our projection tests to mesoscale network testing for neuroimaging data, which was
originally analyzed by \citebody{badea17exploring}, and is publicly accessible in its preprocessed form \citepbody{xu23data}.
The data consists of functional magnetic resonance imaging (fMRI) brain images for $40$ subjects: $20$ healthy controls and $20$ with Parkinson's disease. There are $17$ female and $23$ male subjects, with mean age around $65$ years.
Each weighted, undirected edge corresponds to a Fisher-transformed (variance stabilized) correlation between activation patterns of two brain regions, parcellated according to a standard automatic anatomical labeling of 116 regions \citepbody[AAL116,][]{tzourio02automated}, which identifies the brain system these regions belong to.  
Research on Parkinson's disease has identified cognitive impairment in tasks related to the frontal lobe (FL, comprised of 28 regions / nodes) \citepbody{taylor86frontal}, as well as dysfunction in the cerebellum (CBM, 26 nodes) \citepbody{wu13cerebellum}.
We will test three mesoscale hypotheses: changes in the mean functional connectivity within these two lobes of interest (``FL/FL'', ``CBM/CBM''), as well as between them (``FL/CBM'').

We apply our mesoscale projection testing approach with learned projections for Gaussian edge networks (``Proj''), as well as the basic $F$-test (``Basic''), a projection test using uniform random orthonormal matrices from the Stiefel manifold (``RandProj''), and a test which compares the edge mean over the matrix block corresponding to each mesoscale test (``Block''), equivalent to the projection test with $1$-dimensional projection onto the vector with all entries proportional to $1$. To learn projections of a given dimension $d=d_r=d_c$, we will use submatrices of the eigenvectors of $\mathbb{T}$, as defined in \eqref{T_matrix}, where we hold out edges incident to either the FL or CBM regions. 
By symmetry, $\Theta_{\mathcal{C}}^{(g)} = \{\Theta_{\mathcal{R}}^{(g)}\}^{\tp}$ for $g=1 ,2$.
When we estimate $\mathbb{T}$, we replace each submatrix by its sample mean. 

Visual inspection of a scree plot 
leads us to select $d=6$ as the dimension of our learned projections.
However, for both learned and uniform random projections, we also evaluate sensitivity to the projection dimension, showing results for $d = 2, 6 ,10$. 
The projection testing results for this selected dimension $d=6$ are bolded in Table~\ref{tab:pval_table_mri}.
This sensitivity analysis is purely for illustrative purposes---in practice selection of $d$, like the choice of projections, should be independent of the edges in the hypothesis set to preserve test validity.
The results for each hypothesis test are reported in Table~\ref{tab:pval_table_mri}.
For the uniform random projections, we report the median $p$-value and rejection rate over $100$ replications.
\begin{table}
\begin{center}
\small
\begin{tabular}{|l|lll|}
\hline
             & \multicolumn{3}{c|}{$p$-value (or median $p$-value)}                                                          \\ \hline
             {\bf Test} & \multicolumn{1}{l|}{{\bf FL/FL}} & \multicolumn{1}{l|}{{\bf CBM/CBM}} & {\bf FL/CBM} \\ \hline
Basic        & \multicolumn{1}{l|}{0.453}  & \multicolumn{1}{l|}{$4.0 \times 10^{-3}$}  & $7.8 \times 10^{-5}$  \\ \hline
Proj ($d=2$)     & \multicolumn{1}{l|}{0.097} & \multicolumn{1}{l|}{0.016} & $6.5 \times 10^{-4}$    \\ \hline
{\bf Proj ($\bm{d=6}$)}     & \multicolumn{1}{l|}{{\bf 0.069}} & \multicolumn{1}{l|}{$\bm{3.1 \times 10^{-3}}$} & $\bm{1.1 \times 10^{-4}}$   \\ \hline
Proj ($d=10$)   & \multicolumn{1}{l|}{0.093} & \multicolumn{1}{l|}{$7.3 \times 10^{-3}$} & $6.9 \times 10^{-6}$     \\ \hline
RandProj ($d=2$)  & \multicolumn{1}{l|}{0.553}  & \multicolumn{1}{l|}{0.467}  & 0.429  \\ \hline
RandProj ($d=6$)  & \multicolumn{1}{l|}{0.348}  & \multicolumn{1}{l|}{0.210}  & 0.215   \\ \hline
RandProj ($d=10$) & \multicolumn{1}{l|}{0.541}  & \multicolumn{1}{l|}{0.151}  & 0.077  \\ \hline
Block & \multicolumn{1}{l|}{0.073}  & \multicolumn{1}{l|}{0.279}  & 0.215 \\ \hline
\end{tabular}
\caption{$p$-values for fMRI data mesoscale testing. For uniform random projections, the median $p$-value and rejection rate over $100$ replications is reported. Results for the heuristically selected projection dimension ($d=6$) are bolded.}
\label{tab:pval_table_mri}
\end{center}
\end{table}

Table~\ref{tab:pval_table_mri} shows that the basic $F$-test does not find the connections within the FL/FL hypothesis set to be significantly different across groups, while the differences in the other two rectangles are found to be highly significant, at the level $\alpha = 0.01$.
The mesoscale projection-based tests generally agree with these conclusions, but the $p$-value for the FL/FL set are much smaller, and significant at $\alpha= 0.1$.  The results are also not  highly sensitive to the selected projection dimension.
The random projections do not find anything significant, showing that the learned projections are indeed identifying differential structure.
The naive block pooling gives a comparable $p$-value to the learned projection tests for the FL/FL hypothesis set, but for the other two hypotheses (CBM/CBM and FL/CBM) it is not significant, and most of the signal appears to be lost in this crude summarization.
Our proposed mesoscale tests adaptively balance dimension reduction and signal aggregation, detecting differences in all three sets of connections, which agrees with previuos findings that these lobes show damage or dysfunction in subjects with Parkinson's disease \citepbody{taylor86frontal,wu13cerebellum}.

We additionally demonstrate the sensitivity of these testing approaches to the number of networks in each sample.
For $\tilde m=4, 8, 12, 16$, we form subsamples by selecting $\underline{m}$ networks without replacement from each of the two populations.
We then test the mesoscale null hypothesis for the CBM/CBM hypothesis set, which was rejected by both the basic and projection tests on the full data.
Both the learned and uniform random projection tests use dimension $d=6$.
This process was repeated $200$ times.  The results are summarized in Figure~\ref{fig:m_sensitivity_mri}, showing the learned projection test leads to a smaller $p$-value on average than the competing methods for all values of $\tilde{m}$, suggesting it is indeed increasing power by reducing the dimension of the test without removing much signal (since for these case-control tests, we would expect to reject the mesoscale hypotheses based on previous findings).  

\begin{figure}
  \begin{center}
  \includegraphics[width=.7\textwidth]{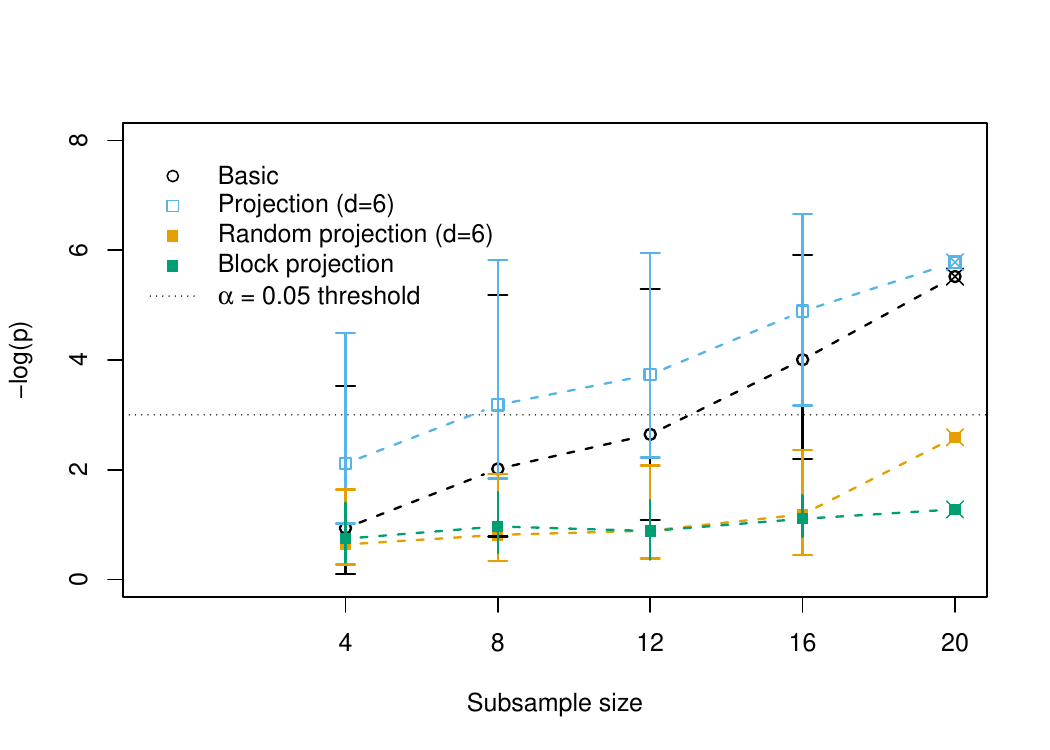}
  \caption{Median $-\log(p)$ for testing the CBM/CBM hypothesis in fRMI data as a function of sample size $m$. For $m < 20$, vertical bars span from the 25th to 75th empirical quantiles over $200$ replications. The horizontal dotted line corresponds to $\alpha=0.05$.}
  \label{fig:m_sensitivity_mri}
\end{center}
\end{figure}




\section{Discussion and future work} \label{sec:conclusion}

In this paper, we have developed new mesoscale two-sample tests for aligned networks, allowing for comparisons at local (individual node pair), intermediate (a subnetwork) and global (full network) scales. 
Testing at mesoscale allows the analyst to effectively balance power and localization, and we show that 
 in the presence of low rank or low effective rank structure, which has been commonly observed in networks, information from outside the mesoscale hypothesis edge set can used to produce powerful tests that are robust to structural misspecification.
\ifjasa
\else
We prove finite sample type I error control and exact power calculations for Gaussian edge networks, and propose methodology for networks with general exponential family edges, with proven asymptotic type I error control.
\fi


In Appendix~\ref{subsec:fdr}, we discuss further modifications to methodology and theory for simultaneous mesoscale testing of several hypothesis sets.
Extensions of interest include adjustments for edge dependence or heteroscedasticity, perhaps motivated by nested model comparison procedures for generalized linear mixed models with more flexible covariance structures.
\ifjasa
\else
There is also potential to work with more general non-linear projected null hypotheses, essentially extending the use of linear subspaces to submanifolds of $\mathbb{R}^r$ and $\mathbb{R}^c$.
Finally,  it would be valuable to tailor our mesoscale hypothesis testing methodology to robustly test for equality of positions in latent position models, perhaps by developing a procedure to adaptively learn a suitable set of ``null'' nodes to align the network embeddings.
\fi

The setting of multiple matrix-valued observations split into two samples is not unique to networks.   
For example, one could observe a matrix of counts (e.g., of microbiotic taxa), or subject-specific covariance or precision matrix estimates \citepbody{xia18multiple}.
Many models in those settings also reduce the intrinsic dimension of data using low rank or latent variable models, and thus our methodology would be applicable to two-sample testing for fixed sub-rectangles in general matrix data, as long as suitable assumptions can be worked out.


\bibliographystylebody{abbrvnat}
\bibliographybody{graph_testing0}

\pagebreak

\begin{center}
\section*{Supplementary materials for ``Mesoscale two-sample testing for networks''}

\ifblind
\else
\bigskip

Peter W. MacDonald, University of Waterloo

\smallskip

Elizaveta Levina, University of Michigan

\smallskip

Ji Zhu, University of Michigan
\fi

\bigskip

\today

\bigskip
\end{center}

\appendix

\section{Methodological extensions}

\subsection{Mesoscale testing with nuisance covariates} \label{app:nuisance}

Here we generalize the tests in Sections~\ref{subsubsec:gaussian_test1} and \ref{subsubsec:gaussian_test2} to the setting with network-level nuisance covariates. For example, in a neuroimaging data set like the one we analyzed in Section~\ref{subsec:mri_data}, nuisance covariates may include age, sex, and head motion in the scanner.  
 
As in Section~\ref{sec:testing}, for simplicity we focus on a rectangular hypothesis set $\mathcal{S}$ with $r$ rows and $c$ columns. In contrast to Section~\ref{sec:testing}, we do not require balanced groups here, and instead assume we observe $m_1 + m_2 = M$ networks with networks $1, \dots, m_1$ from group 1 and $m_1+1, \dots M$ networks from group 2.   Let $g(k) \in \{1,2\}$ for $k=1,\ldots,M,$ denote the group membership of network $k$.   Our goal is construct an efficient mesoscale test for
$$
  H_0: \Theta_{ij}^{(1)} = \Theta_{ij}^{(2)} \quad \forall (ij) \in \mathcal{S},
$$
after projection by an arbitrary orthonormal matrix $(V \otimes U) \in \real^{rc \times d_rd_c}$. 
As in Section~\ref{sec:testing}, let
$
  Y = \left[ 
  \vecz([A_1]_{\mathcal{S}}) \  \cdots  \   \vecz([A_{M}]_{\mathcal{S}})
\right] \in \real^{rc \times M}
$
be a matrix made up of the collected edges in the hypothesis set. 

Suppose that, for each $k$ we also observe a $q$-dimensional vector of nuisance covariates $\bm{z}_k \in \real^q$ for $k=1,\ldots,M$, and that each entry of $Y$ follows an independent normal distribution with
\begin{equation} \label{gaussian_model_nuisance}
  Y_{(ij),k} \ind \mathcal{N}\left( \Theta_{ij}^{(2)} + (\Theta_{ij}^{(1)} - \Theta_{ij}^{(2)})\mathbb{I}(g(k)=1) + \beta_{ij}^{\tp} \bm{z}_k  \, , \,  \sigma^2 \right).
\end{equation}
In the usual two-sample testing setting, we have no nuisance covariates, and $q=0$. Each row of $Y$ (i.e., ~the edge variables for a fixed node pair) gives $M$ observations from a Gaussian linear regression model with an intercept and $q+1$ covariates, indexed by $k$. All rows share the same design matrix, but have their own coefficient vectors.

To simplify notation, let 
\begin{align}
  \bm{Z}
= \begin{pmatrix}
    \bm{1}_{m_1} & \bm{1}_{m_1} & \bm{z}_1^{\tp} \\ ~ & ~ & \vdots \\ \bm{1}_{m_2} & \bm{0}_{m_2} & \bm{z}_M^{\tp}
\end{pmatrix} \in \real^{M \times (q+2)}, \label{X_nuisance} \ \ \   
\mathcal{B} = \begin{pmatrix}
  \vdots & \vdots & \vdots \\ \Theta_{ij}^{(2)} & \Theta_{ij}^{(1)} - \Theta_{ij}^{(2)} & \beta_{ij}^{\tp} \\ \vdots & \vdots & \vdots
\end{pmatrix} \in \real^{rc \times (q+2)} 
\end{align}
where the rows of $\mathcal{B}$ correspond to pairs $(i,j) \in \mathcal{S}$.   
With this notation, we can define generalized test statistics, analogous to the ones defined in Sections~\ref{subsubsec:gaussian_test1} and \ref{subsubsec:gaussian_test2}, and state their exact distributions.
The two propositions which follow generalize Propositions~\ref{prop:fobs_tilde} and \ref{prop:fobs}. 
The proofs of these two results immediately follow their statements.

Analogous to Section~\ref{subsubsec:gaussian_test1}, to define a test statistic, we must find an independent estimator of $\sigma^2$ with known non-central $\chi^2$ distribution after rescaling.
The discussion in Section~\ref{subsubsec:gaussian_test1} motivates a test statistic with a doubly non-central F distribution under any choice of edge expectation parameters.

\begin{proposition} \label{prop:fobs_tilde_nuisance}
  Suppose $Y$ follows the model defined in \eqref{gaussian_model_nuisance}. 
  Suppose $U \in \real^{r \times d_r}$ and $V \in \real^{c \times d_c}$ are orthonormal matrices, and $d_rd_c < rc$. 
  Define $\bm{Z}$ as in \eqref{X_nuisance} and let
  \begin{align}
  Q_{\perp}(U,V) &= I_{rc} - (V \otimes U)(V \otimes U)^{\tp} \label{q_nuisance}, \\
  h_2 &= \bm{e}_2^{\tp}(\bm{Z}^{\tp}\bm{Z})^{-1}\bm{e}_2, \nonumber
  \end{align}
  where $\bm{e}_2 = (0,1,0,\ldots,0)^{\tp} \in \real^{q+2}$. Then
  \[
  w^{(\mathrm{G})} = \frac{\nu_2 \lVert (V \otimes U)^{\tp} Y \bm{Z}(\bm{Z}^{\tp}\bm{Z})^{-1}\bm{e}_2 \rVert_2^2}{\nu_1 \lVert Q_{\perp}(U,V) Y \bm{Z}(\bm{Z}^{\tp}\bm{Z})^{-1}\bm{e}_2 \rVert_2^2}
\]
  has a doubly non-central $F_{\nu_1,\nu_2}(\psi,\zeta)$ distribution \citepsupp{bulgren71representations}, with   $\nu_1 = d_rd_c$, $\nu_2 = d_rd_c(M - q - 2)$, and 
  \begin{align*}
    \psi = \psi(U,V) &= \frac{1}{\sigma^2 h_2} \lVert (V \otimes U)^{\tp} Y \bm{Z}(\bm{Z}^{\tp}\bm{Z})^{-1}\bm{e}_2 \rVert_2^2 , \\
    \zeta = \zeta(U,V) &= \frac{1}{2\sigma^2 h_2} \lVert Q_{\perp}(U,V) \vecz(\Theta^{(1)}_{\mathcal{S}} - \Theta^{(2)}_{\mathcal{S}}) \rVert_2^2.
  \end{align*}
\end{proposition}

\begin{proof}[Proof of Proposition~\ref{prop:fobs_tilde_nuisance}]
   Following the notation in \eqref{X_nuisance}, $Y$ has a matrix normal distribution,
$
  Y \sim \mathcal{MN}_{rc \times M}\left( \mathcal{B} \bm{Z}^{\tp} ~,~ \sigma^2 (I_{rc} \otimes I_{M}) \right).
$
After projection onto the coordinates of $(V \otimes U) \in \real^{rc \times d_rd_c}$, and using a standard linear regression estimator for each row-wise regression model, we have
\[
  (V \otimes U)^{\tp} Y \bm{Z}(\bm{Z}^{\tp}\bm{Z})^{-1} \sim \mathcal{MN}_{d_rd_c \times (q+2)}\left( (V \otimes U)^{\tp} \mathcal{B} ~,~ \sigma^2 (I_{d_rd_c} \otimes (\bm{Z}^{\tp}\bm{Z})^{-1}) \right).
\]
Note that the second column of $\mathcal{B}$ is $\vecz(\Theta^{(2)}_{\mathcal{S}} - \Theta^{(1)}_{\mathcal{S}})$. 
Thus,
\[
  (V \otimes U)^{\tp} Y \bm{Z}(\bm{Z}^{\tp}\bm{Z})^{-1}\bm{e}_2 \sim \mathcal{N}_{d_rd_c}\left( (V \otimes U)^{\tp} \vecz(\Theta^{(1)}_{\mathcal{S}} - \Theta^{(2)}_{\mathcal{S}}) ~,~ \sigma^2 h_2 I_{d_rd_c} \right),
\]
and it follows that
\[
  \frac{1}{\sigma^2 h_2} \lVert (V \otimes U)^{\tp} Y \bm{Z}(\bm{Z}^{\tp}\bm{Z})^{-1}\bm{e}_2 \rVert_2^2
\]
has a non-central $\chi^2$ distribution with $d_rd_c$ degrees of freedom, and non-centrality parameter
\[
  \psi(U,V) = \frac{1}{2\sigma^2h_2}\lVert (V \otimes U)^{\tp} \vecz(\Theta^{(1)}_{\mathcal{S}} - \Theta^{(2)}_{\mathcal{S}}) \rVert_2^2 = \frac{1}{2\sigma^2h_2} \lVert U^{\tp}(\Theta^{(1)}_{\mathcal{S}} - \Theta^{(2)}_{\mathcal{S}})V \rVert_F^2.
\]
Note that $Q_{\perp}(U,V)$ is an idempotent projection matrix with rank $rc - d_rd_c$. Then
\[
  Q_{\perp}(U,V) Y \bm{Z}(\bm{Z}^{\tp}\bm{Z})^{-1}\bm{e}_2 \sim \mathcal{N}_{rc} \left( Q_{\perp} \mathcal{B} \bm{e}_2 ~,~ \sigma^2 h_2 Q_{\perp} \right),
\]
and it follows that
\[
  \frac{1}{\sigma^2 h_2} \lVert Q_{\perp}(U,V) Y \bm{Z}(\bm{Z}^{\tp}\bm{Z})^{-1}\bm{e}_2 \rVert_2^2
\]
has a non-central $\chi^2$ distribution with $(rc - d_rd_c)$ degrees of freedom, and non-centrality parameter
\[
  \zeta(U,V) = \frac{1}{2\sigma^2 h_2} \lVert Q_{\perp}(U,V) \vecz(\Theta^{(1)}_{\mathcal{S}} - \Theta^{(2)}_{\mathcal{S}}) \rVert_2^2.
\]
It is independent of $(V \otimes U)^{\tp} Y \bm{Z}(\bm{Z}^{\tp}\bm{Z})^{-1}\bm{e}_2$ since the two matrices which left-multiply $Y$ are mutually orthogonal. This completes the proof by the definition of the doubly non-central $F$ distribution.
\end{proof}

If $M > q+2$, then similar to the construction of Section~\ref{subsubsec:gaussian_test2}, we can find an alternative test statistic based on an independent estimator of $\sigma^2$ which uses only the projected data.

\begin{proposition} \label{prop:fobs_nuisance}
  Suppose $Y$ follows the model defined in \eqref{gaussian_model_nuisance}, and $M > q+2$. 
  Suppose $U \in \real^{r \times d_r}$ and $V \in \real^{c \times d_c}$ are orthonormal matrices. 
  Define $\bm{Z}$ as in \eqref{X_nuisance}, $h_2$ and $\bm{e}_2$ as in Proposition~\ref{prop:fobs_tilde_nuisance}, and let
  \begin{equation}
  \bm{H}_Z^{\perp} = I_M - \bm{Z} (\bm{Z}^{\tp}\bm{Z})^{-1} \bm{Z}^{\tp}. \label{hat_nuisance}
  \end{equation}
  Then
  \[
  w^{(\mathrm{G-P})} = \frac{\nu'_2 \lVert (V \otimes U)^{\tp} Y \bm{Z}(\bm{Z}^{\tp}\bm{Z})^{-1}\bm{e}_2 \rVert_2^2}{\nu_1 h_2 \lVert (V \otimes U)^{\tp} Y \bm{H}_Z^{\perp} \rVert_F^2}
\]
  has a non-central $F_{\nu_1,\nu'_2}(\psi)$ distribution with
  \begin{equation*}
    \nu_1 = d_rd_c, \quad \nu'_2 = d_rd_c(M - q - 2),
  \end{equation*}
  and $\psi = \psi(U,V)$ defined as in Proposition~\ref{prop:fobs_tilde_nuisance}.
\end{proposition}

\begin{proof}[Proof of Proposition~\ref{prop:fobs_nuisance}]
    The proof proceeds similarly to the proof of Proposition~\ref{prop:fobs_tilde_nuisance} to show that 
    \[
  \frac{1}{\sigma^2 h_2} \lVert (V \otimes U)^{\tp} Y \bm{Z}(\bm{Z}^{\tp}\bm{Z})^{-1}\bm{e}_2 \rVert_2^2
    \]
has a non-central $\chi^2$ distribution with $d_rd_c$ degrees of freedom, and non-centrality parameter $\psi(U,V)$.

Note that $\bm{H}_Z^{\perp}$ is an idempotent projection matrix with rank $M - q - 2$.
Then
\[
  (V \otimes U)^{\tp} Y \bm{H}_Z^{\perp} \sim \mathcal{MN}_{d_rd_c \times M}\left( \bm{0} ~,~ \sigma^2 (I_{d_rd_c} \otimes  \bm{H}_Z^{\perp}) \right),
\]
and by idempotency,
\[
  \frac{1}{\sigma^2} \lVert (V \otimes U)^{\tp} Y \bm{H}_Z^{\perp} \rVert_F^2 \sim \chi^2_{d_rd_c(M - (q+2))}.
\]
It is independent of $(V \otimes U)^{\tp} Y \bm{Z}(\bm{Z}^{\tp}\bm{Z})^{-1}$ since the two matrices which right-multiply $Y$ are mutually orthogonal. This completes the proof by the definition of the  non-central $F$ distribution.
\end{proof}

We apply this mesoscale testing with nuisance correction on the fMRI data from Section~\ref{subsec:mri_data}, adjusting for age and sex as nuisance covariates.
We test for differences in the regression effect of Parkinson's status for the same three hypothesis sets as in Section~\ref{subsec:mri_data}: node pairs within the frontal lobe (FL/FL), node pairs within the cerebellum (CBM/CBM), and node pairs between these two lobes (FL/CBM).
We compare our nuisance-corrected projection test with learned projection to a basic test which aggregates test statistics for the edgewise regression effects (``Basic''), and test based the difference in mean estimated regression effects over the matrix block (``Block'').
As in the original uncorrected data, analysis of a scree plot based on the held out edges suggests a learned projection dimension of $d=6$ is appropriate.

\begin{table}
\begin{center}
\small
\begin{tabular}{|l|lll|}
\hline
             & \multicolumn{3}{c|}{$p$-value}                                                          \\ \hline
             {\bf Nuisance-corrected test} & \multicolumn{1}{l|}{{\bf FL/FL}} & \multicolumn{1}{l|}{{\bf CBM/CBM}} & {\bf FL/CBM} \\ \hline
Basic        & \multicolumn{1}{l|}{0.204}  & \multicolumn{1}{l|}{$1.5 \times 10^{-3}$}  & $1.1 \times 10^{-6}$  \\ \hline
Proj ($d=6$)    & \multicolumn{1}{l|}{0.052} & \multicolumn{1}{l|}{$0.013$} & $4.5 \times 10^{-6}$   \\ \hline
Block & \multicolumn{1}{l|}{0.051}  & \multicolumn{1}{l|}{0.264}  & 0.131 \\ \hline
\end{tabular}
\caption{$p$-values for fMRI data mesoscale testing, adjusting for age and sex. Learned projection results (``Proj'') are shown for the heuristically selected projection dimension $d=6$.}
\label{tab:pval_table_nuisance}
\end{center}
\end{table}

Results in Table~\ref{tab:pval_table_nuisance} differ from the corresponding rows of Table~\ref{tab:pval_table_mri}, but the same rejection decisions would be made at the $\alpha=0.05$ level. This is unsurprising, as in this dataset, the age and sex covariates are approximately balanced between the case and control groups.
The rejection decisions made by these nuisance-corrected tests let us conclude that the CBM/CBM and FL/CBM blocks are both significantly different between Parkinson's patients and healthy controls, after adjusting for age and sex. 

\ifjasa

\else
\fi

\subsection{Non-rectangular or undirected hypothesis sets} \label{subsec:nonrect}

One interpretation of the tests in Section~\ref{sec:testing} is as the comparison of nested GLMs.
When the hypothesis node pairs form a rectangle, the left and right projections are easy to write compactly with matrices, but the underlying idea is still compatible with non-rectangular hypotheses.
For projections with dimensions $d_r$ and $d_c$, the edge expectations under the corresponding projection null will have at most $d_rd_c$ free parameters, and thus whether or not it is rectangular, whenever $\lvert \mathcal{S} \rvert \geq d_rd_c$, projection will reduce the test dimension.

In the general non-rectangular case, we require left and right subspaces of $\mathbb{R}^n$, with orthonormal bases $\bar{U} \in \mathbb{R}^{n \times d_r}$ and $\bar{V} \in \mathbb{R}^{n \times d_c}$.
These subspaces, along with $\mathcal{S}$, will define a subspace of $\mathbb{R}^{\lvert \mathcal{S} \rvert}$.
We will restrict our estimates of the (vectorized) group means to this low dimensional subspace.

Let $\mathcal{P}_{\mathcal{M}}$ denote the $p \times p$ matrix of the projection operator onto a subspace $\mathcal{M} \subseteq \mathbb{R}^p$,
and for a non-empty subset $\mathcal{N} \subseteq \{1,\ldots,p\}$, let $\pi_{\mathcal{N}}$ denote the $\lvert \mathcal{N} \rvert \times p$ matrix made up of the rows of $I_p$ which correspond to the elements in $\mathcal{N}$.

Given left and right-hand side projections, we can write a projected null in the following general form
\begin{equation} \label{working_assump_general}
  \vecz(\Theta_{\mathcal{S}}^{(g)}) \in \pi_{\mathcal{S}}(\bar{V} \otimes \bar{U} )
\end{equation}
for $g=1,2$.

Assuming $ \pi_{\mathcal{S}} (\bar{V} \otimes \bar{U} )$ has full rank, the corresponding projection-based tests are exactly the ones in the previous sections; motivated by the comparison of groupwise (generalized) linear models with coefficients $\gamma_1, \gamma_2 \in \mathbb{R}^{d_rd_c}$, the means are given by
$
    \left\{ \pi_{\mathcal{S}} (\bar{V} \otimes \bar{U} )\right\} (\gamma_1 + \gamma_2)
$
for group $1$,
$
    \left\{ \pi_{\mathcal{S}} (\bar{V} \otimes \bar{U} )\right\} (\gamma_1 - \gamma_2)
$
for group $2$, and $\gamma_2 = \bm{0}_{d_rd_c}$ under the mesoscale null hypothesis.

This extension to general hypothesis sets also suggests a modification of the above tests for undirected networks, and networks with no self loops.
For networks with no self loops, we can simply remove diagonal entries from $\mathcal{S}$.
For undirected networks, we must discard dependent edge variables below the diagonal from $\mathcal{S}$, but the methodological modification takes slightly more care, since symmetry can also induce rank deficiency in the projected data (for the Gaussian edge model) or GLM design matrix (for general exponential family edges).
As a simple example, suppose the hypothesis set is $\mathcal{S} = \mathcal{T} \times \mathcal{T}$ in an undirected network.  Denoting both the left and right projections by $V$, the projected data takes the form 
$   V^{\tp} [A^{(g)}_k]_{\mathcal{S}} V \in \mathbb{R}^{d \times d} $
for $g=1,2$ and $k=1,\ldots,m$ will also be symmetric, and thus restricted to a subspace of $\mathbb{R}^{d \times d}$, where $d$ is the column dimension of $V$.

In the undirected case, the hypothesis set must contain the symmetric counterpart of all of its node pairs, and the learned left and right projections have to be the same ($\bar{V} = \bar{U}$).
For given projections, the mesoscale test operates under a working GLM, with mean structure
\begin{equation} \label{working_glm}
  \mathbb{E}\left\{ \pi_{\mathcal{S}} \operatorname{vec}\left( A_k^{(g)} \right) \right\} = h\{ \pi_{\mathcal{S}}(\bar{U} \otimes \bar{U}) \beta^{(g)} \}.
\end{equation}

Let $\pi^{\triangle}_{\mathcal{S}}$ denote the operator which restricts to only node pairs on or above the diagonal, and suppose $G_{\mathcal{S}}$ is the linear mapping which satisfies
\begin{equation*}
  \pi_{\mathcal{S}} \operatorname{vec}(M) = G_{\mathcal{S}} \pi^{\triangle}_{\mathcal{S}} \operatorname{vec}(M)
\end{equation*}
for any $n \times n$ matrix $M$.
By construction, $G_{\mathcal{S}}$ has a (left) pseudo-inverse $G_{\mathcal{S}}^{\dagger}$, and \eqref{working_glm} can be rewritten as
\begin{equation} \label{working_glm2}
  \mathbb{E}\left\{ \pi^{\triangle}_{\mathcal{S}} \operatorname{vec}\left( A_k^{(g)} \right) \right\}  = G_{\mathcal{S}}^{\dagger} h\{ \pi_{\mathcal{S}}(\bar{U} \otimes \bar{U}) \beta^{(g)} \}.
\end{equation}
In words, application of the linear map $G_{\mathcal{S}}^{\dagger}$ on the right-hand side of \eqref{working_glm2} replaces the edge expectation of a given node pair by the average of its expectation and that of its symmetric counterpart. Then by injectivity of $h$ we can move $G_{\mathcal{S}}^{\dagger}$ inside the link function and reparametrize
\begin{equation*}
    \mathbb{E}\left\{ \pi^{\triangle}_{\mathcal{S}} \operatorname{vec}\left( A_k^{(g)} \right) \right\} = h( \tilde{W} \tilde{\beta}^{(g)}),
\end{equation*}
where $\tilde{W}$ is an orthonormal basis for $\operatorname{col} \left\{ G_{\mathcal{S}}^{\dagger} \pi_{\mathcal{S}}(\bar{U} \otimes \bar{U}) \right\}$.
In general,
$
  \operatorname{col} \left\{ G_{\mathcal{S}}^{\dagger} \pi_{\mathcal{S}}(\bar{U} \otimes \bar{U}) \right\} $
will not coincide with
$
  \operatorname{col} \left\{ \pi^{\triangle}_{\mathcal{S}}(\bar{U} \otimes \bar{U}) \right\},
$ the ``naive'' choice of subspace under undirectedness, and they may even have different dimensions.

In summary, for undirected networks, we augment the general working assumption in \eqref{working_assump_general} as
\begin{equation*}
  \pi^{\triangle}_{\mathcal{S}} \vecz(\Theta^{(g)}) \in \operatorname{col}\left\{ G_{\mathcal{S}}^{\dagger }\pi_{\mathcal{S}} (\bar{U} \otimes \bar{U} )\right\}
\end{equation*}
for $g=1,2$.


\subsection{Density and degree correction with restricted projections} \label{subsec:densitydegree}

We previously noted that the projected null hypothesis is conservative, and restricts the class of alternatives for which the projection test has non-trivial power.  
This same property of projection testing can be used to design modified tests which ignore certain types of differences between the edge expectation parameter matrices.
Here we describe two such modifications which are of particular interest in applications, the density-corrected and the degree-corrected two-sample mesoscale tests (which is their interpretation for binary edge networks).
As above, we restrict to the case of rectangular hypothesis sets for simplicity, but these modifications apply in the non-rectangular case as well (see Section~\ref{subsec:nonrect}).

In the density-corrected case, note that under the model \eqref{expo_edges}, two-sample differences in edge density will correspond to differences in $\vecz(\Theta^{(1)}_{\mathcal{S}})$ and $\vecz(\Theta^{(2)}_{\mathcal{S}})$ in a one-dimensional subspace
$
  \operatorname{span}(\bm{1}_c \otimes \bm{1}_r).
$
Thus the density corrected projection test will center the left and right projections $U$ and $V$, restricting them to satisfy
$$
  (V \otimes U)^{\tp} (\bm{1}_c \otimes \bm{1}_r) = \bm{0}_{d_rd_R} , 
$$
and therefore changes in the sum of the entries of $\Theta_{\mathcal{S}}^{(1)}$ and $\Theta_{\mathcal{S}}^{(1)}$ will be ignored by the corresponding projection test.
Practically, while this centering step can be applied to $U$ and $V$ estimated as in Section~\ref{subsec:proj_learning}, it is helpful to modify the projection learning method, applying imputation or one-step estimation to a centered estimate of the difference in edge expectation parameters.
The result is an improved ordering of the projection directions tailored to density corrected testing, as the learned orthonormal matrix is already approximately centered.

Under model \eqref{expo_edges}, degree corrections are typically modeled with additive row and column fixed or random effects $\delta_i^{(g)}$ for $g=1,2$ and $i=1,\ldots,n$, such that
$$
  \mathbb{E}[A_k^{(g)}]_{ij} = h\left( \delta_i^{(g)} + \delta_j^{(g)} + \Theta_{ij}^{(g)} \right)
$$
for $g=1,2$, $k=1,\ldots,m$, and $1 \leq i,j \leq n$, with additional centering assumptions on $\Theta^{(g)}$ \citepbody[e.g.,][]{ma20universal}.
Thus, restricted to the hypothesis set, two-sample differences in these fixed effect parameters will correspond to a linear subspace,
$$
  \operatorname{col}\begin{pmatrix} \bm{1}_c \otimes I_r & \bm{1}_r \otimes I_c \end{pmatrix},
$$
which has $(r + c - 1)$ linearly independent directions.
Just as for density correction, a degree corrected projection test could restrict the left and right projections so that $\operatorname{col}(V \otimes U)$ is orthogonal to changes in the degree effects, analogous to allowing sample-specific row and column fixed effect terms under the null.

\subsection{Multiple testing and false discovery rates} \label{subsec:fdr}

In many applications, we may be interested in simultaneous multiple testing for several non-overlapping hypothesis sets $\mathcal{S}_1,\ldots,\mathcal{S}_K$, with guarantees on an error metric such as the false discovery rate (FDR).
For instance, we may want to test for differences between all the blocks induced by a collection of known node communities, or rows corresponding to multiple nodes of interest.
Under the same independent edge assumption made in model \eqref{expo_edges}, testing with fixed projections will render the $p$-values independent, and thus standard thresholding approaches
control the FDR at the nominal level \citepsupp{storey04strong}.

However, our methodology for mesoscale testing with learned projections will lead to dependent test statistics, as the edges in $\mathcal{S}_1$ will be used directly for testing $\mathcal{H}_{0,\mathcal{S}_1}$, and indirectly for learning projections to test $\mathcal{H}_{0,\mathcal{S}_j}$ for $j=2,\ldots,K$.   Some simple alterations to the data splitting strategy can be used to produce independent $p$-values.
First, if $\mathcal{S}_1 \cup \cdots \cup \mathcal{S}_K \neq \{1,\ldots,n\}^2$, we may use the remaining edges to learn projections for all the tests.
A conditioning argument, integrating over these remaining edge variables, shows that standard thresholding approaches for independent $p$-values will still control the FDR.
If this first strategy is not feasible, but $m > 1$, we may instead do an initial layer of data splitting, using a subset of the networks to learn projections, and the held out set to test the mesoscale hypotheses.
The same conditioning argument justifies control of the FDR at the nominal level.

In future work, we will consider both empirical and theoretical analysis of the effect of dependence among multiple mesoscale hypothesis test statistics, both when edges are re-used, directly and indirectly, for different hypotheses, and when the hypothesis sets are overlapping.
While this will render the test statistics dependent, it is unclear if and how this dependence will affect FDR control.

\subsection{Relationship to latent position testing} \label{subsubsec:positions}

In this section, we will expand on the relationship between mesoscale network testing and testing for equality of individual node positions under a latent space model.
As is typical of two-sample testing under latent space models, in this section we will assume that the expected adjacency matrices are symmetric.

A general latent space model assumes that the expected adjacency matrices for each sample are parameterized by $d$-dimensional latent positions $Z_i^{(g)}$ for $i=1,\ldots,n$ and $g=1,2$, and specifies
\begin{equation} \label{general_lsm}
  [\Theta^{(g)}]_{ij} = D\left(Z_i^{(g)},Z_j^{(g)}\right)
\end{equation}
for a similarity function $D: \mathbb{R}^d \times \mathbb{R}^d \rightarrow \mathbb{R}$.
For instance, $D$ may be the (generalized) inner product or a linear function of the Euclidean distance.
This model is typically identifiable only up to some group of invariant transformations; 
for instance, when $D$ is the inner product, the model is invariant under orthogonal transformations of latent positions.

Under this general latent space model, we may formally state a {\em positional} null hypothesis
$$
  \mathcal{H}^{(\mathrm{posn})}_{0,i}: Z_i^{(1)} = Z_i^{(2)},
$$
to test the equality of node $i$'s latent position in two samples.
The positional null has both pros and cons.
It can be of scientific value to isolate changes in the behavior of a given node $i$.
On the other hand, it can be difficult to interpret since it depends on latent variables which are not fully identified.   Using terminology introduced by \citesupp{breusch86hypothesis}, model unidentifiability renders $\mathcal{H}^{(\mathrm{posn})}_{0,i}$ {\em untestable} for any $i=1,\ldots,n$.
In particular, it is not {\em confirmable}: for any parameter set where the positional null is true, there is an equivalent parameter set where it is false.
While in some cases $\mathcal{H}^{(\mathrm{posn})}_{0,i}$ may be {\em refutable} (with an inner product similarity we can distinguish between latent positions with different Euclidean norms), in others it is completely vacuous (with a Euclidean distance similarity, the latent space model is invariant to translation of the latent positions).

In recent work, \citesupp{du21hypothesis} propose a test for the positional null hypothesis under the generalized RDPG by making an additional identifying assumption.
They 
assume a known ``seed'' set of nodes $\mathcal{J} \subseteq \{1,\ldots,n\}$ such that the positional null hypothesis holds for all $j \in \mathcal{J}$, and the map
\begin{equation} \label{posn_ident}
  z \mapsto \left\{ D(z,Z_j^{(1)}) \right\}_{j \in \mathcal{J}}
\end{equation}
is injective.
In their specific setting with generalized inner product similarity, this condition requires that the $\lvert \mathcal{J} \rvert \times d$ matrix of seed node positions has linearly independent columns.

The condition in \eqref{posn_ident} renders the positional null hypothesis testable, and in fact makes it equivalent to a mesoscale null hypothesis.
In particular, under the general latent space model \eqref{general_lsm}, and if \eqref{posn_ident} holds for a seed set $\mathcal{J}$, $\mathcal{H}^{(\mathrm{posn})}_{0,i}$ is equivalent to the mesoscale null hypothesis $\mathcal{H}_{0,\mathcal{S}^{(\mathcal{J})}_i}$, where
$
  \mathcal{S}^{(\mathcal{J})}_i = \{(i,j) : j \in \mathcal{J} \}.
$

%
%

While \eqref{posn_ident} makes $\mathcal{H}^{(\mathrm{posn})}_{0,i}$ testable in principle, any practical test, including the one used in \citesupp{du21hypothesis}, depends on the correct specification of $\mathcal{J}$, and validity of the test cannot be guaranteed otherwise.
On the other hand, the corresponding mesoscale hypothesis based on the node pairs in $\mathcal{S}^{(\mathcal{J})}_i$ is defined in terms of identified parameters, and although it may no longer coincide with the positional null, it remains  a well-defined mesoscale null hypothesis, even when $\mathcal{J}$ is misspecified.
In summary, while our mesoscale testing methodology does not specifically address the positional null, it can be used to more robustly test closely related hypotheses about fully identified model parameters.

\section{Technical proofs} \label{app:proofs}

\subsection{Proofs of Propositions~\ref{prop:onestep_weak} and \ref{prop:onestep_strong}}


\begin{proof}[Proof of Proposition~\ref{prop:onestep_weak}]
  Recall the definition of $\mathbb{T}$,
  \begin{equation*}
    \mathbb{T} = (\Theta^{(1)}_{\mathcal{C}} - \Theta^{(2)}_{\mathcal{C}})(\Theta^{(1)}_{\mathcal{D}} - \Theta^{(2)}_{\mathcal{D}})^{\dagger}(\Theta^{(1)}_{\mathcal{R}} - \Theta^{(2)}_{\mathcal{R}}) \in \mathbb{R}^{r \times c}.
  \end{equation*}
  From \eqref{block_svd}, we expand
  \begin{equation*}
    \mathbb{T} = \left( U_{[r]}SV_{-[c]}^{\tp} \right) \left( U_{-[r]}SV_{-[c]}^{\tp} \right)^{\dagger} \left( U_{-[r]}SV_{[c]}^{\tp} \right).
  \end{equation*}
  Then by the definition of the Moore-Penrose pseudo-inverse,
  \begin{equation} \label{t_expanded}
    \mathbb{T} = U_{[r]}S \left( V_{-[c]}^{\tp} V_{-[c]} \right) S^{-1} \left( U_{-[r]}^{\tp} U_{-[r]} \right) SV_{[c]}^{\tp},
  \end{equation}
  where $S^{-1}$ exists since $\Theta^{(1)} - \Theta^{(2)}$ has rank $d_*$.

  By assumption, $V_{-[c]}^{\tp} V_{-[c]}$ and $U_{-[r]}^{\tp} U_{-[r]}$ are full rank, so we can write
  \begin{align*}
    \Theta_{\mathcal{S}}^{(1)} - \Theta_{\mathcal{S}}^{(2)} &= U_{[r]}S V_{[c]}^{\tp} \\
    \mathbb{T} &= U_{[r]} \left\{ S \left( V_{-[c]}^{\tp} V_{-[c]} \right) S^{-1} \left( U_{-[r]}^{\tp} U_{-[r]} \right) S \right\} V_{[c]}^{\tp},
  \end{align*}
  and both $S$ and $ \{ S ( V_{-[c]}^{\tp} V_{-[c]} ) S^{-1} ( U_{-[r]}^{\tp} U_{-[r]} ) S \}$ are full rank.
  It follows that the $d_{\mathcal{S}}$-dimensional column and row spaces of $\Theta_{\mathcal{S}}^{(1)} - \Theta_{\mathcal{S}}^{(2)}$ and $\mathbb{T}$ coincide, which completes the proof.
\end{proof}

\begin{proof}[Proof of Proposition~\ref{prop:onestep_strong}]
  Beginning from \eqref{t_expanded}, we have by assumption that
  \begin{align*}
    \mathbb{T} &= U_{[r]}S \left( V_{-[c]}^{\tp} V_{-[c]} \right) S^{-1} \left( U_{-[r]}^{\tp} U_{-[r]} \right) SV_{[c]}^{\tp} \\
    &= (\rho_U\rho_V) U_{[r]}SV_{[c]}^{\tp} \\
    &= (\rho_U\rho_V) \left( \Theta_{\mathcal{S}}^{(1)} - \Theta_{\mathcal{S}}^{(2)} \right).
  \end{align*}
  Since $\mathbb{T}$ is a scalar multiple of $\Theta_{\mathcal{S}}^{(1)} - \Theta_{\mathcal{S}}^{(2)}$, it follows that they have the same leading left and right singular subspaces for any $d \in \{1,\ldots,d_{\mathcal{S}}\}$.
\end{proof}

\subsection{Proof of Proposition~\ref{prop:wobs}}
\label{sec:app:proof2}
\begin{proof}[Proof of Proposition~\ref{prop:wobs}]

Throughout this proof, since $m$ is the only growing dimension, we will index sequences by $m$ instead of $n$, which was used in statements in the body of the paper. 

To begin, we define $(2d_rd_c) \times (2d_rd_c)$ covariance matrices $F_m$ and $G_m(\tilde{\gamma})$ as a function of an arbitrary $\tilde{\gamma} \in \mathbb{R}^{d_rd_c}$:
\begin{equation*}
  F_m = \bar{X}^{\tp} \operatorname{diag}(h'(\mu_{\mathcal{S},m})) \bar{X}, \tabby
  G_m(\tilde{\gamma}) = \bar{X}^{\tp} \\ \operatorname{diag}(h'(\bar{X}\tilde{\gamma})) \bar{X}.
\end{equation*}

The proof will proceed in three steps.
First, we verify Conditions 1 to 4 required for \citesupp{lv14model}, Theorem 7, which shows the asymptotic normality of the GLM coefficients.
Then we verify that under this particular local alternative setup, the MLEs $\hat{\gamma}_{1}$ and $\hat{\gamma}_{2}$ are asymptotically independent, and the rescaled limit $\sqrt{m}\gamma_{m,2}$ converges.
Finally, we will prove that $\widehat{G}\widehat{F}^{-1}\widehat{G}$ is a consistent estimator of $\tilde{G}\tilde{F}^{-1}\tilde{G}$, and the result will follow by Slutsky's theorem.


Recall the definition of the population target  $\gamma_m = \begin{pmatrix} \gamma_{1,m} & \gamma_{2,m} \end{pmatrix}$ which solves the population score equations.

For a given $m$ and arbitrary $\tilde{\gamma}$, the covariance under the misspecified GLM is
\begin{equation*}
  \bm{X}^{\tp} \operatorname{diag}\left\{ h'(\bm{X}\tilde{\gamma})\right\} \bm{X} = m G_m(\tilde{\gamma}),
\end{equation*}
while the covariance under the saturated model is
\begin{equation*}
  \bm{X}^{\tp} \operatorname{diag}\left\{ h'(\bm{1}_m \otimes \mu_{\mathcal{S},m})\right\} \bm{X} = m F_m
\end{equation*}

Condition 1 of \citesupp{lv14model} follows immediately from Assumption~\ref{assump:glm}, part (A), and orthonormality of $U$ and $V$.

To verify Condition 2 of \citesupp{lv14model}, define
\begin{equation*}
  V(\tilde{\gamma}) = F_m^{-1/2}G_m(\tilde{\gamma})F_m^{-1/2}, \tabby
  N_m(\kappa) = \left\{ \tilde{\gamma} : \lVert \sqrt{m} F_m^{1/2}(\tilde{\gamma} - \gamma_m) \rVert_2 \leq \kappa \right\} \subseteq \mathbb{R}^{2d_rd_c},
\end{equation*}
and let $\lambda_{\mathrm{min}}(M)$ denote the smallest eigenvalue of a symmetric matrix $M$.
Note that by Assumption~\ref{assump:glm}, part (A), $V$ is a continuous function of $\tilde{\gamma}$.

Define
\begin{equation*}
  \delta = \inf_{m \geq 1} \min_{g \in \{1,2\}} \min_{(i,j) \in \mathcal{S}} \left\{ h'([\Theta^{(g)}_{\mathcal{S},m}]_{ij}) \right\}.
\end{equation*}
Note that $\delta > 0$ by assumption, since $h'$ is strictly positive, continuous, and the entries of $\Theta^{(1)}_{\mathcal{S},m}$ and $\Theta^{(2)}_{\mathcal{S},m}$ are uniformly bounded.
The first part of the condition holds since for any $m$,
$$
  \lambda_{\mathrm{min}}(F_m) \geq 2\delta > 0,
$$
where the inequality holds by definition of $\delta$ and since $2^{-1/2}\bar{X}$ is orthonormal.

This minimum eigenvalue condition also implies that for any $m$ and $\kappa > 0$, $N_m(\kappa)$ is contained in a ball centered at $\gamma$ with radius proportional to $m^{-1/2}$.
Thus, by continuity of $V$, for any $\kappa > 0$,
$$
  \min_{\tilde{\gamma} \in N_m(\kappa)}  \lambda_{\mathrm{min}}\left\{ V(\tilde{\gamma}) \right\} \geq \min_{\lVert \tilde{\gamma} - \gamma \rVert_2 \leq 1} \left( \lambda_{\mathrm{min}} \left\{ V(\tilde{\gamma}) \right\} \right) > 0
$$
for sufficiently large $m$.

To verify Condition 3 of \citesupp{lv14model}, define
$$
  \tilde{V}(\tilde{\gamma}_1,\ldots,\tilde{\gamma}_{2d_rd_c}) = F_m^{-1/2} \tilde{G}_m(\tilde{\gamma}_1,\ldots,\tilde{\gamma}_{2d_rd_c}) F_m^{-1/2},
$$
where
$$
  \tilde{G}(\tilde{\gamma}_1,\ldots,\tilde{\gamma}_{2d_rd_c}) = \bar{X}^{\tp}
 \left\{ \begin{pmatrix} \vdots & & \vdots \\
   h'(\bar{X}\tilde{\gamma}_1) & \cdots & h'(\bar{X}\tilde{\gamma}_{2d_rd_c}) \\ \vdots & & \vdots
 \end{pmatrix} \circ \bar{X} \right\},
$$
and $\circ$ denotes the Hadamard (element-wise) product of two matrices.
By Assumption~\ref{assump:glm}, part (A), $\tilde{V}$ is a continuous function of its arguments.
Thus, similar to the verification of Condition 2, for any $\kappa > 0$
$$
  \max_{\tilde{\gamma}_1,\ldots,\tilde{\gamma}_{2d_rd_c} \in N_m(\kappa)} \lVert \tilde{V}(\tilde{\gamma}_1,\ldots,\tilde{\gamma}_{2d_rd_c}) - V(\gamma_m) \rVert_2 \rightarrow 0
$$
as $m \rightarrow \infty$, by construction of $N_m(\kappa)$ and continuity of $\tilde{V}$.

The first part of Condition 4 of \citesupp{lv14model} follows from Assumption~\ref{assump:glm}, part (B).
For the second part, let $\bar{x}_i$ denote the $i$th row of $\bar{X}$ for $i=1,\ldots,2rc$.
Then
\begin{align*}
  \sum_{k=1}^m \sum_{i=1}^{2rc} \left( \bar{x}_i^{\tp} \left\{ m F_m \right\}^{-1} \bar{x}_i \right)^{3/2} &= m^{-3/2} \sum_{k=1}^m \sum_{i=1}^{2rc} \left( \bar{x}_i^{\tp} F_m^{-1} \bar{x}_i \right)^{3/2} \\
  &\leq m^{-3/2} \sum_{k=1}^m \sum_{i=1}^{2rc} \left( 2\delta \bar{x}_i^{\tp}\bar{x}_i \right)^{3/2} \\
  &= m^{-1/2} \sum_{i=1}^{2rc} \left( 2\delta \bar{x}_i^{\tp}\bar{x}_i \right)^{3/2} \\ &\rightarrow 0 \\
\end{align*}
as $m \rightarrow \infty$.
This completes the verification of the four conditions.

Thus, by \citesupp{lv14model}, Theorem 7,
\begin{equation} \label{thm7}
  \sqrt{m} F_m^{-1/2} G_m(\gamma_m) \left( \hat{\gamma}_m - \gamma_m \right) \indist \mathcal{N}\left( \bm{0}_{2d_rd_c}, I_{2d_rd_c} \right).
\end{equation}

Next, we characterize the limiting behavior of $\gamma_m$ .
Note that the system of equations which $\gamma_m$ solves can be rewritten as
\begin{small}
\begin{align*}
  \bm{0}_{d_ld_c} &= (V \otimes U)^{\tp} \left\{ h\left(\tau + \frac{\Delta}{\sqrt{m}} \right) +  h\left(\tau - \frac{\Delta}{\sqrt{m}} \right) - h\left\{(V \otimes U)(\gamma_{m,1} + \gamma_{m,2})\right\} - h\left\{(V \otimes U)(\gamma_{m,1} - \gamma_{m,2}) \right\} \right\} \\
  \bm{0}_{d_ld_c} &= (V \otimes U)^{\tp} \left\{ h\left(\tau + \frac{\Delta}{\sqrt{m}} \right) - h\left(\tau - \frac{\Delta}{\sqrt{m}} \right) - h\left\{(V \otimes U)(\gamma_{m,1} + \gamma_{m,2})\right\} + h\left\{(V \otimes U)(\gamma_{m,1} - \gamma_{m,2}) \right\} \right\}.
\end{align*}
\end{small}

Equivalently, we define
\begin{small}
\begin{equation*}
  \hspace{-1cm}
  Q_m(x,y) = \begin{pmatrix} (V \otimes U)^{\tp} \left\{ h \left(\tau + \frac{\Delta}{\sqrt{m}} \right) +  h \left( \tau - \frac{\Delta}{\sqrt{m}} \right) - h\left\{(V \otimes U)\left(x + \frac{y}{\sqrt{m}} \right)\right\} - h\left\{(V \otimes U) \left( x - \frac{y}{\sqrt{m}} \right) \right\} \right\} \\
    \sqrt{m} (V \otimes U)^{\tp} \left\{ h\left(\tau + \frac{\Delta}{\sqrt{m}} \right) - h\left(\tau - \frac{\Delta}{\sqrt{m}} \right) - h\left\{ (V \otimes U)\left(x + \frac{y}{\sqrt{m}}\right)\right\} + h\left\{(V \otimes U)\left(x - \frac{y}{\sqrt{m}} \right)\right\} \right\}
  \end{pmatrix}
\end{equation*}
\end{small}
so that $\bm{0}_{2d_rd_c} = Q_m(\gamma_{m,1},\sqrt{m}\gamma_{m,2})$.
Similarly recall the definition
\begin{equation*}
  Q(x,y) = 2\begin{pmatrix}
    (V \otimes U)^{\tp} \left\{ h(\tau) - h((V \otimes U)x) \right\} \\
    (V \otimes U)^{\tp}\left\{ \operatorname{diag}\{h'(\tau)\} \Delta - \operatorname{diag}\{h'((V \otimes U) x)\}(V \otimes U)y \right\}
  \end{pmatrix}
\end{equation*}
which also has a unique solution $\bm{0}_{2d_rd_c} = Q(\tilde{\tau},\tilde{\Delta})$.
Consider the compact set $\bm{K} = \{(x,y) : \lVert (x,y) \rVert_2 \leq K_{\gamma}\}$.
We will first show that $Q_m$ converges uniformly to $Q$ on $\bm{K}$.


First, it is easy to see that
\begin{equation*}
  (V \otimes U)^{\tp} \left\{ h(\tau + \Delta/\sqrt{m}) +  h(\tau - \Delta/\sqrt{m}) \right\} \rightarrow 2(V \otimes U)^{\tp}h(\tau),
\end{equation*}
free of $(x,y)$.

Next consider
\begin{equation*}
  (V \otimes U)^{\tp} \left\{ h((V \otimes U)(x + y/\sqrt{m}) + h((V \otimes U)(x - y/\sqrt{m}) \right\}.
\end{equation*}
By compactness of $\bm{K}$, $x \pm y/\sqrt{m}$ converges uniformly to $x$.
Then note that the composition of a continuous function on compact domains preserves uniform continuity, so this term converges uniformly to
\begin{equation*}
    2 (V \otimes U)^{\tp} h((V \otimes U)x).
\end{equation*}
Thus convergence holds for the first $d_rd_c$ coordinates.

For the remaining coordinates, consider a single element
\begin{equation*}
  \sqrt{m} \left\{ h([(V \otimes U)(x - y/\sqrt{m}]_{ij}) - h([(V \otimes U)(x + y/\sqrt{m}]_{ij}) \right\}.
\end{equation*}
By a linear expansion of both terms around $h([(V \otimes U)x]_{ij})$, we have that it equals
\begin{equation*}
  (h'(\xi_{ij}^{(m+)}) + h'(\xi_{ij}^{(m-)}))[(V \otimes U)y]_{ij}
\end{equation*}
for intermediate values $\xi_{ij}^{(m\pm)}$.
By compactness of $\bm{K}$, these intermediate values both converge uniformly to $[(V \otimes U)x]_{ij}$.
Combining over all the elements, we conclude that
\begin{equation*}
  \sqrt{m} (V \otimes U)^{\tp} \left\{ h((V \otimes U)(x - y/\sqrt{m}) - h((V \otimes U)(x + y/\sqrt{m}) \right\} \rightarrow 2 \{\operatorname{diag}h'((V \otimes U) x)\}(V \otimes U)y
\end{equation*}
uniformly over $\bm{K}$.
Similarly,
\begin{equation*}
  \sqrt{m} (V \otimes U)^{\tp} \left\{ h(\tau - \Delta/\sqrt{m}) - h(\tau + \Delta/\sqrt{m}) \right\} \rightarrow 2 \{\operatorname{diag}h'(\tau)\} \Delta,
\end{equation*}
free of $(x,y)$.

Then by uniform convergence, and the definition of $\bm{K}$, we conclude that the roots of $Q_m$ also converge to the roots of $Q$ \citepsupp[][Theorem 9.4-3]{keener10theoretical}, that is
\begin{equation*}
  \gamma_{m,1} \rightarrow \tilde{\tau}, \tabby \sqrt{m}\gamma_{m,2} \rightarrow \tilde{\Delta}
\end{equation*}
as $m \rightarrow \infty$.
Note that $\gamma_{m,2} \rightarrow \bm{0}$.

Let $S^F_{g,m} = \operatorname{diag}(h'(\tau + (-1)^{g-1} \Delta/\sqrt{m}))$ and rewrite
\begin{equation*}
  F_m = \begin{pmatrix}
    (V \otimes U)^{\tp} (S^F_{1,m} + S^F_{2,m}) (V \otimes U) & (V \otimes U)^{\tp} (S^F_{1,m} - S^F_{2,m}) (V \otimes U) \\
    (V \otimes U)^{\tp} (S^F_{1,m} - S^F_{2,m}) (V \otimes U) & (V \otimes U)^{\tp} (S^F_{1,m} + S^F_{2,m}) (V \otimes U)
  \end{pmatrix}.
\end{equation*}
Thus,
\begin{equation*}
  F_m \rightarrow \begin{pmatrix} \tilde{F} & 0 \\ 0 & \tilde{F} \end{pmatrix},
\end{equation*}
where $\tilde{F} = 2 (V \otimes U)^{\tp} \operatorname{diag}(h'(\tau)) (V \otimes U)$.

Similarly let $S^G_{g,m} = \operatorname{diag}(h'((V \otimes U) (\gamma_{m,1} + (-1)^{g-1}\gamma_{m,2})))$ and rewrite
\begin{equation*}
  G_m(\gamma_m) = \begin{pmatrix}
    (V \otimes U)^{\tp} (S^G_{1,m} + S^G_{2,m} (V \otimes U) & (V \otimes U)^{\tp} (S^G_{1,m} - S^G_{2,m}) (V \otimes U) \\
    (V \otimes U)^{\tp} (S^G_{1,m} - S^G_{2,m}) (V \otimes U) & (V \otimes U)^{\tp} (S^G_{1,m} + S^G_{2,m}) (V \otimes U).
  \end{pmatrix}
\end{equation*}
Thus,
\begin{equation*}
  G_m(\gamma_m) \rightarrow \begin{pmatrix} \tilde{G} & 0 \\ 0 & \tilde{G} \end{pmatrix}
\end{equation*}
where $\tilde{G} = 2 (V \otimes U)^{\tp} \operatorname{diag}(h'((V \otimes U)\tilde{\tau})) (V \otimes U)$.

By \eqref{thm7} and Slutsky's theorem, the last $d_rd_c$ coordinates of $\hat{\gamma}_m$ satisfy
\begin{equation*}
  \sqrt{m} \tilde{F}^{-1/2} \tilde{G} \hat{\gamma}_{m,2} \indist \mathcal{N}\left( \tilde{F}^{-1/2} \tilde{G} \tilde{\Delta}, I_{d_rd_c} \right),
\end{equation*}
and by continuous mapping theorem,
\begin{equation*}
  m \hat{\gamma}_{m,2}^{\tp} \tilde{G} \tilde{F}^{-1} \tilde{G} \hat{\gamma}_{m,2} \indist \chi^2_{d_rd_c} \left( \tilde{\Delta}^{\tp} \tilde{G} \tilde{F}^{-1} \tilde{G} \tilde{\Delta} \right).
\end{equation*}

To complete the proof, we will show that $\widehat{G}\widehat{F}^{-1}\widehat{G}$ is a consistent estimator of $\tilde{G}\tilde{F}^{-1}\tilde{G}$, in the sense that
$$
  \lVert \widehat{G}\widehat{F}^{-1}\widehat{G} - \tilde{G}\tilde{F}^{-1}\tilde{G} \rVert_2 = o_{\mathbb{P}}(1).
$$
First, by \eqref{thm7}, $\lVert \hat{\gamma}_1 - \tilde{\tau} \rVert_2 = o_{\mathbb{P}}(1)$, thus by continuity,
\begin{equation} \label{var_pt1}
  \lVert \widehat{G} - \tilde{G} \rVert_2 = o_{\mathbb{P}}(1).
\end{equation}
Which implies that
\begin{equation} \label{var_pt2}
  \lVert \widehat{G} \rVert_2 \leq \lVert \tilde{G} \rVert_2 + \lVert \widehat{G} - \tilde{G} \rVert_2 = O_{\mathbb{P}}(1).
\end{equation}

Define
\begin{equation*}
  \delta' = \min_{(i,j) \in \mathcal{S}} \left\{ h'(\tau_{ij})\right\}.
\end{equation*}
Then
\begin{equation*}
  \frac{1}{2} \lambda_{\mathrm{min}}(\widehat{F}) \geq \min_{(i,j) \in \mathcal{S}} \left\{ h'([\widetilde{\Theta}_{\mathcal{S}}]_{ij})\right\} \geq \delta' + o_{\mathbb{P}}(1)
\end{equation*}
since $\widetilde{\Theta}_{\mathcal{S}}$ is a consistent estimator of $\tau$.
This implies that
\begin{equation} \label{var_pt3}
  \lVert \widehat{F}^{-1} \rVert_2 \leq \frac{1}{\lambda_{\mathrm{min}}(\widehat{F})} = O_{\mathbb{P}}(1), 
\end{equation}
and finally
\begin{align*} \label{var_pt4}
  \lVert \widehat{F}^{-1} - \tilde{F}^{-1} \rVert_2 &\leq \frac{\lVert \widehat{F} - \tilde{F} \rVert_2}{\lambda_{\mathrm{min}}(\widehat{F}) \lambda_{\mathrm{min}}(\tilde{F})} \\
  &\leq \frac{1}{2\delta' \lambda_{\mathrm{min}}(\widehat{F})} \left\lVert h'\left\{ \vecz(\widetilde{\Theta}_{\mathcal{S}}) \right\} - h'\{ \vecz(\tau)\} \right\rVert_{\infty} = o_{\mathbb{P}}(1) \numberthis
\end{align*}
by the law of large numbers and Assumption~\ref{assump:glm}, part (A).
Combining \eqref{var_pt1}-\eqref{var_pt4}, we conclude that
$$
  \lVert \widehat{G}\widehat{F}^{-1}\widehat{G} - \tilde{G}\tilde{F}^{-1}\tilde{G} \rVert_2 = o_{\mathbb{P}}(1),
$$
which completes the proof of Proposition~\ref{prop:wobs}

\end{proof}

\subsection{Proof of Proposition~\ref{prop:wobs_largen}}

\begin{proof}[Proof of Proposition~\ref{prop:wobs_largen}]
    This proof proceeds similarly to the proof of Proposition~\ref{prop:wobs}, by verifying the regularity conditions of \citesupp{lv14model}, then applying their main results. Some detailed arguments in this proof are omitted as they are similar to the proof of Proposition~\ref{prop:wobs}.
    Define $\bm{X}_n = \sqrt{rc}(\bar{X}_n~\otimes~\bm{1}_m)$.
    For each $n$, define a system of equations in $\gamma \in \real^{2d_rd_c}$,
\[
  \bm{0}_{2d_rd_c} = g_n(\gamma) = \bm{X}_n^{\tp} \left\{ h(\tau_n \otimes \bm{1}_{2m} ) - h(\bm{X}_n \gamma)\right\}.
\]
Expanding the definition of $\bm{X}_n$,
$$
  g_n(\gamma) = m \sqrt{rc} \begin{pmatrix}
    (V_n \otimes U_n)^{\tp} \left\{ h(\Theta_n) - h(\sqrt{rc}(V_n \otimes U_n)(\gamma_1 + \gamma_2)\right\} \\
    (V_n \otimes U_n)^{\tp} \left\{ h(\Theta_n) - h(\sqrt{rc}(V_n \otimes U_n)(\gamma_1 - \gamma_2))\right\}
\end{pmatrix}
$$
and by construction, if $\bar{\gamma}_n = (\bar{\gamma}_{n,1},\bar{\gamma}_{n,2})$ denotes the unique solution to this system, $\bar{\gamma}_{n,2} = \bm{0}_{d_rd_c}$. 
\citesupp{lv14model}, Theorem 5, guarantees that there is a unique solution to this system of equations.
Define
$$
  \bm{F}_n = \bm{X}_n^{\tp} \diag\{ h'(\tau_n) \} \bm{X}_n, \quad \bm{G}_n = \bm{X}_n^{\tp} \diag\{ h'(\bm{X}_n \bar{\gamma}_{n}) \} \bm{X}_n,
$$
which can be simplified as
\begin{equation} \label{big_FG}
  \bm{F}_n = (rcm) \begin{pmatrix}
  F_n & \bm{0} \\ \bm{0} & F_n
\end{pmatrix}, \quad \bm{G}_n = (rcm) \begin{pmatrix}
G_n & \bm{0} \\ \bm{0} & G_n
\end{pmatrix}.
\end{equation}
where
\begin{align*}
  F_n &= (V_n \otimes U_n)^{\tp} \diag\{ h'(\tau_n) \} (V_n \otimes U_n), \\
  G_n &= (V_n \otimes U_n)^{\tp} \diag\{ h'(\sqrt{rc}(V_n \otimes U_n) \bar{\gamma}_{n,1}) \} (V_n \otimes U_n),
\end{align*}

Similar to the proof of Proposition 3, Conditions 1--4 of \citetsupp[][Section 4.1]{lv14model} will hold as long as
\begin{equation} \label{largen_want}
 (rcm) \lambda_{\min}^{3}(F_n) \rightarrow \infty, \quad \inf_{n \geq 1} \lambda_{\min}(F_n^{-1/2}G_nF_n^{-1/2}) > 0.
\end{equation}
Note that
$$
  \lambda_{\min}(F_n^{-1/2}G_nF_n^{-1/2}) \geq \lambda_{\min}(F_n^{-1})\lambda_{\min}(G_n) = \frac{\lambda_{\min}(G_n)}{\lambda_{\max}(F_n)}.
$$
Finally, we have that $\sqrt{rc}\bar{\gamma}_{n,1} = \gamma_{n,1}$, since $\gamma_{n,1}$ is defined in Section~\ref{subsubsec:glm_largen} based on the solution to \eqref{population_glm}, which uses an unnormalized design matrix.
Thus, the desired results in \eqref{largen_want} will hold by \eqref{hprime_cond}, since $rcm \rightarrow \infty$.
Also note that \eqref{largen_want} may continue to hold even when $\lambda_{\min}(F_n) \rightarrow 0$, see the discussion in Remark~\ref{rem:sparse_logit} in the body of the paper, and Remark~\ref{rem:sparse_logit_supp} following this proof.

\citetsupp{lv14model}, Theorem 7, along with the continuous mapping theorem, gives that
$$
    m \hat{\gamma}_{2,n} G_n F_n^{-1} G_n \hat{\gamma}_{2,n} \indist \chi^2_{d_rd_c}.
$$
To complete the proof, it suffices to show that $F_n$ and $G_n$ can be estimated consistently (in operator norm) as $n \rightarrow \infty$.

To estimate $G_n$, by continuity of $h'$ (Assumption~\ref{assump:glm}) we need
$$
  \lVert \sqrt{rc} (V_n \otimes U_n) \{(rc)^{-1/2}\hat{\gamma}_{n,1} - \bar{\gamma}_{n,1}\} \rVert_{\infty} \leq \sqrt{rc} \lVert (V_n \otimes U_n) \rVert_{2 \rightarrow \infty} \lVert (rc)^{-1/2} \hat{\gamma}_{n,1} - \gamma_{n,1} \rVert_2  \inprob 0.
$$
This will hold by \citetsupp{lv14model} Theorem 6, and \eqref{coherence_cond}.

To estimate $F_n$, suppose we have proposed an estimator of $\tau_n$ which satisfies Assumption~\ref{assump:theta_est}.
Then
\begin{align*}
    \lVert \widehat{F}_n - F_n \rVert_2 &\leq (d_rd_c) \max_{1 \leq r_1,r_2 \leq d_rd_c} \left\lvert [\widehat{F}_n - F_n]_{r_1r_2} \right\rvert \\
    &= (d_rd_c) \left\lvert \sum_{(ij) \in \mathcal{S}} \{ h'([\widetilde{\Theta}_n]_{ij}) - h'([\vecz^{-1}(\tau_n)]_{ij})\} [(V \otimes U]_{ij,r1}[(V \otimes U]_{ij,r2} \right\rvert \\
    &\leq d_rd_c \left\{ \sup_{n \geq 1} \lVert (V \otimes U) \rVert_{2 \rightarrow \infty} \right\}^2 \sum_{(ij) \in \mathcal{S}} \left\lvert h'([\widetilde{\Theta}_n]_{ij}) - h'([\vecz^{-1}(\tau_n)]_{ij}) \right\rvert \inprob 0
\end{align*}
where the final convergence is given by Assumption~\ref{assump:theta_est} and \eqref{coherence_cond}.
This completes the proof.

\end{proof}

\begin{remark} \label{rem:sparse_logit_supp}
    The sparsity condition given in Remark~\ref{rem:sparse_logit} follows clearly from the first statement in display \eqref{largen_want}.
    Furthermore, note that bounding the minimum eigenvalues of $\widetilde{F}_n$ and $\widetilde{G}_n$ only requires $\bm{v}^{\tp}\diag\{h'(\Theta_n)\}\bm{v}$ to be bounded away from $0$ for unit vectors $\bm{v} \in \operatorname{col}\{(V_n \otimes U_n)\}$.
    Thus if $(V_n \otimes U_n)$ is chosen to restrict the attention of the test to a denser subset of the hypothesis set edges, Proposition~\ref{prop:wobs_largen} may hold under weaker overall sparsity conditions.
    We leave further investigation of this sort of ``sparsity-avoiding'' mesoscale projection to future work.
\end{remark}

\subsection{Proofs of Lemma~\ref{lem:od_consistent} and Corollary~\ref{cor:wobs_od}}

\begin{proof}[Proof of Lemma~\ref{lem:od_consistent}]
For each $(ij) \in \mathcal{S}$, $g=1,2$, and $k=1,\ldots,m$ define the (random) function
\[
  f_{ij,g,k}(\theta) = \frac{([A_{k}^{(g)}]_{ij} - h(\theta))^2}{h'(\theta)}
\]
By a linear expansion,
\[
  f_{ij,g,k}(\widetilde{\Theta}_{ij}^{(g)}) = f_{ij,g,k}(\Theta_{ij}^{(g)}) + f'_{ij,g,k}(\breve{\Theta}_{ij}^{(g)})(\widetilde{\Theta}_{ij}^{(g)} - \Theta_{ij}^{(g)}),
\]
where $\breve{\Theta}_{ij}^{(g)}$ is between $\widetilde{\Theta}_{ij}^{(g)}$ and $\Theta_{ij}^{(g)}$.
\[
  f'_{ij,g,k}(\theta) = (h(\theta) - [A_{k}^{(g)}]_{ij}) - ([A_{k}^{(g)}]_{ij} - h(\theta))^2 \frac{h''(\theta)}{\{h'(\theta)\}^2} = O_{\prob}(1)
\]
for $\theta$ in a neighbourhood of $\Theta_{ij}^{(g)}$ by continuity of $h$, $h'$, and $h''$, and since higher moments of $[A_{k}^{(g)}]_{ij}$ exist (Assumption~\ref{assump:glm_od}). Thus,
\[
  f'_{ij,g,k}(\breve{\Theta}_{ij}^{(g)})(\widetilde{\Theta}_{ij}^{(g)} - \Theta_{ij}^{(g)}) = o_{\prob}(1).
\]
It follows that
\[
  \hat{\phi} = \frac{m}{m-1} \cdot \frac{1}{2rcm}  \sum_{g=1}^n \sum_{k=1}^m \sum_{(i,j) \in \mathcal{S}} f_{ij,g,k}(\Theta_{ij}^{(g)}) + o_{\prob}(1).
\]
Note that
\[
  \expect\left\{ f_{ij,g,k}(\Theta_{ij}^{(g)}) \right\} = \expect\left\{ \big\lvert f_{ij,g,k}(\Theta_{ij}^{(g)}) \big\rvert \right\} = \frac{\var([A_{k}^{(g)}]_{ij})}{h'(\Theta_{ij}^{(g)})} = \phi.
\]
By Assumptions~\ref{assump:glm_od} and \eqref{hprime_cond},
$$
    \var\left\{ f_{ij,g,k}(\Theta_{ij}^{(g)}) \right\} \leq \expect\left(\left\{ f_{ij,g,k}(\Theta_{ij}^{(g)}) \right\}^2\right) = \frac{\expect\{ ([A_{k}^{(g)}]_{ij} - \Theta_{ij}^{(g)})^4\}}{\{h'(\Theta_{ij}^{(g)})\}^2}
$$
is uniformly bounded, so by the weak law of large numbers, $\hat{\phi} \inprob \phi$.
\end{proof}

\begin{proof}[Proof of Corollary~\ref{cor:wobs_od}]
    Identical arguments to the proof of Proposition~\ref{prop:wobs_largen} can be used to show that
    \begin{equation} \label{intermediate_od}
        \phi^{-2} m \hat{\gamma}_{2,n} G_n F_n^{-1} G_n \hat{\gamma}_{2,n} \indist \chi^2_{d_rd_c}.
    \end{equation}
    since in place of $\bm{F}_n$ (see display \ref{big_FG}) we will have 
    $$
        \bm{X}_n^{\tp} \operatorname{Cov}(\vecz(Y)) \bm{X}_n = \phi^2 \bm{X}_n^{\tp} \diag\{ h'(\tau_n) \} \bm{X}_n = \phi^2 \bm{F}_n.
    $$
    under the exponential family edge model with dispersion parameter $\phi$.
    The result then follows by \eqref{intermediate_od}, Lemma~\ref{lem:od_consistent}, and Slutsky's theorem.
\end{proof}

\subsection{Proofs of Propositions~\ref{prop:fobs_tilde} and \ref{prop:fobs}}

\begin{proof}[Proof of Proposition~\ref{prop:fobs_tilde}]
    We prove this result as a special case of Proposition~\ref{prop:fobs_tilde_nuisance}. Following the notation of Appendix~\ref{app:nuisance}, with no nuisance covariates ($q=0$) we have
    \begin{align*}
  \bm{Z} &= \begin{pmatrix}
      \bm{1}_{m} & \bm{1}_{m} \\ \bm{1}_{m} & \bm{0}_{m},
  \end{pmatrix} \in \real^{M \times 2} \\
   \bm{Z}^{\tp}\bm{Z} &= \begin{pmatrix}
    2m & m \\ m & m
\end{pmatrix}, \\
   (\bm{Z}^{\tp}\bm{Z})^{-1} &= \frac{1}{m^2} \begin{pmatrix}
    m & -m \\ -m & 2m
\end{pmatrix} = \begin{pmatrix}
  1/m & -1/m \\ -1/m & 2/m
\end{pmatrix},  \\
 \bm{Z}(\bm{Z}^{\tp}\bm{Z})^{-1}\bm{e}_2 &= \begin{pmatrix}
    \frac{1}{m} \bm{1}_{m} \\ -\frac{1}{m} \bm{1}_{m}
\end{pmatrix}. \\
\end{align*}
Thus,
\begin{equation} \label{ydiff}
  Y \bm{Z}(\bm{Z}^{\tp}\bm{Z})^{-1}\bm{e}_2 = \frac{1}{m} \sum_{k=1}^m \vecz([A^{(1)}_k]_{\mathcal{S}}) - \frac{1}{m} \sum_{k=1}^m \vecz([A^{(2)}_k]_{\mathcal{S}}) = \bar{Y}^{(\mathrm{diff})} \in \real^{rc}.
\end{equation}
Substituting \eqref{ydiff} into Proposition~\ref{prop:fobs_tilde_nuisance} proves the result.
\end{proof}

\begin{proof}[Proof of Proposition~\ref{prop:fobs}]
    We prove this result as a special case of Proposition~\ref{prop:fobs_nuisance}. Following the notation of Appendix~\ref{app:nuisance}, with no nuisance covariates ($q=0$) and $m > 1$, we have
    \begin{align}
        h_2 &= 2/m, \label{h2} \\
           \bm{Z} (\bm{Z}^{\tp}\bm{Z})^{-1} \bm{Z}^{\tp} &= \begin{pmatrix}
    \frac{1}{m} \bm{1}_{m \times m} & \bm{0}_{m \times m} \\
    \bm{0}_{m \times m} & \frac{1}{m} \bm{1}_{m \times m}
\end{pmatrix}. \nonumber
    \end{align}
    Thus,
    \begin{equation} \label{yresid}
        Y \bm{H}^{\perp}_{Z} = Y\left\{ I_{2m} - \begin{pmatrix}
    \frac{1}{m} \bm{1}_{m \times m} & \bm{0}_{m \times m} \\
    \bm{0}_{m \times m} & \frac{1}{m} \bm{1}_{m \times m}
\end{pmatrix} \right\} = Y^{(\mathrm{resid})}.
    \end{equation}
    Substituting \eqref{ydiff}, \eqref{h2}, and \eqref{yresid} into Proposition~\ref{prop:fobs_nuisance} proves the result.
\end{proof}

\subsection{Proof of Proposition~\ref{prop:power_bound}}

Before proving Proposition~\ref{prop:power_bound}, we state and prove two lemmas.
The first is a property of Kronecker products of orthonormal matrices, and the second is a property of the CDF of the non-central $F$-distribution.


\begin{lemma} \label{lem:on_kron}
  Suppose $d_1d_2 \leq n$, and $\hat{U}_1,U_1 \in \mathbb{R}^{n \times d_1}$ and $\hat{U}_2,U_2 \in \mathbb{R}^{n \times d_2}$ are orthonormal matrices which satisfy
  $$
    \lVert \hat{U}_1 - U_1O_1 \rVert_2 \leq \epsilon_1, \tabby \lVert \hat{U}_2 - U_2O_2 \rVert_2 \leq \epsilon_2
  $$
  for some orthonormal matrices $O_1 \in \mathbb{R}^{d_1 \times d_1}$ and $O_2 \in \mathbb{R}^{d_2 \times d_2}$.
  Then $\hat{U}_1 \otimes \hat{U}_2$, $U_1 \otimes U_2$, and $O_1 \otimes O_2$ are orthonormal matrices, and
  $$
    \lVert (\hat{U}_1 \otimes \hat{U}_2) - (U_1 \otimes U_2) (O_1 \otimes O_2) \rVert_2 \leq \epsilon_1 + \epsilon_2.
  $$
\end{lemma}

\begin{proof}[Proof of Lemma~\ref{lem:on_kron}]
  For orthonormal $U \in \mathbb{R}^{n \times r}$ and $V \in \mathbb{R}^{n \times s}$, we have
  \begin{align*}
    (U \otimes V)^{\tp} (U \otimes V) = (U^{\tp}U) \otimes (V^{\tp}V) = I_{rs},
  \end{align*}
  Which shows that $\hat{U}_1 \otimes \hat{U}_2$, $U_1 \otimes U_2$, and $O_1 \otimes O_2$ are orthonormal matrices.
  Then,
  \begin{align*}
    &\lVert (\hat{U}_1 \otimes \hat{U}_2) - (U_1 \otimes U_2) (O_1 \otimes O_2) \rVert_2 \\
    = &\lVert (\hat{U}_1 \otimes \hat{U}_2) - (U_1O_1 \otimes U_2O_2) \rVert_2 \\
    = &\lVert (\hat{U}_1 \otimes \hat{U}_2) - (\hat{U}_1 \otimes U_2O_2) + (\hat{U}_1 \otimes U_2O_2) - (U_1O_1 \otimes U_2O_2) \rVert_2 \\
    \leq &\lVert \hat{U}_1 \otimes (\hat{U}_2 - U_2O_2) \rVert_2 + \lVert (\hat{U}_1 - U_1O_1) \otimes U_2O_2 \rVert_2 \\
    = &\lVert \hat{U}_1 \rVert_2 \lVert \hat{U}_2 - U_2O_2 \rVert_2 + \lVert \hat{U}_1 - U_1O_1 \rVert_2 \lVert U_2O_2 \rVert_2 \\
    \leq &\epsilon_1 + \epsilon_2.
  \end{align*}
\end{proof}


\begin{lemma} \label{lem:ncf_cdf}
  For a fixed cutoff $t > 0$, and degrees of freedom parameters $\nu_1$ and $\nu_2$, define the function
  $$
    \Phi(z ; t,\nu_1,\nu_2) = \mathbb{P}(F_{\nu_1,\nu_2}(z) \leq t),
  $$
  the CDF of a non-central $F$ distribution with non-centrality parameter $z$.
  Then for constants $0 < c_1 < c_2 < \infty$, $\Phi(z ; t,\nu_1,\nu_2)$ is differentiable with respect to $z$ for $z \in [c_1,c_2]$, and
  $$
    \left\lvert \Phi'(z ; t,\nu_1,\nu_2) \right\rvert \leq \frac{1}{2}\Phi(z ; t,\nu_1,\nu_2).
  $$
\end{lemma}

\begin{proof}[Proof of Lemma~\ref{lem:ncf_cdf}]
  The function $\Phi$ is defined in terms of an infinite series,
  $$
    \Phi(z ; t,\nu_1,\nu_2) = \sum_{j=0}^{\infty} \exp(-z/2) \frac{(z/2)^j}{j!} \mathcal{I}\left( \frac{\nu_1t}{\nu_2 + \nu_1t}; \frac{\nu_1}{2} + j, \frac{\nu_2}{2}\right),
  $$
  where $\mathcal{I}(x ; a,b)$ is the regularized incomplete beta function \citesupp{tiku67tables}.
  Note that $\exp(z/2)\Phi(z)$ is a power series, which converges for any $z$.
  Thus it is uniformly convergent for $z \in [c_1,c_2]$, and we can interchange differentiation and summation:
  $$
    \Phi'(z ; t,\nu_1,\nu_2) = \sum_{j=0}^{\infty} \frac{\partial}{\partial z} \left\{ \exp(-z/2) \frac{(z/2)^j}{j!} \mathcal{I}\left( \frac{\nu_1t}{\nu_2 + \nu_1t}; \frac{\nu_1}{2} + j, \frac{\nu_2}{2}\right) \right\}.
  $$
  Each term on the inside is a product and so the derivative splits into two terms.
  The first terms form the convergent series
  $$
    \sum_{j=0}^{\infty} (-1/2) \left\{ \exp(-z/2) \frac{(z/2)^j}{j!} \mathcal{I}\left( \frac{\nu_1t}{\nu_2 + \nu_1t}; \frac{\nu_1}{2} + j, \frac{\nu_2}{2}\right) \right\} = -\frac{1}{2} \Phi(z ; t,\nu_1,\nu_2).
  $$
  Then the second terms form another convergent series
  $$
    \sum_{j=1}^{\infty} (1/2) \left\{ \exp(-z/2) \frac{(z/2)^{j-1}}{(j-1)!} \mathcal{I}\left( \frac{\nu_1t}{\nu_2 + \nu_1t}; \frac{\nu_1}{2} + j, \frac{\nu_2}{2}\right) \right\}.
  $$
  Define $\omega_1 = \nu_1+2$ and $s = \nu_1t / (\nu_1+2)$, and $i=j-1$.
  This series can be rewritten as
  $$
    \frac{1}{2} \sum_{i=0}^{\infty} \left\{ \exp(-z/2) \frac{(z/2)^i}{i!} \mathcal{I}\left( \frac{\omega_1s}{\nu_2 + \omega_1s}; \frac{\omega_1}{2} + i, \frac{\nu_2}{2}\right) \right\},
  $$
  which is a non-central $F$ CDF evaluated at a smaller cutoff $s$, and with larger numerator degrees of freedom parameter $\omega_1$.
  Note that these changes in cutoff and degrees of freedom imply $\Phi(z ; s,\omega_1,\nu_2) < \Phi(z ; t,\nu_1,\nu_2)$, due to the stochastic monotonicity of the non-central $F$ distribution in its non-centrality parameter.s

  Adding these two convergent series, we have
  $$
    \Phi'(z ; t,\nu_1,\nu_2) = \frac{1}{2}\{\Phi(z ; s,\omega_1,\nu_2) - \Phi(z ; t,\nu_1,\nu_2)\} < 0.
  $$
  Then since $\Phi(z ; s,\omega_1,\nu_2) < \Phi(z ; t,\nu_1,\nu_2)$, we conclude that
  $$
    \left\lvert \Phi'(z ; t,\nu_1,\nu_2) \right\rvert \leq \frac{1}{2}\Phi(z ; t,\nu_1,\nu_2),
  $$
  as desired.
\end{proof}


\begin{proof}[Proof of Proposition~\ref{prop:power_bound}]
  To begin, we bound the learned projection test's type II error rates compared to the oracle test, for fixed left and right projections $\hat{U}$ and $\hat{V}$ which satisfy \eqref{subspace_cond}.

  Denote the type II error rate for this test by $\beta_{\mathrm{learn}}^{\hat{U},\hat{V}}$.
  By the independence of the test statistics and the learned projections, it only depends on $\hat{U}$ and $\hat{V}$ through the non-centrality parameter of the test statistic.
  Using notation from Lemma~\ref{lem:ncf_cdf},
  $$
    \beta_{\mathrm{learn}}^{\hat{U},\hat{V}} = \Phi(\psi_{\mathrm{learn}}^{\hat{U},\hat{V}} ; t^*,\nu_1,\nu_2),
  $$
  where $\nu_1$ and $\nu_2$ are defined as in Proposition~\ref{prop:fobs}, $t^*$ is the level $\alpha$ threshold, and
  $$
    \psi_{\mathrm{learn}}^{\hat{U},\hat{V}} = \frac{m}{2\sigma^2} \lVert \hat{U}^{\tp} \left( \Theta^{(1)}_{\mathcal{S}} - \Theta^{(2)}_{\mathcal{S}} \right) \hat{V} \rVert_F^2.
  $$
  Similarly,
  $$
    \beta_{\mathrm{orc}} = \Phi(\psi_{\mathrm{orc}} ; t^*,\nu_1,\nu_2),
  $$
  where by construction, $\psi_{\mathrm{orc}} \geq \psi_{\mathrm{learn}}^{\hat{U},\hat{V}}$, and therefore $\beta_{\mathrm{orc}} \leq \beta_{\mathrm{learn}}^{\hat{U},\hat{V}}$.

  To prove an upper bound, we first bound the difference in non-centrality parameters.
  Let $O_*$ denote the orthogonal transformation matrix which aligns $(\hat{V} \otimes \hat{U})$ and $(V_* \otimes U_*)$.
  Then 
  \begin{align*}
    \psi_{\mathrm{learn}}^{\hat{U},\hat{V}} &= \frac{m}{2\sigma^2} \lVert (\hat{V} \otimes \hat{U})^{\tp} \vecz (\Theta^{(1)}_{\mathcal{S}} - \Theta^{(2)}_{\mathcal{S}}) \rVert_2^2 \\
    &= \frac{m}{2\sigma^2} \lVert \left\{(\hat{V} \otimes \hat{U}) - (V_* \otimes U_*)O_* + (V_* \otimes U_*)O_* \right\}^{\tp} \vecz (\Theta^{(1)}_{\mathcal{S}} - \Theta^{(2)}_{\mathcal{S}}) \rVert_2^2 \\
    &\geq \frac{m}{2\sigma^2} \left\{ \lVert (V_* \otimes U_*)^{\tp} \vecz (\Theta^{(1)}_{\mathcal{S}} - \Theta^{(2)}_{\mathcal{S}}) \rVert_2 - \lVert (\hat{V} \otimes \hat{U}) - (V_* \otimes U_*)O_* \rVert_2 \lVert \vecz (\Theta^{(1)}_{\mathcal{S}} - \Theta^{(2)}_{\mathcal{S}}) \rVert_2 \right\}^2 \\
    &\geq \frac{m}{2\sigma^2} \left\{ \left( \frac{2\sigma^2 \psi_{\mathrm{orc}}}{m} \right)^{1/2} - \epsilon \left( \frac{2\sigma^2 \psi_{\mathrm{max}}}{m} \right)^{1/2} \right\}^2 \\
    &= \psi_{\mathrm{orc}} \left( 1 - \epsilon \sqrt{\frac{\psi_{\mathrm{orc}}}{\psi_{\mathrm{max}}}} \right)^2,
  \end{align*}
  where the second inequality uses Lemma~\ref{lem:on_kron}, and the fact that $\epsilon < \sqrt{\psi_{\mathrm{orc}}/\psi_{\mathrm{max}}}$.
  It then follows that
  \begin{align*} \label{psi_diff}
    \psi_{\mathrm{orc}} - \psi_{\mathrm{learn}}^{\hat{U},\hat{V}} &\leq \psi_{\mathrm{orc}} \left\{ 1 - \left( 1 - \epsilon \sqrt{\frac{\psi_{\mathrm{orc}}}{\psi_{\mathrm{max}}}} \right)^2 \right\} \\
    &\leq 2 \psi_{\mathrm{orc}} \left( \epsilon \sqrt{\frac{\psi_{\mathrm{orc}}}{\psi_{\mathrm{max}}}} \right) \\
    &= 2 \epsilon \sqrt{\psi_{\mathrm{orc}}\psi_{\mathrm{max}}}. \numberthis
  \end{align*}

  We now proceed to upper bound the type II error rate.
  By the mean value theorem, there exists $\tilde{\psi} \in (\psi_{\mathrm{learn}}^{\hat{U},\hat{V}}, \psi_{\mathrm{orc}})$ such that
  $$
    \beta_{\mathrm{orc}} - \beta_{\mathrm{learn}}^{\hat{U},\hat{V}} = \Phi'\left( \tilde{\psi} ; t^*,\nu_1,\nu_2 \right) \left(\psi_{\mathrm{orc}} - \psi_{\mathrm{learn}}^{\hat{U},\hat{V}}\right).
  $$
  Thus,
  \begin{align*} \label{typeii_diff}
    \beta_{\mathrm{learn}}^{\hat{U},\hat{V}} &\leq \beta_{\mathrm{orc}} + \left\lvert \Phi'\left( \tilde{\psi} ; t^*,\nu_1,\nu_2 \right) \right\rvert \left(\psi_{\mathrm{orc}} - \psi_{\mathrm{learn}}^{\hat{U},\hat{V}}\right) \\
    &\leq  \beta_{\mathrm{orc}} + \frac{1}{2} \Phi\left( \tilde{\psi} ; t^*,\nu_1,\nu_2 \right) \left(\psi_{\mathrm{orc}} - \psi_{\mathrm{learn}}^{\hat{U},\hat{V}}\right) \\
    &\leq \beta_{\mathrm{orc}} + \frac{1}{2} \beta_{\mathrm{learn}}^{\hat{U},\hat{V}} \left(\psi_{\mathrm{orc}} - \psi_{\mathrm{learn}}^{\hat{U},\hat{V}}\right), \numberthis
  \end{align*}
  where the second inequality uses Lemma~\ref{lem:ncf_cdf}, and the third uses stochastic monotonicity.

  Combining \eqref{psi_diff} and \eqref{typeii_diff}, we have
  \begin{equation*}
    \beta_{\mathrm{learn}}^{\hat{U},\hat{V}} \leq \beta_{\mathrm{orc}} + \beta_{\mathrm{learn}}^{\hat{U},\hat{V}} \epsilon \sqrt{\psi_{\mathrm{orc}}\psi_{\mathrm{max}}},
  \end{equation*}
  which implies
  \begin{equation} \label{typeii_ub}
    \beta_{\mathrm{learn}}^{\hat{U},\hat{V}} \leq \left( \frac{1}{1 - \epsilon \sqrt{\psi_{\mathrm{orc}}\psi_{\mathrm{max}}}} \right) \beta_{\mathrm{orc}}
  \end{equation}
  since $\epsilon < 1/\sqrt{\psi_{\mathrm{orc}}\psi_{\mathrm{max}}}$.

  To complete the proof, we integrate over learned projections, noting that
  \begin{equation*}
    \beta_{\mathrm{learn}} = \int \beta_{\mathrm{learn}}^{\tilde{U},\tilde{V}} d\mathbb{P}(\tilde{U},\tilde{V}).
  \end{equation*}
  Let $\mathcal{G}$ denote the event \eqref{subspace_cond}.
  By assumption $\mathbb{P}(\mathcal{G}) \geq 1-\xi$, independently of the test statistics.
  The lower bound holds directly from the pointwise lower bounds:
  \begin{equation*}
    \beta_{\mathrm{orc}} = \int \beta_{\mathrm{orc}} d\mathbb{P}(\tilde{U},\tilde{V}) \leq \int \beta_{\mathrm{learn}}^{\tilde{U},\tilde{V}} d\mathbb{P}(\tilde{U},\tilde{V}) = \beta_{\mathrm{learn}}.
  \end{equation*}
  For the upper bound, we have
  \begin{equation*}
    \beta_{\mathrm{learn}} = \int_{\mathcal{G}} \beta_{\mathrm{learn}}^{\tilde{U},\tilde{V}} d\mathbb{P}(\tilde{U},\tilde{V}) + \int_{\mathcal{G}^c} \beta_{\mathrm{learn}}^{\tilde{U},\tilde{V}} d\mathbb{P}(\tilde{U},\tilde{V}).
  \end{equation*}
  By \eqref{typeii_ub}, the first term is bounded above by
  \begin{equation*}
    \left( \frac{1}{1 - \epsilon \sqrt{\psi_{\mathrm{orc}}\psi_{\mathrm{max}}}} \right) \beta_{\mathrm{orc}};
  \end{equation*}
  by construction of $\mathcal{G}$, and since $\beta_{\mathrm{learn}}^{\tilde{U},\tilde{V}} \leq 1$ uniformly, the second term is bounded above by $\xi$, which completes the proof.
\end{proof}

\subsection{Proof of Proposition~\ref{prop:subspace_error}}

We will begin by proving Lemma~\ref{lem:operator_norms}, followed by the proof of Proposition~\ref{prop:subspace_error}.
For the proof of Lemma~\ref{lem:operator_norms}, we state a special form of \citesupp{bandeira16sharp}, Corollary 3.11 which will be applied to bound the operator norm error of each block $\mathcal{C}$, $\mathcal{R}$, $\mathcal{D}$, for $g=1,2$.

\begin{lemma}[\citesupp{bandeira16sharp}, Corollary 3.11] \label{lem:bandeira}
  Suppose $Z$ is a $p \times q$ matrix with iid $\mathcal{N}(0,\sigma^2)$ entries.
  Then, with probability at least
  $$
    1 - \operatorname{exp}\left\{ - \frac{\sigma^2}{8}(\sqrt{p} + \sqrt{q})^2 \right\}
  $$
  we have
  $$
    \lVert Z \rVert_2 \leq \{2 + \delta(p,q)\} \sigma (\sqrt{p} + \sqrt{q}),
  $$
  where $\delta(p,q)$ is defined as in Lemma~\ref{lem:operator_norms}.
\end{lemma}


\begin{proof}[Proof of Lemma~\ref{lem:operator_norms}]
  To begin, we will apply Lemma~\ref{lem:bandeira} to bound the operator norm error of each block $\mathcal{C}$, $\mathcal{R}$, $\mathcal{D}$, for $g=1,2$.

  By \eqref{block_means}, for block $\mathcal{C}$, and $g \in \{1,2\}$, we have that
  $$
    \widehat{\Theta}^{(g)}_{\mathcal{C}} - \Theta^{(g)}_{\mathcal{C}}
  $$
  is an $r \times (n-c)$ matrix, and has iid Gaussian entries with variance $\sigma^2/m$.
  Thus, by Lemma~\ref{lem:bandeira}, union bound, and triangle inequality,
  \begin{equation}
    \lVert (\widehat{\Theta}^{(1)}_{\mathcal{C}} - \widehat{\Theta}^{(2)}_{\mathcal{C}}) - (\Theta^{(1)}_{\mathcal{C}} - \Theta^{(2)}_{\mathcal{C}}) \rVert_2 \leq \epsilon_{\mathcal{C}}
  \end{equation}
  with probability at least
  \begin{equation*}
    1 - 2\exp\left\{ - \frac{\sigma^2}{8m} (\sqrt{r} + \sqrt{n-c})^2 \right\}.
  \end{equation*}

  Similarly, for block $\mathcal{R}$,
  \begin{equation}
    \lVert (\widehat{\Theta}^{(1)}_{\mathcal{R}} - \widehat{\Theta}^{(2)}_{\mathcal{R}}) - (\Theta^{(1)}_{\mathcal{R}} - \Theta^{(2)}_{\mathcal{R}}) \rVert_2 \leq \epsilon_{\mathcal{R}}
  \end{equation}
  with probability at least
  \begin{equation*}
    1 - 2\exp\left\{ - \frac{\sigma^2}{8m} (\sqrt{n-r} + \sqrt{c})^2 \right\}.
  \end{equation*}

  For block $\mathcal{D}$, first suppose $\hat{M}$ and $M$ are matrices, and  $M$ has rank $R$.
  Then
  \begin{equation*}
    \lVert [ \hat{M} ]_{(R)} - M \rVert_2 \leq \lVert [ \hat{M} ]_{(R)} - \hat{M} \rVert_2 + \lVert \hat{M} - M \rVert_2 \leq 2 \lVert \hat{M} - M \rVert_2,
  \end{equation*}
  where the second inequality follows from Eckhart-Young theorem.
  Thus, by a similar argument using Lemma~\ref{lem:bandeira}, union bound, and triangle inequality,
  \begin{equation}
    \left\lVert \left[ \widehat{\Theta}^{(1)}_{\mathcal{D}} - \widehat{\Theta}^{(2)}_{\mathcal{D}} \right]_{(d_*)} -  ( \Theta^{(1)}_{\mathcal{D}} - \Theta^{(2)}_{\mathcal{D}} ) \right\rVert_2 \leq \epsilon_{\mathcal{D}}
  \end{equation}
  with probability at least
  \begin{equation*}
    1 - 2\exp\left\{ - \frac{\sigma^2}{8m} (\sqrt{n-r} + \sqrt{n-c})^2 \right\}.
  \end{equation*}
  A final union bound over the three blocks $\mathcal{C}$, $\mathcal{R}$, and $\mathcal{D}$ completes the proof.
\end{proof}


\begin{proof}[Proof of Proposition~\ref{prop:subspace_error}]
  Suppose the event in Lemma~\ref{lem:operator_norms} holds.
  We will begin by proving a perturbation upper bound
  \begin{equation} \label{T_bound}
    \lVert \widehat{\mathbb{T}} - \mathbb{T} \rVert_2 \leq \frac{\epsilon_{\mathbb{T}}}{(1-\rho_U)(1-\rho_V)}.
  \end{equation}
  Note that both $\mathbb{T}$ and $\widehat{\mathbb{T}}$ are given by the product of three matrices.

  In general, suppose we have matrices $A$, $B$ and $C$ such that the product $ABC$ is well-defined, and perturbed versions of those matrices $\hat{A}$, $\hat{B}$ and $\hat{C}$ such that
  \begin{equation*}
    \lVert \hat{A} - A \rVert_2 \leq \epsilon_A, \tabby \lVert \hat{B} - B \rVert_2 \leq \epsilon_B, \tabby \lVert \hat{C} - C \rVert_2 \leq \epsilon_C.
  \end{equation*}
  Also suppose that
  \begin{equation*}
    \lVert A \rVert_2 \leq \lambda_A, \tabby \lVert B \rVert_2 \leq \lambda_B, \tabby \lVert C \rVert_2 \leq \lambda_C, \tabby
  \end{equation*}
  Then
  \begin{align*} \label{three_terms}
    \lVert \hat{A}\hat{B}\hat{C} - ABC \rVert_2 &\leq \epsilon_A \lambda_B \epsilon_C + \epsilon_A \epsilon_B \epsilon_C + \epsilon_A\lambda_B\lambda_C + \lambda_A \lambda_B \epsilon_C \\
    &+ \epsilon_A \epsilon_B \lambda_C + \lambda_A \epsilon_B \epsilon_C + \lambda_A \epsilon_B \lambda_C \numberthis
  \end{align*}
  Thus we can establish a the upper bound with bounds on the largest singular value of each factor of $\mathbb{T}$, as well as operator norm perturbations for each factor.

  For the first factor of $\mathbb{T}$, $\Theta^{(1)}_{\mathcal{C}} - \Theta^{(2)}_{\mathcal{C}}$, the operator norm perturbation is bounded above by $\epsilon_{\mathcal{C}}$ by Lemma~\ref{lem:operator_norms}, and
  $$
    \lVert \Theta^{(1)}_{\mathcal{C}} - \Theta^{(2)}_{\mathcal{C}} \rVert_2 \leq \lVert U_{[r]}S_{\mathrm{diff}}V_{[-c]}^{\tp} \rVert_2 \leq s_{\mathrm{max}} \sqrt{\rho_U(1 - \kappa_V)} \leq s_{\mathrm{max}} \sqrt{\rho_U}
  $$
  by assumption.
  Similarly, for the third factor we have operator norm perturbation bounded above by $\epsilon_{\mathcal{R}}$, and
  $$
    \lVert \Theta^{(1)}_{\mathcal{R}} - \Theta^{(2)}_{\mathcal{R}} \rVert_2 \leq s_{\mathrm{max}} \sqrt{(1 - \kappa_U)\rho_V} \leq s_{\mathrm{max}} \sqrt{\rho_V}.
  $$

  For the second factor, we need to take additional steps due to the pseudo-inverse.
  Let $\sigma_{\mathrm{min}}(M)$ denote the smallest non-zero singular value of a rank $R$ matrix $M$, and suppose $\hat{M}$ is a perturbed version of $M$ which also has rank $R$.
  Then by \citesupp{wedin73perturbation}, if $\sigma_{\mathrm{min}}(M) \geq \lVert \hat{M} - M \rVert_2$,
  $$
    \lVert \hat{M}^{\dagger} - M^{\dagger} \rVert_2 \leq \frac{2 \lVert \hat{M} - M \rVert_2}{\sigma_{\mathrm{min}}(M) \{ \sigma_{\mathrm{min}}(M) - \lVert \hat{M} - M \rVert_2 \}}. 
  $$

  By Lemma~\ref{lem:operator_norms}, we have
  $$
    \left\lVert \left[ \widehat{\Theta}^{(1)}_{\mathcal{D}} - \widehat{\Theta}^{(2)}_{\mathcal{D}} \right]_{(d_*)} -  ( \Theta^{(1)}_{\mathcal{D}} - \Theta^{(2)}_{\mathcal{D}} ) \right\rVert_2 \leq \epsilon_{\mathcal{D}},
  $$
  and by assumption we have
  $$
    \sigma_{\mathrm{min}}(\Theta^{(1)}_{\mathcal{D}} - \Theta^{(2)}_{\mathcal{D}}) \geq s_{\mathrm{min}} \sqrt{(1-\rho_U)(1-\rho_V)} \geq  2\epsilon_{\mathcal{D}}.
  $$
  Thus,
  \begin{align*}
    &\left\lVert \left( \left[ \widehat{\Theta}^{(1)}_{\mathcal{D}} - \widehat{\Theta}^{(2)}_{\mathcal{D}} \right]_{(d_*)} \right)^{\dagger} - \left(\Theta^{(1)}_{\mathcal{D}} - \Theta^{(2)}_{\mathcal{D}} \right)^{\dagger} \right\rVert_2 \\
    \leq &\frac{4\epsilon_{\mathcal{D}}}{2s_{\mathrm{min}}\sqrt{(1-\rho_U)(1-\rho_V)}\{ s_{\mathrm{min}}\sqrt{(1-\rho_U)(1-\rho_V)} - \epsilon_{\mathcal{D}} \}} \\
    &\leq \frac{4\epsilon_{\mathcal{D}}}{s_{\mathrm{min}}^2 (1 - \rho_U)(1 - \rho_V)}
  \end{align*}
  is an operator norm perturbation bound for the second factor.
  Finally, we have
  \begin{align*}
    \left\lVert \left(\Theta^{(1)}_{\mathcal{D}} - \Theta^{(2)}_{\mathcal{D}} \right)^{\dagger} \right\rVert_2 &= \frac{1}{\sigma_{\mathrm{min}}(\Theta^{(1)}_{\mathcal{D}} - \Theta^{(2)}_{\mathcal{D}})} \\
    &\leq \frac{1}{s_{\mathrm{min}}\sqrt{(1-\rho_U)(1-\rho_V)}} \\
    &\leq \frac{1}{s_{\mathrm{min}}(1-\rho_U)(1-\rho_V)}.
  \end{align*}

  We can then establish \eqref{T_bound} by plugging these six upper bounds into the general result \eqref{three_terms}, where the order of the terms $\epsilon_{\mathbb{T}}^{(1)},\ldots,\epsilon_{\mathbb{T}}^{(7)}$ corresponds to the order in which they appear in \eqref{three_terms}.

  The final result of Proposition~\ref{prop:subspace_error} follows from a perturbation bound on the singular subspaces of $\widehat{\mathbb{T}}$.
  From \citesupp{qi22minimax}, Theorem 3.1, we have
  $$
    \min_{O \in \mathcal{O}_{d_*}} \lVert \hat{U} - U_{[r]}O \rVert_2 + \min_{O \in \mathcal{O}_{d_*}} \lVert \hat{V} - V_{[c]}O \rVert_2 \leq \frac{4 \sqrt{2} \lVert \widehat{\mathbb{T}} - \mathbb{T} \rVert_2}{\sigma_{\mathrm{min}}(\mathbb{T})} \leq \frac{6 \lVert \widehat{\mathbb{T}} - \mathbb{T} \rVert_2}{\sigma_{\mathrm{min}}(\mathbb{T})}.
  $$
  The proof is complete by plugging in \eqref{T_bound}, and lower bounding $\sigma_{\mathrm{min}}(\mathbb{T})$:
  \begin{align*}
    \sigma_{\mathrm{min}}(\mathbb{T}) &= \sigma_{\mathrm{min}}\left\{ S (V_{[-c]}^{\tp}V_{[-c]}) S^{-1} (U_{[-r]}^{\tp}U_{[-r]}) S \right\} \\
    &\geq \frac{s_{\mathrm{min}}^2 (1-\rho_U)(1 - \rho_V)}{s_{\mathrm{max}}}.
  \end{align*}
\end{proof}

\section{Additional experiments on simulated data}

\ifjasa

\else
\fi

\subsection{Misspecification of latent dimension}

In this appendix, we perform a simulation study in the setting where marginal edge dependence is induced by misspecification of the dimension of latent structure. 
In this case, for both Gaussian and binary edge networks, we generate the edge expectation parameters for each group from $p$-dimensional latent space models with inner product similarity.
  Settings are otherwise similar to
  \ifjasa
\else
    Sections~\ref{subsec:gaussian_sims} and \ref{subsec:binary_sims},
  \fi
  with $n=100$, $m=8$, $\sigma^2 = 50$, and an (off-diagonal) rectangular hypothesis set with $r=20$ rows and $c=30$ columns.
  The signal level is specified by the variance of the perturbation of latent positions between group 1 and group 2, that is
  $$
    X_i^{(2)} = X_i^{(1)} + \mathcal{N}(0, \{\mathrm{signal}\}I_p), \quad Y_j^{(2)} = Y_j^{(1)} + \mathcal{N}(0, \{\mathrm{signal}\}I_p)
  $$
  for $i$ corresponding to the hypothesis set rows, and $j$ corresponding to the hypothesis set columns.
  For the Gaussian edge networks we set the signal level to $1/2p$, and for the binary edge networks we set signal level to $1/8p$.
  Setting these signal levels to shrink with $p$ implies that the total signal
  $$
    \lVert \Theta^{(2)}_{\mathcal{S}} - \Theta^{(1)}_{\mathcal{S}} \rVert_F^2
  $$
  should remain approximately constant for different $p$, but that difference will be spread over more latent dimensions.
  
  We perform our mesoscale projection test with projection dimension $d \in \{4,8,12,16\}$. The projection dimension which matches the generative model is $d = 2p$, implying that for $d < 2p$, there will be non-trivial marginal dependence among the edge variables in the orthogonal complement of the projection $(\hat{V} \otimes \hat{U})$. However, our theoretical results imply that the mesoscale projection test should continue to control type I errors at the nominal level, no matter the choice of $d$.

  In Figure~\ref{fig:gaussian_varyp} we plot the results for Gaussian edge networks.
  As predicted by theory, type I errors are controlled at the nominal level for any projection dimension.
  The power under the alternative depends on the relationship between $d$ and $p$: we see that the best choice is always $d = 2p$, as the most efficient projection will capture the structure in the hypothesis set with as few dimensions as possible.

 \begin{figure}
     \centering
     \includegraphics[width=0.48\linewidth]{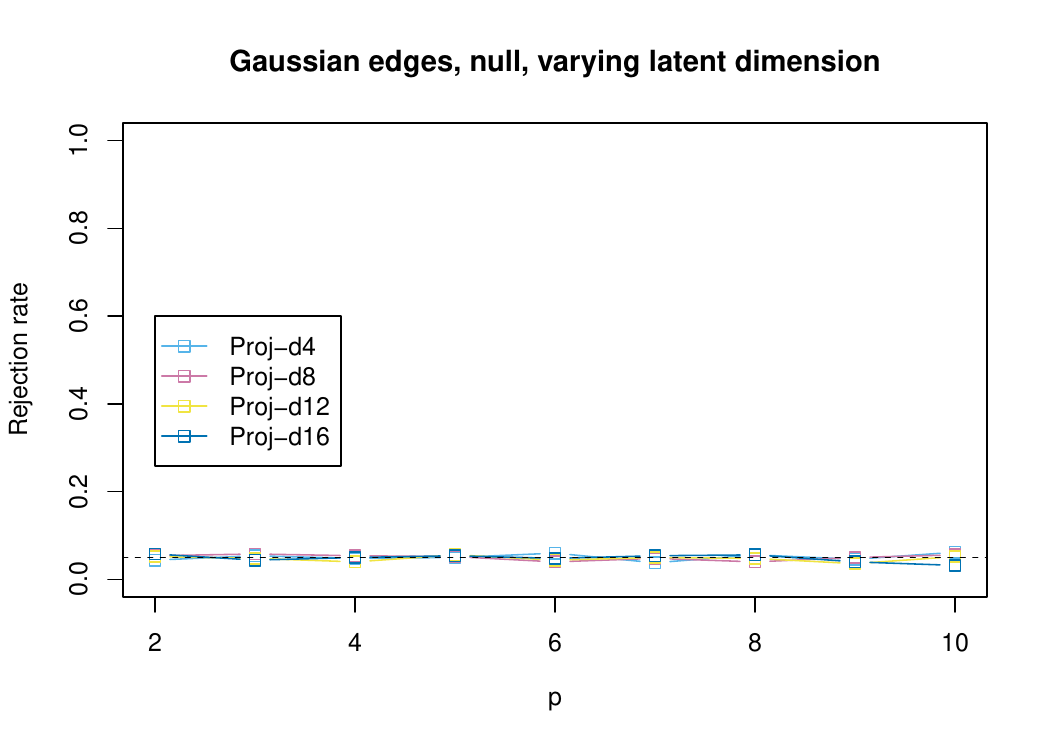}
     \hfill
     \includegraphics[width=0.48\linewidth]{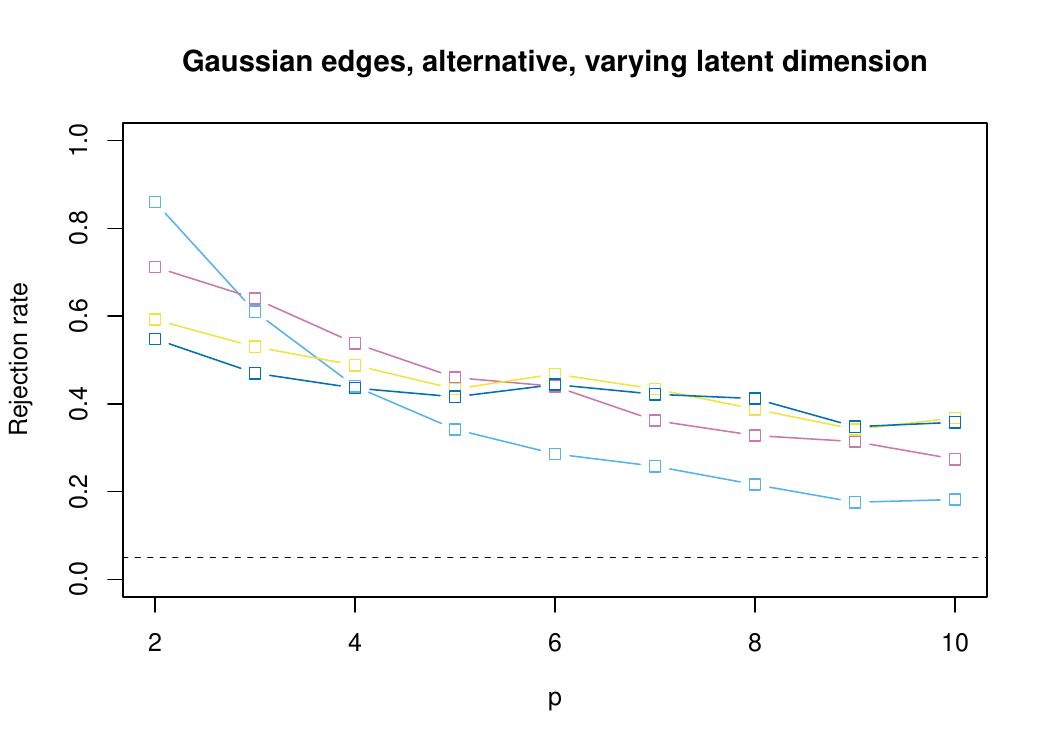}
     \caption{Rejection rates, varying latent space dimension $p$. Gaussian edge networks, $n=100$, $m=8$, $\sigma^2 = 50$, $\mathcal{S}$ is a $20 \times 30$ rectangle.}
     \label{fig:gaussian_varyp}
 \end{figure}

 In Figure~\ref{fig:logit_varyp} we plot the results for binary edge networks.
  With binary edges, we still control type I errors at the nominal level, but the tests become more conservative for larger values of $p$.
  The power under the alternative depends less strongly on the relationship between $d$ and $p$, $d=8$ gives the greatest rejection rate for nearly all choices of $p$. 

 \begin{figure}
     \centering
     \includegraphics[width=0.48\linewidth]{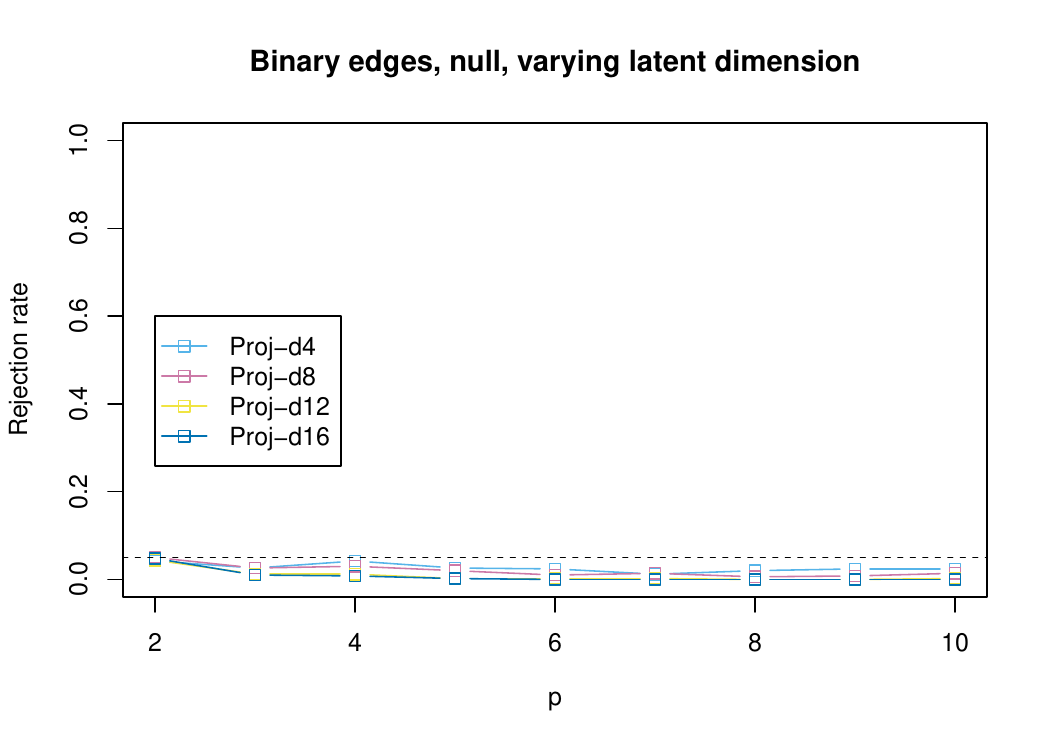}
     \hfill
     \includegraphics[width=0.48\linewidth]{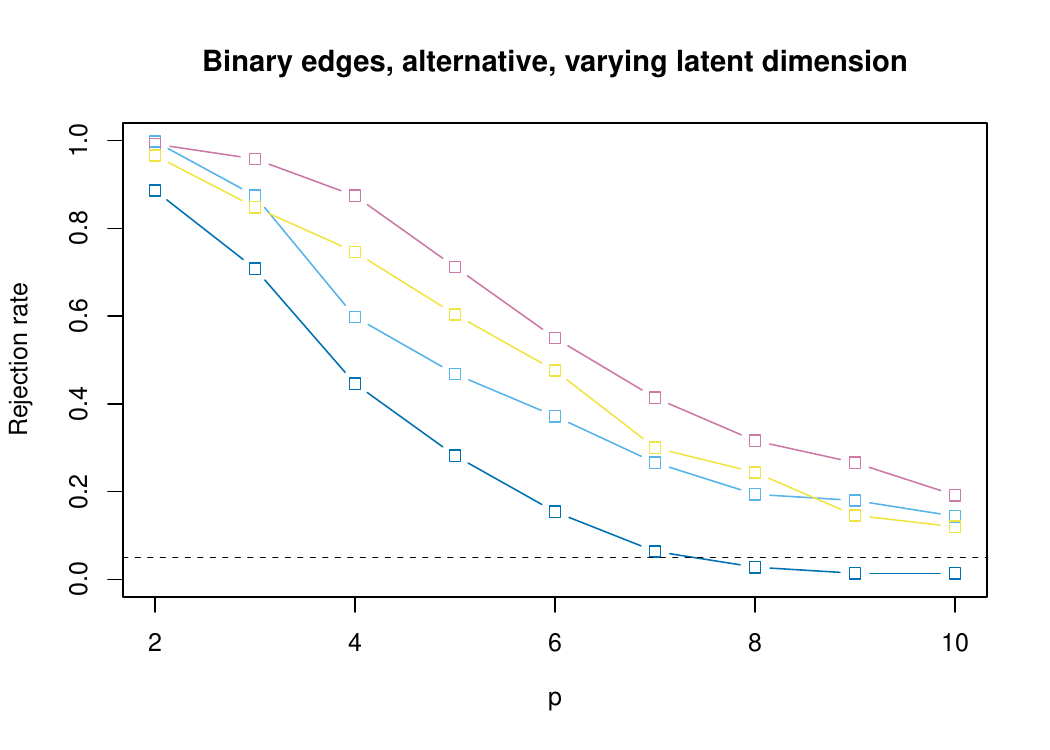}
     \caption{Rejection rates, varying latent space dimension $p$. Binary edge networks, $n=100$, $m=8$, $\mathcal{S}$ is a $20 \times 30$ rectangle.}
     \label{fig:logit_varyp}
 \end{figure}

\subsection{Performance with dependent edges}

In this appendix, we perform a simulation study to evaluate our mesoscale projection tests under {\em sender-receiver edge dependence}, which has been previously introduced as a specific model for weighted networks with dependent edges \citepsupp{du24optimal}.
In this case, dependence is induced through a random effect model, where for all $(ij)$, $k$ and $g$ we generate
  $$
    [A_k^{(g)}]_{ij} = \Theta_{ij}^{(g)} + \eta_{k,i}^{(g)} + \xi_{k,j}^{(g)} + \epsilon_{k,ij}^{(g)},
  $$
  where
  $$
    \eta_{k,i}^{(g)} \iid \mathcal{N}\left( 0, \frac{\kappa\sigma^2}{2} \right), \quad  \xi_{k,j}^{(g)} \iid \mathcal{N}\left( 0, \frac{\kappa\sigma^2}{2} \right), \quad \epsilon_{k,ij}^{(g)} \iid \mathcal{N}\left( 0, (1-\kappa)\sigma^2 \right).
  $$
  In this way, edge variables in the same sample which share a sender or receiver will be correlated; the random effect and residual variances are chosen so that the marginal variance of the edges remains constant ($\sigma^2$).
  Settings are otherwise similar to 
  \ifjasa
\else 
  Section~\ref{subsec:gaussian_sims},
  \fi 
  with $n=100$, $m=8$, $\sigma^2 = 50$, an (off-diagonal) rectangular hypothesis set with $r=20$ rows and $c=30$ columns, a $3$-dimensional inner product latent space model, and signal level $1/6$.




 Results with sender-receiver edge dependence are given in Figure~\ref{fig:senderreceiver}. In this setting we implement the ``Basic'' method, the oracle ($d=6$) projection test, and the usual ($d=6$) projection test. 
 We also implement a projection test where the projection matrices $\hat{V}$ and $\hat{U}$ are chosen to be centered over each row and column, so they are orthogonal to the degree random effects, as described in Appendix A.3, denoted ``Proj. (centered)''.

 \begin{figure}
     \centering
     \includegraphics[width=0.48\linewidth]{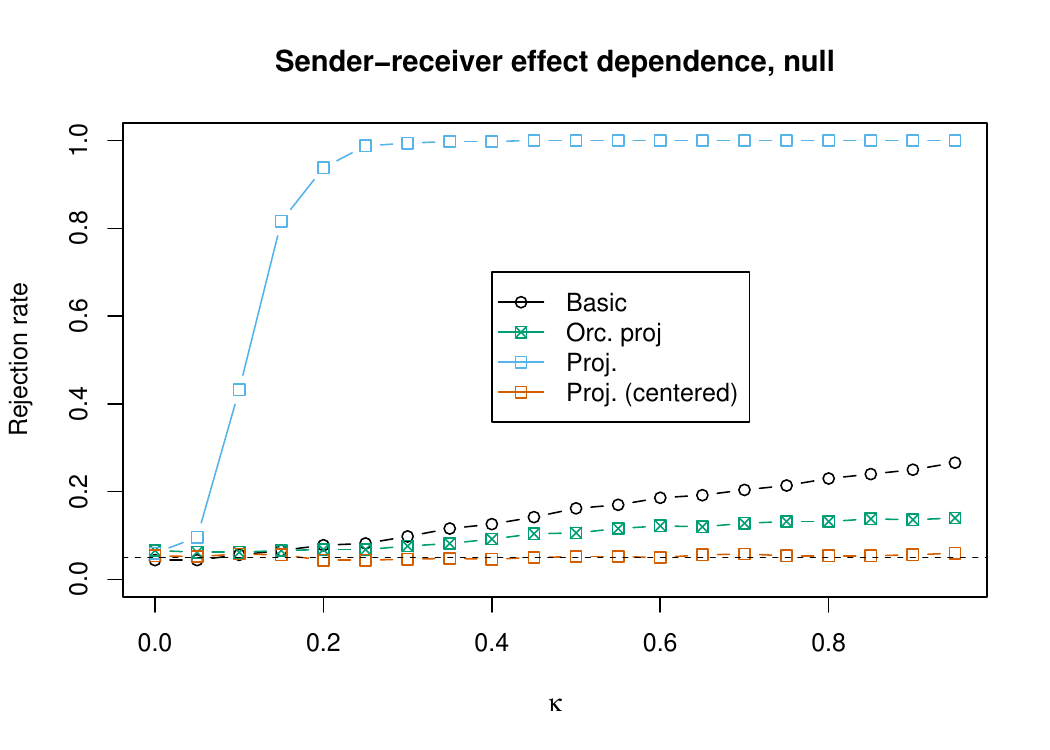}
     \hfill
     \includegraphics[width=0.48\linewidth]{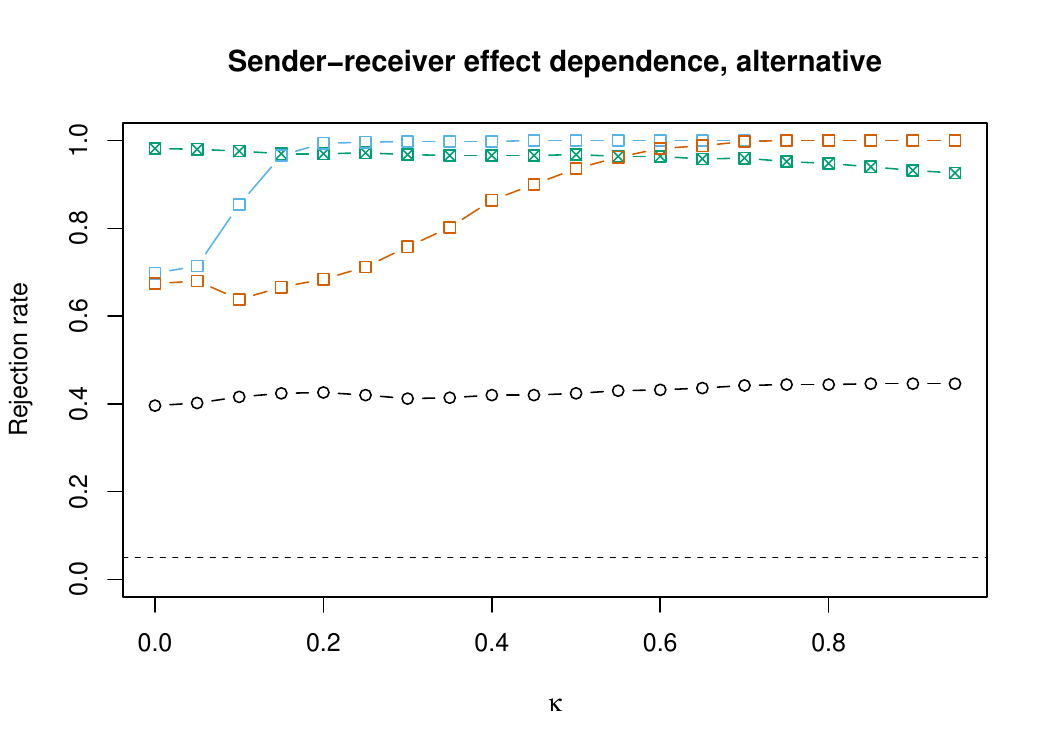}
     \caption{Rejection rates, varying proportion of variance $\kappa$ due to sender-receiver effects. $n=100$, $m=8$, $\sigma^2 = 50$, $\mathcal{S}$ is a $20 \times 30$ rectangle.}
     \label{fig:senderreceiver}
 \end{figure}

 Figure~\ref{fig:senderreceiver} shows that sender-receiver edge dependence has an effect on all of these tests, including the oracle mesoscale projection and the basic test which aggregates all signals. The effect on the level of the usual projection test is strongest, as it nearly always rejects under the null for $\kappa \geq 0.4$. 
 However, the centered projection test does maintain control of type I errors at the nominal level, and by ignoring the variance due to random effects, it achieve better power for large values of $\kappa$.
 In this particular data generating model, where the signal in $\Theta^{(1)}_{\mathcal{S}} - \Theta^{(2)}_{\mathcal{S}}$ is also approximately centered over each row and column, the effect on power is minimal, however in other cases the additional centering restriction may remove signal as well as noise.
 
 In summary, while our projection test is not robust to sender-receiver dependence, neither is the competing basic test. Moreover, our projection-based approach suggests a simple workaround which corrects the level of the test for this particular type of random effect model.

\subsection{Comparison to edge-wise testing approaches}

\newcommand{\sizeS}{\lvert \mathcal{S} \rvert}

In this appendix, we perform a simulation study to compare our more complicated projection-based testing methodology to other approaches based on aggregation of edgewise hypothesis tests.
We implement and compare these approaches for networks with Gaussian edges.

We briefly describe and develop these edge-wise testing approaches as follows.
Define
\[
  W_{ij} = \frac{1}{m_1} \sum_{k=1}^{m_1} [A_k^{(1)}]_{ij} - \frac{1}{m_2} \sum_{k=1}^{m_2} [A_k^{(2)}]_{ij}
\]
for $(i,j) \in \mathcal{S}$.
Under the mesoscale null hypothesis,
\[
  \expect W_{ij} = 0, \quad \var(W_{ij}) = \frac{1}{m_1} \{\sigma_{ij}^{(1)}\}^2 + \frac{1}{m_2} \{\sigma_{ij}^{(2)}\}^2 =: \sigma_{W,ij}^2
\]
where $\{\sigma_{ij}^{(g)}\}^2 = \var(A_{ij}^{(g)})$. 
For $m_1$ and $m_2$ greater than 1, let $\hat{\sigma}_{W,ij}^2$ denote the usual unbiased estimator of $\sigma_{W,ij}^2$ based on the sample variances for each $(i,j,g)$ triple.
Then, by the central limit theorem, as $m_1$ and $m_2$ diverge we have
\[
  x_{ij} = \frac{W_{ij}}{\sigma_{W,ij}} \cdot \frac{\sigma_{W,ij}}{\hat{\sigma}_{W,ij}} \indist \mathcal{N}(0,1).
\]
Thus for a finite edge set,
\[
  \sum_{(i,j) \in \mathcal{S}} x_{ij}^2 \indist \chi^2\left( \sizeS \right),
\]
which suggests an asymptotically valid test based on the quantiles of the $\chi^2$-distribution.
If additionally $\sizeS$ diverges, we can apply a Gaussian approximation for the $\chi^2$-distribution, and get
\[
  \frac{\sum_{(i,j) \in \mathcal{S}} x_{ij}^2 - \sizeS}{\sqrt{2\sizeS}}\indist \mathcal{N}(0,1),
\]
which also suggests an asymptotically valid (one-sided) test based on the upper quantiles of the standard Gaussian distribution.

We can also develop a test based on the upper quantiles of the maximal edge-wise difference in means, after standardization.
In particular, \citesupp{xia22hypothesis} show that
\[
  \prob\left( \max_{(i,j) \in \mathcal{S}} x^2_{ij} \geq \log \sizeS - \log\log \sizeS - \log \pi - 2 \log\log (1-\alpha)^{-1} \right) \rightarrow \alpha
\]
as $m_1$, $m_2$, $\sizeS$ diverge, which suggests an asymptotically valid test.
Finally, suppose for simplicity that both $m_1$ and $m_2$ are even. Define $W^{(s)}_{ij}$ for $s=1,2$ based on a random split of the samples within each group. That is, $W^{(s)}_{ij}$ and $W^{(s)}_{ij}$ are each are based on $m_g/2$ samples in group $g=1,2$.
By independence,
\[
  \expect W_{ij}^{(1)} W_{ij}^{(2)} = 0,
\]
and thus
\[
  \var \left( W_{ij}^{(1)} W_{ij}^{(2)} \right) = \var \left( W_{ij}^{(1)} \right)\var \left( W_{ij}^{(2)} \right) = 4 \sigma_{W,ij}^4.
\]
By edge independence
\[
  \var \left( \sum_{(i,j) \in \mathcal{S}} W_{ij}^{(1)} W_{ij}^{(2)} \right) = 4 \sum_{(i,j) \in \mathcal{S}} \sigma_{W,ij}^4,
\]
which has a consistent estimator $\hat{V}^2 = 4 \sum_{(i,j) \in \mathcal{S}} \left( \hat{\sigma}_{W,ij}^2 \right)^2$.
Then, by central limit theorem, as $\sizeS$ diverges, we have
\[
  \frac{1}{\hat{V}} \sum_{(i,j) \in \mathcal{S}} W_{ij}^1 W_{ij}^2 \indist \mathcal{N}(0,1),
\]
which suggests an asymptotically valid test based on the quantiles of the standard Gaussian distribution.

Thus we have four competing tests based on three statistics: $\sum_{(i,j) \in \mathcal{S}} x^2_{ij}$, $\max_{(i,j) \in \mathcal{S}} x^2_{ij}$ and $\hat{V}^{-1} \sum_{(i,j) \in \mathcal{S}} W_{ij}^1 W_{ij}^2$, and calibrated using asymptotic approximations.
The same three statistics can also be used to define tests which are exactly calibrated with a parametric bootstrap under the mesoscale null.
  
We plot the results in Figure~\ref{fig:edgewise}, based on the same Gaussian edge settings as implemented in   
\ifjasa
\else 
  Section~\ref{subsec:gaussian_sims}.
  \fi 
Note that we do not evaluate the test for $m=1$, since none of the edge-wise test statistics can be computed in this setting.
We denote the test based on $\sum_{(i,j) \in \mathcal{S}} x^2_{ij}$ as ``Sum, Chi'', ``Sum, Z-approx'' when calibrated with a Gaussian approximation to the $\chi^2$, or ``Sum, bootstrap'' when calibrated with a parametric bootstrap. 
We denote the test based on $\max_{(i,j) \in \mathcal{S}} x^2_{ij}$ by  or ``Maximal, bootstrap'' when calibrated with a parametric bootstrap. 
  We compare these edge-wise approaches to the basic test, our mesoscale projection test with $d=4$, and the oracle ($d=6$) projection test.

   \begin{figure}
     \centering
     \includegraphics[width=0.48\linewidth]{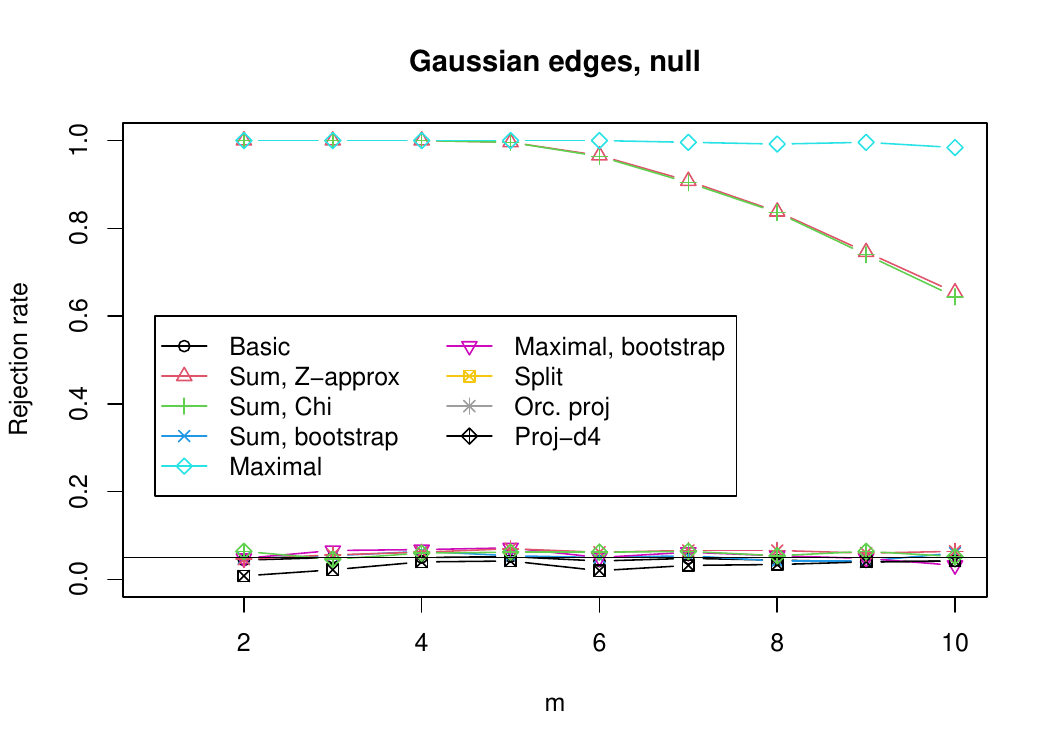}
     \hfill
     \includegraphics[width=0.48\linewidth]{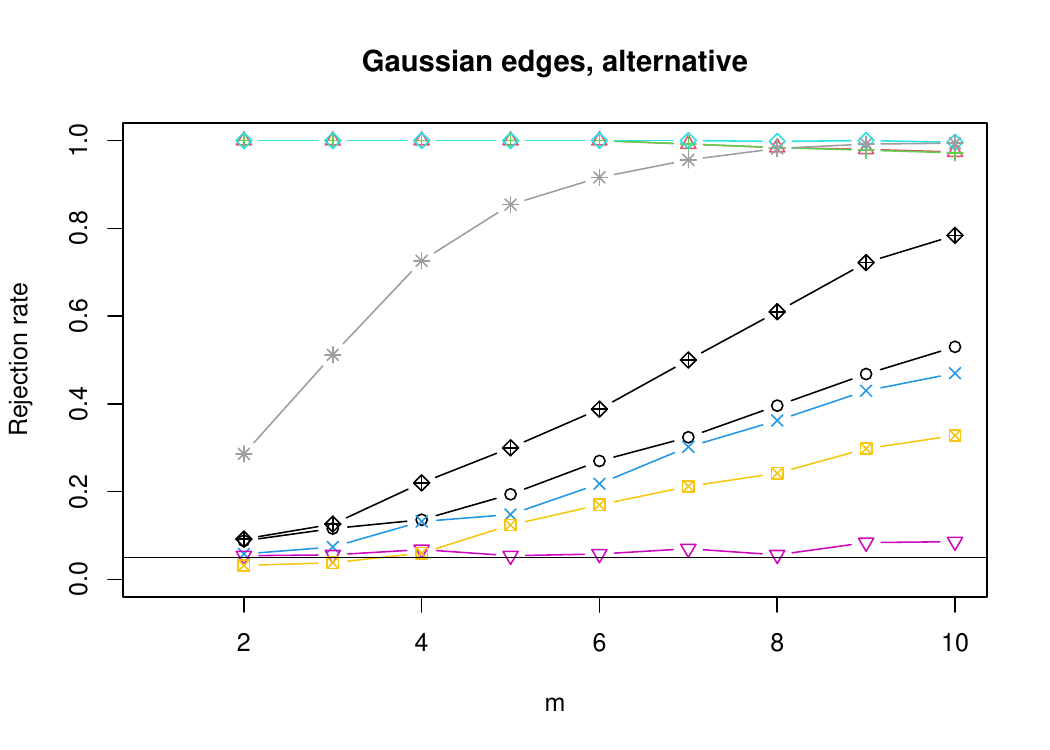}
     \caption{Rejection rates of edge-wise approaches, varying $m$. $n=100$, $\sigma^2 = 50$, $\mathcal{S}$ is a $20 \times 30$ rectangle.}
     \label{fig:edgewise}
 \end{figure}

  Figure~\ref{fig:edgewise} shows that without exact bootstrap calibration, ``Sum, Chi'', ``Sum, Z-approx'', and ``Maximal'' fail to control type I errors at the nominal level. The three edgewise methods that do control type I errors have less power than the basic $F$-test and the mesoscale projection test.
  We observe that the tests based on $\sum_{(i,j) \in \mathcal{S}} x^2_{ij}$ and $\hat{V}^{-1} \sum_{(i,j) \in \mathcal{S}} W_{ij}^1 W_{ij}^2$ are most similar to the ``Basic'' test since they aggregate signals over the hypothesis set with no dimension reduction.
  The test based on $\max_{(i,j) \in \mathcal{S}} x^2_{ij}$ does not effectively aggregate signals in this setting, since no single node pair has contributes overwhelmingly to the signal.

\section{ADMM algorithm for position-based testing} \label{app:admm}


For the competing position model bootstrap test described in Section~\ref{sec:simulations}, we fit $d$-dimensional inner product latent position models to each sample of networks, under the mesoscale null hypothesis.
Under the Gaussian edge network model the constrained maximum likelihood estimator solves the optimization problem
\begin{equation} \label{admm_opt}
  \min_{\Theta^{(1)},\Theta^{(2)}} \left\{ \lVert \bar{A}^{(1)} - \Theta^{(1)} \rVert_F^2 + \lVert \bar{A}^{(2)} - \Theta^{(2)} \rVert_F^2 \right\}
\end{equation}
subject to
$$ \operatorname{rank}(\Theta^{(g}) \leq d, \tabby \Theta_{\mathcal{S}}^{(1)} = \Theta_{\mathcal{S}}^{(2)}; $$
where
$$
  \bar{A}^{(g)} = \frac{1}{m} \sum_{k=1}^m A_k^{(g)}
$$
for $g=1,2$.

To solve \eqref{admm_opt}, we derive the following alternating direction method of multipliers (ADMM) steps with step size parameter $\rho > 0$.
For step $k \geq 0$, the iterates $u_k,z_k,w_k$ are each $2n^2$-dimensional vectors, corresponding to the entries of the $n \times n$ expected adjacency matrices for the two samples.

We set starting values
$$ u_0 = \bm{0}, \tabby [z_0^{(g)}]_{\mathcal{S}} = \frac{1}{2} [\bar{A}^{(1)}]_{\mathcal{S}} + \bar{A}^{(2)}_{\mathcal{S}} ), \tabby [z_0^{(g)}]_{\mathcal{S}^c} = [\bar{A}^{(g)}]_{\mathcal{S}^c}. $$

Then the iterations proceed as follows.
\begin{enumerate}
  \item Update $w_{k+1}$
  $$ [w_{k+1}^{(g)}]_{\mathcal{S}} = \frac{2}{4 + 2\rho} ( [\bar{A}^{(1)}]_{\mathcal{S}} + [\bar{A}^{(2)}]_{\mathcal{S}} ) + \frac{\rho}{4 + 2\rho} ( [z_k^{(1)}]_{\mathcal{S}} + [z_k^{(2)}]_{\mathcal{S}} - [u_k^{(1)}]_{\mathcal{S}} - [u_k^{(2)}]_{\mathcal{S}} ) $$
  $$ [w_{k+1}^{(g)}]_{\mathcal{S}^c} = \frac{2}{2 + \rho} [\bar{A}^{(g)}]_{\mathcal{S}^c} + \frac{\rho}{2 + \rho} ([z_k^{(g)}]_{\mathcal{S}^c} - [u_k^{(g)}]_{\mathcal{S}^c} ) $$
  \item Update $z_{k+1}$
  $$z^{(g)}_{k+1} = \left[ \vecz^{-1}(w_{k+1}^{(g)} + u_k^{(g)} )\right]_{(d)}$$
  for $g=1,2$.
  \item Update $u_{k+1} = u_k + w_{k+1} - z_{k+1}$.
\end{enumerate}
These iterations continue either a maximum number of iterations $K$, or if
$$
  \max\{ \lVert z_{k+1} - z_k \rVert_2^2, \lVert [z_{k+1}^{(1)}]_{\mathcal{S}} - [z_{k+1}^{(2)}]_{\mathcal{S}} \rVert_1 \}
$$
is smaller than some prespecified convergence threshold.



\bibliographystylesupp{abbrvnat}
\bibliographysupp{graph_testing0}


\begin{thebibliography}{45}
\providecommand{\natexlab}[1]{#1}
\providecommand{\url}[1]{\texttt{#1}}
\expandafter\ifx\csname urlstyle\endcsname\relax
  \providecommand{\doi}[1]{doi: #1}\else
  \providecommand{\doi}{doi: \begingroup \urlstyle{rm}\Url}\fi

\bibitem[Agterberg et~al.(2020)Agterberg, Tang, and
  Priebe]{agterberg20nonparametric}
J.~Agterberg, M.~Tang, and C.~E. Priebe.
\newblock Nonparametric two-sample hypothesis testing for random graphs with
  negative and repeated eigenvalues.
\newblock \emph{arXiv:2012.09828}, 2020.

\bibitem[Ahn and Chen(1995)]{ahn95generation}
H.~Ahn and J.~J. Chen.
\newblock Generation of over-dispersed and under-dispersed binomial variates.
\newblock \emph{Journal of Computational and Graphical Statistics}, 4\penalty0
  (1):\penalty0 55--64, 1995.

\bibitem[Athreya et~al.(2018)Athreya, Fishkind, Tang, Priebe, Park, Vogelstein,
  Levin, Lyzinski, Qin, and Sussman]{athreya18statistical}
A.~Athreya, D.~E. Fishkind, M.~Tang, C.~E. Priebe, Y.~Park, J.~T. Vogelstein,
  K.~Levin, V.~Lyzinski, Y.~Qin, and D.~L. Sussman.
\newblock Statistical inference on random dot product graphs: {{A}} survey.
\newblock \emph{The Journal of Machine Learning Research}, 18:\penalty0 92,
  2018.

\bibitem[Athreya et~al.(2021)Athreya, Tang, Park, and
  Priebe]{athreya21estimation}
A.~Athreya, M.~Tang, Y.~Park, and C.~E. Priebe.
\newblock On estimation and inference in latent structure random graphs.
\newblock \emph{Statistical Science}, 36\penalty0 (1), 2021.

\bibitem[Badea et~al.(2017)Badea, Onu, Wu, Roceanu, and
  Bajenaru]{badea17exploring}
L.~Badea, M.~Onu, T.~Wu, A.~Roceanu, and O.~Bajenaru.
\newblock Exploring the reproducibility of functional connectivity alterations
  in {{Parkinson}}'s disease.
\newblock \emph{PLOS ONE}, 12\penalty0 (11):\penalty0 e0188196, 2017.

\bibitem[Bhatia(1997)]{bhatia97matrix}
R.~Bhatia.
\newblock \emph{Matrix Analysis}.
\newblock {Springer}, 1997.

\bibitem[Bulgren(1971)]{bulgren71representations}
W.~G. Bulgren.
\newblock On representations of the doubly non-central {{{\emph{F}}}}
  distribution.
\newblock \emph{Journal of the American Statistical Association}, 66\penalty0
  (333):\penalty0 184--186, 1971.

\bibitem[Cape et~al.(2019)Cape, Tang, and Priebe]{cape19twotoinfinity}
J.~Cape, M.~Tang, and C.~E. Priebe.
\newblock The two-to-infinity norm and singular subspace geometry with
  applications to high-dimensional statistics.
\newblock \emph{Annals of Statistics}, 47\penalty0 (5):\penalty0 2405--2439,
  2019.

\bibitem[Chen et~al.(2020)Chen, Josephs, Lin, Zhou, and
  Kolaczyk]{chen20spectralbased}
L.~Chen, N.~Josephs, L.~Lin, J.~Zhou, and E.~D. Kolaczyk.
\newblock A spectral-based framework for hypothesis testing in populations of
  networks.
\newblock \emph{arXiv:2011.12416}, 2020.

\bibitem[Chung et~al.(2022)Chung, Varjavand, Arroyo-Reli{\'o}n, Alyakin,
  Agterberg, Tang, Priebe, and Vogelstein]{chung22valid}
J.~Chung, B.~Varjavand, J.~Arroyo-Reli{\'o}n, A.~Alyakin, J.~Agterberg,
  M.~Tang, C.~E. Priebe, and J.~T. Vogelstein.
\newblock Valid two-sample graph testing via optimal transport {{Procrustes}}
  and multiscale graph correlation with applications in connectomics.
\newblock \emph{Stat}, 11\penalty0 (1):\penalty0 e429, 2022.

\bibitem[Das~Gupta and Perlman(1974)]{dasgupta74power}
S.~Das~Gupta and M.~D. Perlman.
\newblock Power of the noncentral {{{\emph{F}}}}-test: Effect of additional
  variates on {{Hotelling}}'s {{T2}}-test.
\newblock \emph{Journal of the American Statistical Association}, 69\penalty0
  (345):\penalty0 174--180, 1974.

\bibitem[Donnat and Holmes(2018)]{donnat18tracking}
C.~Donnat and S.~Holmes.
\newblock Tracking network dynamics: {{A}} survey of distances and similarity
  metrics.
\newblock \emph{arXiv:1801.07351 [physics, stat]}, 2018.

\bibitem[Draves and Sussman(2021)]{draves21biasvariance1}
B.~Draves and D.~L. Sussman.
\newblock Bias-variance tradeoffs in joint spectral embeddings.
\newblock \emph{arXiv:2005.02511}, 2021.

\bibitem[Du and Tang(2021)]{du21hypothesis}
X.~Du and M.~Tang.
\newblock Hypothesis testing for equality of latent positions in random graphs.
\newblock \emph{arXiv:2105.10838 [stat]}, 2021.

\bibitem[Fair et~al.(2009)Fair, Cohen, Power, Dosenbach, Church, Miezin,
  Schlaggar, and Petersen]{fair09functional}
D.~A. Fair, A.~L. Cohen, J.~D. Power, N.~U.~F. Dosenbach, J.~A. Church, F.~M.
  Miezin, B.~L. Schlaggar, and S.~E. Petersen.
\newblock Functional brain networks develop from a ``local to distributed''
  organization.
\newblock \emph{PLoS Computational Biology}, 5\penalty0 (5):\penalty0 e1000381,
  2009.

\bibitem[Gavish and Donoho(2017)]{gavish17optimal}
M.~Gavish and D.~L. Donoho.
\newblock Optimal shrinkage of singular values.
\newblock \emph{IEEE Transactions on Information Theory}, 63\penalty0
  (4):\penalty0 2137--2152, 2017.

\bibitem[Ghoshdastidar and {von Luxburg}(2018)]{ghoshdastidar18practical}
D.~Ghoshdastidar and U.~{von Luxburg}.
\newblock Practical methods for graph two-sample testing.
\newblock In \emph{Advances in {{Neural Information Processing Systems}}},
  2018.

\bibitem[Ghoshdastidar et~al.(2017)Ghoshdastidar, Gutzeit, Carpentier, and von
  Luxburg]{ghoshdastidar17twosample1}
D.~Ghoshdastidar, M.~Gutzeit, A.~Carpentier, and U.~von Luxburg.
\newblock Two-sample tests for large random graphs using network statistics.
\newblock In \emph{Proceedings of the 2017 {{Conference}} on {{Learning
  Theory}}}, pages 954--977, 2017.

\bibitem[Ginestet et~al.(2017)Ginestet, Li, Balachandran, Rosenberg, and
  Kolaczyk]{ginestet17hypothesis}
C.~E. Ginestet, J.~Li, P.~Balachandran, S.~Rosenberg, and E.~D. Kolaczyk.
\newblock Hypothesis testing for network data in functional neuroimaging.
\newblock \emph{The Annals of Applied Statistics}, 11\penalty0 (2):\penalty0
  725--750, 2017.

\bibitem[Holland et~al.(1983)Holland, Laskey, and
  Leinhardt]{holland83stochastic}
P.~W. Holland, K.~B. Laskey, and S.~Leinhardt.
\newblock Stochastic blockmodels: First steps.
\newblock \emph{Social Networks}, 5\penalty0 (2):\penalty0 109--137, 1983.

\bibitem[Jin et~al.(2025)Jin, Ke, Luo, and Ma]{jin25optimal}
J.~Jin, Z.~T. Ke, S.~Luo, and Y.~Ma.
\newblock Optimal network pairwise comparison.
\newblock \emph{Journal of the American Statistical Association}, 120\penalty0
  (550):\penalty0 1048--1062, 2025.

\bibitem[Kuntal et~al.(2019)Kuntal, Chandrakar, Sadhu, and
  Mande]{kuntal19netshift}
B.~K. Kuntal, P.~Chandrakar, S.~Sadhu, and S.~S. Mande.
\newblock ``{{NetShift}}'': {{A}} methodology for understanding ``driver
  microbes'' from healthy and disease microbiome datasets.
\newblock \emph{The ISME Journal}, 13\penalty0 (2):\penalty0 442--454, 2019.

\bibitem[Lafond(2015)]{lafond15low}
J.~Lafond.
\newblock Low rank matrix completion with exponential family noise.
\newblock In \emph{Conference on Learning Theory}, pages 1224--1243, 2015.

\bibitem[Levin et~al.(2017)Levin, Athreya, Tang, Lyzinski, and
  Priebe]{levin17central}
K.~Levin, A.~Athreya, M.~Tang, V.~Lyzinski, and C.~E. Priebe.
\newblock A central limit theorem for an omnibus embedding of multiple random
  dot product graphs.
\newblock In \emph{2017 {{IEEE International Conference}} on {{Data Mining
  Workshops}}}, pages 964--967, 2017.

\bibitem[Li and Li(2018)]{li18twosample}
Y.~Li and H.~Li.
\newblock Two-sample test of community memberships of weighted stochastic block
  models.
\newblock \emph{arXiv:1811.12593}, 2018.

\bibitem[Lin et~al.(2021)Lin, Lei, and Roeder]{lin21exponentialfamily}
K.~Z. Lin, J.~Lei, and K.~Roeder.
\newblock Exponential-family embedding with application to cell developmental
  trajectories for single-cell {{RNA-Seq}} data.
\newblock \emph{Journal of the American Statistical Association}, 116\penalty0
  (534):\penalty0 457--470, 2021.

\bibitem[Lovato et~al.(2020)Lovato, Pini, Stamm, and
  Vantini]{lovato20modelfree}
I.~Lovato, A.~Pini, A.~Stamm, and S.~Vantini.
\newblock Model-free two-sample test for network-valued data.
\newblock \emph{Computational Statistics \& Data Analysis}, 144:\penalty0
  106896, 2020.

\bibitem[Loyal and Chen(2021)]{loyal21eigenmodel}
J.~D. Loyal and Y.~Chen.
\newblock An eigenmodel for dynamic multilayer networks.
\newblock \emph{arXiv:2103.12831 [stat]}, 2021.

\bibitem[Lunag{\'o}mez et~al.(2021)Lunag{\'o}mez, Olhede, and
  Wolfe]{lunagomez21modeling}
S.~Lunag{\'o}mez, S.~C. Olhede, and P.~J. Wolfe.
\newblock Modeling network populations via graph distances.
\newblock \emph{Journal of the American Statistical Association}, 116\penalty0
  (536):\penalty0 2023--2040, 2021.

\bibitem[Lv and Liu(2014)]{lv14model}
J.~Lv and J.~S. Liu.
\newblock Model selection principles in misspecified models.
\newblock \emph{Journal of the Royal Statistical Society, Series B},
  76\penalty0 (1):\penalty0 141--167, 2014.

\bibitem[Ma et~al.(2020)Ma, Ma, and Yuan]{ma20universal}
Z.~Ma, Z.~Ma, and H.~Yuan.
\newblock Universal latent space model fitting for large networks with edge
  covariates.
\newblock \emph{The Journal of Machine Learning Research}, 21:\penalty0 67,
  2020.

\bibitem[MacDonald et~al.(2022)MacDonald, Levina, and Zhu]{macdonald22latent}
P.~W. MacDonald, E.~Levina, and J.~Zhu.
\newblock Latent space models for multiplex networks with shared structure.
\newblock \emph{Biometrika}, 109\penalty0 (3):\penalty0 683--706, 2022.

\bibitem[Mazumder et~al.(2010)Mazumder, Hastie, and
  Tibshirani]{mazumder10spectral}
R.~Mazumder, T.~Hastie, and R.~Tibshirani.
\newblock Spectral regularization algorithms for learning large incomplete
  matrices.
\newblock \emph{The Journal of Machine Learning Research}, 11:\penalty0
  2287--2322, 2010.

\bibitem[McCullagh and Nelder(1983)]{mccullagh83generalized}
P.~McCullagh and J.~A. Nelder.
\newblock \emph{Generalized Linear Models}.
\newblock Chapman and Hall, 1983.

\bibitem[Sabanayagam et~al.(2022)Sabanayagam, Vankadara, and
  Ghoshdastidar]{sabanayagam22graphon}
M.~Sabanayagam, L.~C. Vankadara, and D.~Ghoshdastidar.
\newblock Graphon based clustering and testing of networks: {{A}}lgorithms and
  theory.
\newblock In \emph{International Conference on Learning Representations}, 2022.

\bibitem[Sripada et~al.(2020)Sripada, Rutherford, Angstadt, Thompson, Luciana,
  Weigard, Hyde, and Heitzeg]{sripada20prediction}
C.~Sripada, S.~Rutherford, M.~Angstadt, W.~K. Thompson, M.~Luciana, A.~Weigard,
  L.~H. Hyde, and M.~Heitzeg.
\newblock Prediction of neurocognition in youth from resting state {{fMRI}}.
\newblock \emph{Molecular Psychiatry}, 25\penalty0 (12):\penalty0 3413--3421,
  2020.

\bibitem[Tang et~al.(2017)Tang, Athreya, Sussman, Lyzinski, Park, and
  Priebe]{tang17semiparametric}
M.~Tang, A.~Athreya, D.~L. Sussman, V.~Lyzinski, Y.~Park, and C.~E. Priebe.
\newblock A semiparametric two-sample hypothesis testing problem for random
  graphs.
\newblock \emph{Journal of Computational and Graphical Statistics}, 26\penalty0
  (2):\penalty0 344--354, 2017.

\bibitem[Taylor et~al.(1986)Taylor, {Saint-Cyr}, and Lang]{taylor86frontal}
A.~E. Taylor, J.~A. {Saint-Cyr}, and A.~E. Lang.
\newblock Frontal lobe dysfunction in {{Parkinson's}} disease: {{T}}he cortical
  focus of neostriatal outflow.
\newblock \emph{Brain}, 109\penalty0 (5):\penalty0 845--883, 1986.

\bibitem[Tzourio-Mazoyer et~al.(2002)Tzourio-Mazoyer, Landeau, Papathanassiou,
  Crivello, Etard, Delcroix, Mazoyer, and Joliot]{tzourio02automated}
N.~Tzourio-Mazoyer, B.~Landeau, D.~Papathanassiou, F.~Crivello, O.~Etard,
  N.~Delcroix, B.~Mazoyer, and M.~Joliot.
\newblock Automated anatomical labeling of activations in {{SPM}} using a
  macroscopic anatomical parcellation of the {{MNI MRI}} single-subject brain.
\newblock \emph{Neuroimage}, 15\penalty0 (1):\penalty0 273--289, 2002.

\bibitem[Wozniak et~al.(2013)Wozniak, Mueller, Bell, Muetzel, Hoecker, Boys,
  and Lim]{wozniak13global}
J.~R. Wozniak, B.~A. Mueller, C.~J. Bell, R.~L. Muetzel, H.~L. Hoecker, C.~J.
  Boys, and K.~O. Lim.
\newblock Global functional connectivity abnormalities in children with fetal
  alcohol spectrum disorders.
\newblock \emph{Alcohol: Clinical and Experimental Research}, 37\penalty0
  (5):\penalty0 748--756, 2013.

\bibitem[Wu and Hallett(2013)]{wu13cerebellum}
T.~Wu and M.~Hallett.
\newblock The cerebellum in {{Parkinson}}'s disease.
\newblock \emph{Brain}, 136\penalty0 (3):\penalty0 696--709, 2013.

\bibitem[Xia and Li(2022)]{xia22hypothesis}
Y.~Xia and L.~Li.
\newblock Hypothesis testing for network data with power enhancement.
\newblock \emph{Statistica Sinica}, 32:\penalty0 293--321, 2022.

\bibitem[Xia et~al.(2018)Xia, Cai, and Cai]{xia18multiple}
Y.~Xia, T.~Cai, and T.~T. Cai.
\newblock Multiple testing of submatrices of a precision matrix with
  applications to identification of between pathway interactions.
\newblock \emph{Journal of the American Statistical Association}, 113\penalty0
  (521):\penalty0 328--339, 2018.

\bibitem[Xu et~al.(2023)Xu, Yang, Huang, Gururajapathy, Ke, Qiao, Wang, Kumar,
  McGeown, and Kwon]{xu23data}
J.~Xu, Y.~Yang, D.~Huang, S.~S. Gururajapathy, Y.~Ke, M.~Qiao, A.~Wang,
  H.~Kumar, J.~McGeown, and E.~Kwon.
\newblock Data-driven network neuroscience: On data collection and benchmark.
\newblock \emph{Advances in Neural Information Processing Systems},
  36:\penalty0 21841--21856, 2023.

\bibitem[Zhang et~al.(2020)Zhang, Xue, and Zhu]{zhang20flexible}
X.~Zhang, S.~Xue, and J.~Zhu.
\newblock A flexible latent space model for multilayer networks.
\newblock In \emph{Proceedings of the 37th {{International Conference}} on
  {{Machine Learning}}}, pages 11288--11297, 2020.

\end{thebibliography}


\begin{thebibliography}{2}
\providecommand{\natexlab}[1]{#1}
\providecommand{\url}[1]{\texttt{#1}}
\expandafter\ifx\csname urlstyle\endcsname\relax
  \providecommand{\doi}[1]{doi: #1}\else
  \providecommand{\doi}{doi: \begingroup \urlstyle{rm}\Url}\fi

\bibitem[Ghoshdastidar and {von Luxburg}(2018)]{ghoshdastidar18practical}
Debarghya Ghoshdastidar and Ulrike {von Luxburg}.
\newblock Practical methods for graph two-sample testing.
\newblock In \emph{Advances in {{Neural Information Processing Systems}}},
  2018.

\bibitem[Xia and Li(2022)]{xia22hypothesis}
Yin Xia and Lexin Li.
\newblock Hypothesis testing for network data with power enhancement.
\newblock \emph{Statistica Sinica}, 32:\penalty0 293--321, 2022.

\end{thebibliography}


\begin{thebibliography}{12}
\providecommand{\natexlab}[1]{#1}
\providecommand{\url}[1]{\texttt{#1}}
\expandafter\ifx\csname urlstyle\endcsname\relax
  \providecommand{\doi}[1]{doi: #1}\else
  \providecommand{\doi}{doi: \begingroup \urlstyle{rm}\Url}\fi

\bibitem[Bandeira and {van Handel}(2016)]{bandeira16sharp}
A.~S. Bandeira and R.~{van Handel}.
\newblock Sharp nonasymptotic bounds on the norm of random matrices with
  independent entries.
\newblock \emph{The Annals of Probability}, 44\penalty0 (4), 2016.

\bibitem[Breusch(1986)]{breusch86hypothesis}
T.~S. Breusch.
\newblock Hypothesis testing in unidentified models.
\newblock \emph{The Review of Economic Studies}, 53\penalty0 (4):\penalty0
  635--651, 1986.

\bibitem[Bulgren(1971)]{bulgren71representations}
W.~G. Bulgren.
\newblock On representations of the doubly non-central {{{\emph{F}}}}
  distribution.
\newblock \emph{Journal of the American Statistical Association}, 66\penalty0
  (333):\penalty0 184--186, 1971.

\bibitem[Du et~al.(2024)Du, Zhang, and Zhou]{du24optimal}
W.~Du, Y.~Zhang, and W.~Zhou.
\newblock Optimal nonparametric inference on network effects with dependent
  edges.
\newblock \emph{arXiv preprint arXiv:2401.03072}, 2024.

\bibitem[Du and Tang(2021)]{du21hypothesis}
X.~Du and M.~Tang.
\newblock Hypothesis testing for equality of latent positions in random graphs.
\newblock \emph{arXiv:2105.10838 [stat]}, 2021.

\bibitem[Keener(2010)]{keener10theoretical}
R.~W. Keener.
\newblock \emph{Theoretical statistics: topics for a core course}.
\newblock Springer, 2010.

\bibitem[Lv and Liu(2014)]{lv14model}
J.~Lv and J.~S. Liu.
\newblock Model selection principles in misspecified models.
\newblock \emph{Journal of the Royal Statistical Society, Series B},
  76\penalty0 (1):\penalty0 141--167, 2014.

\bibitem[Qi et~al.(2022)Qi, Wang, and Shao]{qi22minimax}
X.~Qi, J.~Wang, and J.~Shao.
\newblock Minimax perturbation bounds of the low-rank matrix under {{Ky Fan}}
  norm.
\newblock \emph{AIMS Mathematics}, 7\penalty0 (5):\penalty0 7595--7605, 2022.

\bibitem[Storey et~al.(2004)Storey, Taylor, and Siegmund]{storey04strong}
J.~D. Storey, J.~E. Taylor, and D.~Siegmund.
\newblock Strong control, conservative point estimation and simultaneous
  conservative consistency of false discovery rates: A unified approach.
\newblock \emph{Journal of the Royal Statistical Society, Series B},
  66\penalty0 (1):\penalty0 187--205, 2004.

\bibitem[Tiku(1967)]{tiku67tables}
M.~L. Tiku.
\newblock Tables of the power of the {{{\emph{F}}}}-test.
\newblock \emph{Journal of the American Statistical Association}, 62\penalty0
  (318):\penalty0 525--539, 1967.

\bibitem[Wedin(1973)]{wedin73perturbation}
P.-{\AA}. Wedin.
\newblock Perturbation theory for pseudo-inverses.
\newblock \emph{BIT}, 13\penalty0 (2):\penalty0 217--232, 1973.

\bibitem[Xia and Li(2022)]{xia22hypothesis}
Y.~Xia and L.~Li.
\newblock Hypothesis testing for network data with power enhancement.
\newblock \emph{Statistica Sinica}, 32:\penalty0 293--321, 2022.

\end{thebibliography}
\end{document}